\documentclass[12pt,english]{article}
\usepackage[T1]{fontenc}
\usepackage[latin9]{inputenc}
\usepackage{geometry}
\usepackage[active]{srcltx}
\usepackage{float}
\usepackage{url}
\usepackage{amsmath}
\usepackage{amssymb}
\usepackage{graphicx}
\usepackage{setspace}
\usepackage{esint}
\makeatletter

\providecommand{\tabularnewline}{\\}

\newtheorem{theorem}{Theorem}

\newtheorem{condition}{Condition}

\newtheorem{corollary}{Corollary}

\newtheorem{example}{Example}

\newtheorem{lemma}{Lemma}

\newtheorem{remark}{Remark}

\newenvironment{proof}[1][Proof]{\textbf{#1.} }{\ \rule{0.5em}{0.5em}}

\makeatother

\usepackage{babel}
\usepackage{comment}

\makeatother

\usepackage{babel}
\begin{document}

\title{Inference for Additive Models in the Presence of Possibly Infinite
Dimensional Nuisance Parameters}

\author{Alessio Sancetta\thanks{E-mail: <asancetta@gmail.com>, URL: <http://sites.google.com/site/wwwsancetta/>.
Address for correspondence: Department of Economics, Royal Holloway
University of London, Egham TW20 0EX, UK }\\
}
\maketitle
\begin{abstract}
A framework for estimation and hypothesis testing of functional restrictions
against general alternatives is proposed. The parameter space is a
reproducing kernel Hilbert space (RKHS). The null hypothesis does
not necessarily define a parametric model. The test allows us to deal
with infinite dimensional nuisance parameters. The methodology is
based on a moment equation similar in spirit to the construction of
the efficient score in semiparametric statistics. The feasible version
of such moment equation requires to consistently estimate projections
in the space of RKHS and it is shown that this is possible using the
proposed approach. This allows us to derive some tractable asymptotic
theory and critical values by fast simulation. Simulation results
show that the finite sample performance of the test is consistent
with the asymptotics and that ignoring the effect of nuisance parameters
highly distorts the size of the tests. 

\end{abstract}

\section{Introduction}

Suppose that you are interested in estimating the number of event
arrivals $Y$ in the next one minute, conditioning on a vector of
covariates $X$ known at the start of the interval. You decide to
minimize the negative log-likelihood for Poisson arrivals with conditional
intensity $\exp\left\{ \mu\left(X\right)\right\} $ for some function
$\mu$. For observation $i$, the negative loglikelihood is proportional
to 
\begin{equation}
\exp\left\{ \mu\left(X_{i}\right)\right\} -Y_{i}\mu\left(X_{i}\right).\label{EQ_poissonLogLik}
\end{equation}
You suppose that $\mu$ lies in some infinite dimensional space. For
example, to avoid the curse of dimensionality, you could choose 
\begin{equation}
\mu\left(X\right):=\sum_{k=1}^{K}f^{\left(k\right)}\left(X^{\left(k\right)}\right)\label{EQ_additiveFunction}
\end{equation}
where $X^{\left(k\right)}$ denotes the $k^{th}$ covariate (the $k^{th}$
element of the $K$-dimensional covariate $X$), and the univariate
functions $f^{\left(k\right)}$ are elements in some possibly infinite
dimensional space. However, you suppose that $f^{\left(1\right)}$
is a linear function. You want to test whether linearity with respect
to the first variable holds against the alternative of a general additive
model. You could also test against the alternative of a general continuous
multivariate function, not necessarily additive. This paper addresses
practical problems such as the above. The paper is not restricted
to this Poisson problem or additive models on real valued variables.

From the example above, we need to (i) estimate $\mu$, which in this
example we chose to be additive with $f^{\left(1\right)}$ linear
under the null; we need to (ii) test this additive restriction, against
a more general non-parametric alternative. Under the null, the remaining
$K-1$ functions in (\ref{EQ_additiveFunction}) are not specified.
Problem (i) is standard, though the actual numerical estimation can
pose problems. Having solved problem (i), solution of problem (ii)
requires to test a non-parametric hypothesis (an additive model with
linear $f^{\left(1\right)}$) with infinite dimensional nuisance parameters
(the remaining unknown $K-1$ functions) against a more general non-parametric
alternative. In this paper, we shall call the restriction under the
null semi-parametric. This does not necessarily mean that the parameter
of interest is finite dimensional, as often the case in the semiparametric
literature. 

Semiparametric inference requires that the infinite dimensional parameter
and the finite dimensional one are orthogonal in the population (e.g.,
Andrews, 1994, eq.(2.12)). In our Poisson motivating example this
is not the case. Even if the restriction is parametric, we do not
need to suppose that the parameter value is known under the null.
This requires us to modify the test statistic in order to achieve
the required orthogonality. Essentially, we project the test statistic
on some space that is orthogonal to the infinite dimensional nuisance
parameter. This is the procedure involved in the construction of the
efficient score in semiparametric statistics. The reader is referred
to van der Vaart (1998) for a review of the basic idea. Here, we are
concerned with functional restrictions and are able to obtain critical
values by fast simulation. In many empirical application, the problem
should possibly allow for dependent observations. The extension to
dependence is not particularly complicated, but will require us to
join together various results in a suitable way. 

Throughout, we shall use the framework of reproducing kernel Hilbert
spaces. The RKHS setup is heavily used in the derivations of the results.
Estimation in these spaces has been studied in depth and is flexible
and intuitive from a theoretical point of view. RKHS also allow us
to consider multivariate problems in a very natural way. In consequence
of these remarks, this paper's main contribution to the literature
is related to testing rather than estimation. Nevertheless, as far
as estimation is concerned, we do provide results that are partially
new. For example, we establish insights regarding the connection between
constrained and penalized estimation together with convergence rates
using possibly dependent observations.

Estimation in RKHS can run into computational issues when the sample
size is large, as it might be the case in the presence of large data
sets. We will address practical computational issues. Estimation of
the model can be carried out via a greedy algorithm, possibly imposing
LASSO kind of constraints under additivity.

Under the null hypothesis, we can find a representation for the limiting
asymptotic distribution which is amenable of fast simulation. In consequence
critical values do not need to be generated using resampling procedures.
While the discussion of the asymptotic validity of the procedure is
involved, the implementation of the test is simple. The Matlab code
for greedy estimation, to perform the test, and compute its critical
values is available from the URL: <https://github.com/asancetta/ARKHS/>.
A set of simulations confirm that the procedure works well, and illustrates
the well known fact that nuisance parameters can considerably distort
the size of a test if not accounted for using our proposed procedure.
The reader can have a preliminary glance at Table \ref{Table_motivatingExample}
in Section \ref{Section_motivation}, and Table \ref{Table_simulationSubset}
in Section \ref{Section_simulations} to see this more vividly.

\subsection{Relation to the Literature}

Estimation in RKHS has been addressed in many places in the literature
(see the monographs of Wahba, 1990, and Steinwart and Christmann,
2008). Inference is usually confined to consistency (e.g., Mendelson,
2002, Christmann and Steinwart, 2007), though there are exceptions
(Hable, 2012, in the frequentist framework). A common restriction
used in the present paper is additivity and estimation in certain
subspaces of additive functions. Estimation of additive models has
been extensively studied by various authors using different techniques
(e.g., Buja et al., 1989, Linton and Nielsen, 1995, Mammen et al.,
1999, Meier et al., 2009, Christmann and Hable, 2012). The last reference
considers estimation in RKHS which allows for a more general concept
of additivity. Here, the assumptions and estimation results are not
overall necessarily comparable to existing results. For example, neither
independence nor the concept of true model are needed. Moreover, we
establish rates of convergence and the link between constrained versus
penalized estimation in RKHS. The two are not always equivalent. 

The problem of parametric inference in the presence of non-orthogonal
nuisance parameters has been addressed by various authors by modification
of the score function or equivalent quantities. Belloni et al. (2017)
provide general results in the context of high dimensional models.
There, the reader can also find the main references in that literature.
The asymptotic distribution usually requires the use of the bootstrap
in order to compute critical values. 

The problem of testing parametric restrictions with finite dimensional
nuisance parameter under the null against general nonparametric alternatives
is well known (H\"{a}rdle and Mammen, 1993), and requires the use
of the bootstrap in order to derive confidence intervals. Fan et al.
(2001) have developed a Generalized Likelihood Ratio test of the null
of parametric or nonparameteric additive restrictions versus general
nonparametric ones. This is based on a Gaussian error model (or parametric
error distribution) for additive regression, and estimation using
smoothing kernels. Fan and Jiang (2005) have extended this approach
to the nonparametric error distribution. The asymptotic distribution
is Chi-square with degrees of freedom equal to some (computable) function
of the the data. Chen et al.(2014) considers the framework of sieve
estimation and derives a likelihood ratio statistic with asymptotic
Chi-square distribution (see also Shen and Shi, 2005).

The approach considered here is complementary to the above references.
It allows the parameter space to be a RKHS of smooth functions. Estimation
in RKHS is well understood and can cater for many circumstances of
interest in applied work. For example, it is possible to view sieve
estimation as estimation in RKHS where the feature space defined by
the kernel increases with the sample size. The testing procedure is
based on a corrected moment condition. Hence, it does not rely on
likelihood estimation. The conditions used are elementary, as they
just require existence of real valued derivatives of the loss function
(in the vein of Christmann and Steinwart, 2007) and mild regularity
conditions on the covariance kernel. We also allow for dependent errors.
The correction is estimated by either ridge regression, or just ordinary
least square using pseudo-inverse. 

For moderate sample sizes (e.g. less than 10,000) estimation in RKHS
does not pose particular challenges and it is trivial for the regression
problem under the square error loss. For large sample sizes, computational
aspects in RKHS have received a lot of attention in the literature
(e.g., Rasmussen and Williams, 2006, Ch.8, Banerjee et al., 2008,
L\'{a}zaro-Gredilla et al., 2010). 

Here we discuss a greedy algorithm, which is simple to implement (e.g.,
Jaggi, 2013, Sancetta, 2016) and, apparently, has not been applied
to the RKHS estimation framework of this paper. 

\subsection{Outline }

The plan for the paper is as follows. Section \ref{Section_inferenceProblem}
reviews some basics of RKHS, defines the problem and model used in
the paper, and describes the implementation of the test. Section \ref{Section_asymptoticAnalysis}
contains the asymptotic analysis of the estimation problem and the
proposed testing procedure in the presence of nuisance parameters.
Section \ref{Section_discussion} provides some additional discussion
of the conditions and the asymptotic analysis. Some details related
to computational implementation can be found in Section \ref{Section_estimationAlgo}.
Section \ref{Section_furtherRemarks} concludes with a finite sample
analysis via simulations. It also discusses the partial extension
to non-smooth loss functions, which exemplify one of the limitations
in the framework of the paper. The proofs, and additional results
are in the Appendices as supplementary material.

\section{The Inference Problem\label{Section_inferenceProblem}}

The explanatory variable $X^{\left(k\right)}$ takes values in $\mathcal{X}$,
a compact subset of a separable Banach space ($k=1,2,...,K$). The
most basic example of $\mathcal{X}$ is $\left[0,1\right]$. The vector
covariate $X=\left(X^{\left(1\right)},...,X^{\left(K\right)}\right)$
takes values in the Cartesian product $\mathcal{X}^{K}$, e.g., $\left[0,1\right]^{K}$.
The dependent variable takes values in $\mathcal{Y}$ usually $\mathbb{R}$.
Let $Z=\left(Y,X\right)$ and this takes values in $\mathcal{Z}=\mathcal{Y}\times\mathcal{X}^{K}$.
If no dependent variable $Y$ can be defined (e.g., unsupervised learning,
or certain likelihood estimators), $Z=X$. Let $P$ be the law of
$Z$, and use linear functional notation, i.e., for any $f:\mathcal{Z}\rightarrow\mathbb{R}$,
$Pf=\int_{\mathcal{Z}}f\left(z\right)dP\left(z\right)$. Let $P_{n}=\frac{1}{n}\sum_{i=1}^{n}\delta_{Z_{i}}$,
where $\delta_{Z_{i}}$ is the point mass at $Z_{i}$, implying that
$P_{n}f=\frac{1}{n}\sum_{i=1}^{n}f\left(Z_{i}\right)$ is the sample
mean of $f\left(Z\right)$. For $p\in\left[1,\infty\right]$, let
$\left|\cdot\right|_{p}$ be the $L_{p}$ norm (w.r.t. the measure
$P$), e.g., for $f:\mathcal{Z}\rightarrow\mathbb{R}$, $\left|f\right|_{p}=\left(P\left|f\right|^{p}\right)^{1/p}$,
with the obvious modification to $\sup$ norm when $p=\infty$.

\subsection{Motivation\label{Section_motivation} }

The problem can be described as follows, though in practice we will
need to add extra regularity conditions. Let $\mathcal{H}^{K}$ be
a vector space of real valued functions on $\mathcal{X}^{K}$, equipped
with a norm $\left|\cdot\right|_{\mathcal{H}^{K}}$. Consider a loss
function $L:\mathcal{Z}\times\mathbb{R}\rightarrow\mathbb{R}$. We
shall be interested in the case where the second argument is $\mu\left(x\right)$:
$L\left(z,\mu\left(x\right)\right)$ with $\mu\in\mathcal{H}^{K}$.
Therefore, to keep notation compact, let $\ell_{\mu}\left(Z\right)=L\left(Z,\mu\left(X\right)\right)$.
For the special case of the square error loss we would have $\ell_{\mu}\left(z\right)=L\left(z,\mu\left(x\right)\right)=\left|y-\mu\left(x\right)\right|^{2}$
($z=\left(y,x\right)$). The use of $\ell_{\mu}$ makes it more natural
to use linear functional notation. The unknown function of interest
is the minimizer $\mu_{0}$ of $P\ell_{\mu}$, and it is assumed to
be in $\mathcal{H}^{K}$. We find an estimator $\mu_{n}=\arg\inf_{\mu}P_{n}\ell_{\mu}$
where the infimum is over certain functions $\mu$ in $\mathcal{H}^{K}$.
The main goal it to test the restriction that $\mu_{0}\in\mathcal{R}_{0}$
for some subspace $\mathcal{R}_{0}$ of $\mathcal{H}^{K}$ (for example
a linear restriction). The restricted estimator in $\mathcal{R}_{0}$
is denoted by $\mu_{0n}$. To test the restriction we can look at
how close 
\begin{equation}
\sqrt{n}P_{n}\partial\ell_{\mu_{0n}}h=\frac{1}{\sqrt{n}}\sum_{i=1}^{n}\partial\ell_{\mu_{0n}}\left(Z_{i}\right)h\left(X_{i}\right)\label{EQ_sampleMomentEquation}
\end{equation}
is to zero for suitable choice of $h\in\mathcal{H}^{K}\setminus\mathcal{R}_{0}$.
Throughout, $\partial^{k}\ell_{\mu}\left(z\right)=\left.\partial^{k}L\left(z,t\right)/\partial t^{k}\right|_{t=\mu\left(x\right)}$
is the $k^{th}$ partial derivative of $L\left(z,t\right)$ with respect
to $t$ and then evaluated at $\mu\left(x\right)$. The validity of
this derivative and other related quantities will be ensured by the
regularity conditions we shall impose. The compact notation on the
left hand side (l.h.s.) of (\ref{EQ_sampleMomentEquation}) shall
be used throughout the paper. If necessary, the reader can refer to
Section \ref{Examples_notation} in the Appendix (supplementary material)
for more explicit expressions when the compact notation is used in
the main text. If the restriction held true, we would expect (\ref{EQ_sampleMomentEquation})
to be mean zero if we used $\mu_{0}$ in place of $\mu_{0n}$. A test
statistic can be constructed from (\ref{EQ_sampleMomentEquation})
as follows: 
\begin{equation}
\frac{1}{R}\sum_{r=1}^{R}\left(\sqrt{n}P_{n}\partial\ell_{\mu_{0n}}h^{\left(r\right)}\right)^{2}\label{EQ_naiveTestStat}
\end{equation}
 where $h^{\left(r\right)}\in\mathcal{H}^{K}\setminus\mathcal{R}_{0}$,
$r=1,2,...,R$. 

If $\mathcal{R}_{0}$ is finite dimensional, or $\mu_{0n}$ is orthogonal
to the functions $h\in\mathcal{H}^{K}\setminus\mathcal{R}_{0}$ (e.g.,
Andrews, 1994, eq. 2.12), the above display is - to first order -
equal in distribution to $\sqrt{n}P_{n}\ell_{\mu_{0}}h$, under regularity
conditions. However, unless the sample size is relatively large, this
approximation may not be good. In fact, supposing stochastic equicontinuity
and the null that $\sqrt{n}P\partial\ell_{\mu_{0}}h=0$, it can be
shown that (e.g., Theorem 3.3.1 in van der Vaart and Wellner, 2000),
\[
\sqrt{n}P_{n}\partial\ell_{\mu_{0n}}h=\sqrt{n}P_{n}\partial\ell_{\mu_{0}}h+\sqrt{n}P\partial^{2}\ell_{\mu_{0}}\left(\mu_{0n}-\mu_{0}\right)h+o_{p}\left(1\right).
\]
The orthogonality condition in Andrews (1994, eq., 2.12) guarantees
that the second term on the right hand side (r.h.s.) is zero (Andrews,
1994, eq.2.8, assuming Fr\'{e}chet differentiability). Hence, we
aim to find/construct functions $h\in\mathcal{H}^{K}\setminus\mathcal{R}_{0}$
such that the second term on the r.h.s. is zero. In fact this term
can severely distort the asymptotic behaviour of $\sqrt{n}P_{n}\partial\ell_{\mu_{0n}}h$. 

An example is given in Table \ref{Table_motivatingExample} which
is an excerpt from the simulation results in Section \ref{Section_simulations}.
Here, the true model is a linear model with 3 variables plus Gaussian
noise with signal to noise ratio equal to one. We call this model
Lin3. We use a sample of $n\in\left\{ 100,1000\right\} $ observations
with $K=10$ variables. Under the null hypothesis, only the first
three variables enter the linear model, against an alternative that
all $K=10$ variables enter the true model in an additive nonlinear
form. The subspace of these three linear functions is $\mathcal{R}_{0}$
while the full model is $\mathcal{H}^{K}$. The test functions $h$
are restricted to polynomials with no linear term. Details can be
found in Section \ref{Section_simulations}. The nuisance parameters
are the three estimated linear functions, which are low dimensional.
It is plausible that the estimation of the three linear functions
(i.e. $\mu_{0n}$) should not affect the asymptotic distribution of
(\ref{EQ_naiveTestStat}). When the variables are uncorrelated, this
is clearly the case as confirmed by the 5\% size of the test in Table
\ref{Table_motivatingExample}. It does not matter whether we use
instruments $h\in\mathcal{H}^{K}\setminus\mathcal{R}_{0}$ that are
orthogonal to the linear functions or not. However, as soon as the
variables become correlated, Table \ref{Table_motivatingExample}
shows that the asymptotic distribution can be distorted. This happens
even in such a simple finite dimensional problem. Nevertheless, the
test that uses instruments that are made orthogonal to functions in
$\mathcal{R}_{0}$ is not affected. The paper will discuss the empirical
procedure used to construct such instruments and will study its properties
via asymptotic analysis and simulations. 

\begin{table}
\caption{Frequency of rejections. Results from 1000 simulations when the number
of covariates $K=10$ and the true model is Lin3 (only the first three
variables enter the model and they do so in a linear way). The column
No $\Pi$ denotes test results using instruments in $\mathcal{H}^{K}\setminus\mathcal{R}_{0}$.
The column $\Pi$ denotes test results using instruments in $\mathcal{H}^{K}\setminus\mathcal{R}_{0}$
that have been made orthogonal to the functions in $\mathcal{R}_{0}$
using the empirical procedure discussed in this paper. The signal
to noise ratio is denoted by $\sigma_{\mu/\varepsilon}^{2}$, while
all the variables have equal pairwise correlation equal to $\rho$.
The column Size denotes the theoretical size of the test. A value
in columns No $\Pi$ and $\Pi$ smaller than $0.05$ indicates that
the test procedure rejects less often than it should. }
\label{Table_motivatingExample}%
\begin{tabular}{ccccccc}
 &  &  & \multicolumn{4}{c}{Lin3}\tabularnewline
 &  &  & \multicolumn{2}{c}{$n$=100} & \multicolumn{2}{c}{$n$=1000}\tabularnewline
$\rho$ & $\sigma_{\mu/\varepsilon}^{2}$ & Size  & No $\Pi$ & $\Pi$ & No $\Pi$ & $\Pi$\tabularnewline
0 & 1 & 0.05 & 0.03 & 0.06 & 0.05 & 0.05\tabularnewline
0.75 & 1 & 0.05 & 0.02 & 0.05 & 0.03 & 0.05\tabularnewline
\end{tabular}
\end{table}

The situation gets really worse with other simulation designs that
can be encountered in applications and details are given in Section
\ref{Section_simulations}. More generally, $\mathcal{R}_{0}$ can
be a high dimensional subspace of $\mathcal{H}^{K}$ or even an infinite
dimensional one, e.g. the space of additive functions when $\mathcal{H}^{K}$
does not impose this additive restriction. In this case, it is unlikely
that functions in $\mathcal{H}^{K}\setminus\mathcal{R}_{0}$ are orthogonal
to functions in $\mathcal{R}_{0}$ and the distortion due to the nuisance
parameters will be larger than what is shown in Table \ref{Table_motivatingExample}. 

Here, orthogonal functions $h\in\mathcal{H}^{K}\setminus\mathcal{R}_{0}$
are constructed to asympotically satisfy 
\begin{equation}
P\partial^{2}\ell_{\mu_{0}}\nu h=0\label{EQ_orthogonalityCondition}
\end{equation}
for any $\nu\in\mathcal{R}_{0}$ when $\mu_{0}$ is inside $\mathcal{R}_{0}$.
The above display will allow us to carry out inferential procedures
as in cases previously considered in the literature. The challenge
is that the set of such orthogonal functions $h\in\mathcal{H}^{K}\setminus\mathcal{R}_{0}$
needs to be estimated. It is not clear before hand that estimation
leads to the same asymptotic distribution as if this set were known.
We show that this is the case. Suppose that $\left\{ \hat{h}^{\left(r\right)}:r=1,2,...,R\right\} $
is a set of such estimated orthogonal functions using the method to
be spelled out in this paper. The test statistic is 
\begin{equation}
\hat{S}_{n}=\frac{1}{R}\sum_{r=1}^{R}\left(\sqrt{n}P_{n}\partial\ell_{\mu_{0n}}\hat{h}^{\left(r\right)}\right)^{2}.\label{EQ_intro2testStat}
\end{equation}
We show that its asymptotic distribution can be easily simulated. 

Next, some basics of RKHS are reviewed and some notation is fixed.
Restrictions for functions in $\mathcal{H}^{K}$ are discussed and
finally the estimation problems is defined.

\subsection{Additional Notation and Basic Facts about Reproducing Kernel Hilbert
Spaces\label{Section_reviewRKHS}}

Recall that a RKHS $\mathcal{H}$ on some set $\mathcal{X}$ is a
Hilbert space where the evaluation functionals are bounded. A RKHS
of bounded functions is uniquely generated by a centered Gaussian
measure with covariance $C$ (e.g., Li and Linde, 1999) and $C$ is
usually called the (reproducing) kernel of $\mathcal{H}$. We consider
covariance functions with representation
\begin{equation}
C\left(s,t\right)=\sum_{v=1}^{\infty}\lambda_{v}^{2}\varphi_{v}\left(s\right)\varphi_{v}\left(t\right),\label{EQ_covSeriesRepresentation}
\end{equation}
for linearly independent functions $\varphi_{v}:\mathcal{X}\rightarrow\mathbb{R}$
and coefficients $\lambda_{v}$ such that $\sum_{v=1}^{\infty}\lambda_{v}^{2}\varphi_{v}^{2}\left(s\right)<\infty$.
Here, linear independent means that if there is a sequence of real
numbers $\left(f_{v}\right)_{v\geq1}$ such that $\sum_{v=1}^{\infty}f_{v}^{2}/\lambda_{v}^{2}<\infty$
and $\sum_{v=1}^{\infty}f_{v}\varphi_{v}\left(s\right)=0$ for all
$s\in\mathcal{X}$, then $f_{v}=0$ for all $v\geq1$. The coefficients
$\lambda_{v}^{2}$ would be the eigenvalues of (\ref{EQ_covSeriesRepresentation})
if the functions $\varphi_{v}$ were orthonormal, but this is not
implied by the above definition of linear independence. The RKHS $\mathcal{H}$
is the completion of the set of functions representable as $f\left(x\right)=\sum_{v=1}^{\infty}f_{v}\varphi_{v}\left(x\right)$
for real valued coefficient $f_{v}$ such that $\sum_{v=1}^{\infty}f_{v}^{2}/\lambda_{v}^{2}<\infty$.
Equivalently, $f\left(x\right)=\sum_{j=1}^{\infty}\alpha_{j}C\left(s_{j},x\right)$,
for coefficients $s_{j}$ in $\mathcal{X}$ and real valued coefficients
$\alpha_{j}$ satisfying $\sum_{j=1}^{\infty}\alpha_{i}\alpha_{j}C\left(s_{i},s_{j}\right)<\infty$.
Moreover, for $C$ in (\ref{EQ_covSeriesRepresentation}),
\begin{equation}
\sum_{j=1}^{\infty}\alpha_{j}C\left(s_{j},x\right)=\sum_{v=1}^{\infty}\left(\sum_{j=1}^{\infty}\alpha_{j}\lambda_{v}^{2}\varphi_{v}\left(s_{j}\right)\right)\varphi_{v}\left(x\right)=\sum_{v=1}^{\infty}f_{v}\varphi_{v}\left(x\right)\label{EQ_functionRepresentationEquality}
\end{equation}
by obvious definition of the coefficients $f_{v}$. The change of
summation is possible by the aforementioned restrictions on the coefficients
$\lambda_{v}$ and functions $\varphi_{v}$. The inner product in
$\mathcal{H}$ is denoted by $\left\langle \cdot,\cdot\right\rangle _{\mathcal{H}}$
and satisfies $f\left(x\right)=\left\langle f,C\left(x,\cdot\right)\right\rangle _{\mathcal{H}}$.
This implies the reproducing kernel property $C\left(s,t\right)=\left\langle C\left(s,\cdot\right),C\left(t,\cdot\right)\right\rangle _{\mathcal{H}}$.
Therefore, the square of the RKHS norm is defined in the two following
equivalent ways 
\begin{equation}
\left|f\right|_{\mathcal{H}}^{2}=\sum_{v=1}^{\infty}\frac{f_{v}^{2}}{\lambda_{v}^{2}}=\sum_{i,j=1}^{\infty}\alpha_{i}\alpha_{j}C\left(s_{i},s_{j}\right)\label{EQ_RHKS_normsEquivalences}
\end{equation}
Throughout, the unit ball of $\mathcal{H}$ will be denoted by $\mathcal{H}\left(1\right):=\left\{ f\in\mathcal{H}:\left|f\right|_{\mathcal{H}}\leq1\right\} $.

The additive RKHS is generated by the Gaussian measure with covariance
function $C_{\mathcal{H}^{K}}\left(s,t\right)=\sum_{k=1}^{K}C^{\left(k\right)}\left(s^{\left(k\right)},t^{\left(k\right)}\right)$,
where $C^{\left(k\right)}\left(s^{\left(k\right)},t^{\left(k\right)}\right)$
is a covariance function on $\mathcal{X}\times\mathcal{X}$ (as $C$
in (\ref{EQ_covSeriesRepresentation})) and $s^{\left(k\right)}$
is the $k^{th}$ element in $s\in\mathcal{X}^{K}$. The RKHS of additive
functions is denoted by $\mathcal{H}^{K}$, which is the set of functions
as in (\ref{EQ_additiveFunction}) such that $f^{\left(k\right)}\in\mathcal{H}$
and $\sum_{k=1}^{K}\left|f^{\left(k\right)}\right|_{\mathcal{H}}^{2}<\infty$.
For such functions, the inner product is $\left\langle f,g\right\rangle _{\mathcal{H}^{K}}=\sum_{k=1}^{K}\left\langle f^{\left(k\right)},g^{\left(k\right)}\right\rangle _{\mathcal{H}}$,
where - for ease of notation - the individual RKHS are supposed to
be the same. However, in some circumstances, it can be necessary to
make the distinction between the spaces (see Example \ref{Example_H0AdditiveH1General}
in Section \ref{Section_testingFunctionalRestrictions}). The norm
$\left|\cdot\right|_{\mathcal{H}^{K}}$ on $\mathcal{H}^{K}$ is the
one induced by the inner product. 

Within this scenario, the space $\mathcal{H}^{K}$ restricts functions
to be additive, where these additive functions in $\mathcal{H}$ can
be multivariate functions. 

\begin{example}\label{Example_d-Dimensional}Suppose that $K=1$
and $\mathcal{X}=\left[0,1\right]^{d}$ ($d>1$) (only one additive
function, which is multivariate). Let $C\left(s,t\right)=\exp\left\{ -a\sum_{j}\left|s_{j}-t_{j}\right|^{2}\right\} $
where $s_{j}$ is the $j^{th}$ element in $s\in\left[0,1\right]^{d}$,
and $a>0$. Then, the RKHS $\mathcal{H}$ is dense in the space of
continuous bounded functions on $\left[0,1\right]^{d}$ (e.g., Christmann
and Steinwart, 2007). A (kernel) $C$ with such property is called
universal.\end{example}

The framework also covers the case of functional data because $\mathcal{X}$
is a compact subset of a Banach space (e.g., Bosq, 2000). Most problems
of interest where the unknown parameter $\mu$ is a smooth function
are covered by the current scenario. 

\subsection{The Estimation Problem\label{Section_theGoal}}

Estimation will be considered for models in $\mathcal{H}^{K}\left(B\right):=\left\{ f\in\mathcal{H}^{K}:\left|f\right|_{\mathcal{H}^{K}}\leq B\right\} $,
where $B<\infty$ is a fixed constant. The goal is to find 
\begin{equation}
\mu_{n}=\arg\inf_{\mu\in\mathcal{H}^{K}\left(B\right)}P_{n}\ell_{\mu},\label{EQ_theEstimator}
\end{equation}
i.e. the minimizer with respect to $\mu\in\mathcal{H}^{K}\left(B\right)$
of the loss function $P_{n}\ell_{\mu}$. 

\begin{example}\label{Example_sqaureErrorLoss}Let $\ell_{\mu}\left(z\right)=\left|y-\mu\left(x\right)\right|^{2}$
so that 
\[
P_{n}\ell_{\mu}=\frac{1}{n}\sum_{i=1}^{n}\ell_{\mu}\left(Z_{i}\right)=\frac{1}{n}\sum_{i=1}^{n}\left|Y_{i}-\mu\left(X_{i}\right)\right|^{2}.
\]
By duality, we can also use $P_{n}\ell_{\mu}+\rho_{B,n}\left|\mu\right|_{\mathcal{H}^{K}}^{2}$
with sample dependent Lagrange multiplier $\rho_{B,n}$ such that
$\left|\mu\right|_{\mathcal{H}^{K}}\leq B$.\end{example}

For the square error loss the solution is just a ridge regression
estimator with (random) ridge parameter $\rho_{B,n}\geq0$. Interest
is not restricted to least square problems.

\begin{example}Consider the negative log-likelihood where $Y$ is
a duration, and $\mathbb{E}\left[Y|X\right]=\exp\left\{ \mu\left(X\right)\right\} $
is the hazard function. Then, $\ell_{\mu}\left(z\right)=y\exp\left\{ \mu\left(x\right)\right\} -\mu\left(x\right)$
so that $P_{n}\ell_{\mu}=\frac{1}{n}\sum_{i=1}^{n}Y_{i}\exp\left\{ \mu\left(X_{i}\right)\right\} -\mu\left(X_{i}\right)$.\end{example}

Even though the user might consider likelihood estimation, there is
no concept of ``true model'' in this paper. The target is the population
estimate 
\begin{equation}
\mu_{0}=\arg\inf_{\mu\in\mathcal{H}^{K}\left(B\right)}P\ell_{\mu}.\label{EQ_targetFunction}
\end{equation}
We shall show that this minimizer always exists and is unique under
regularity conditions on the loss because $\mathcal{H}^{K}\left(B\right)$
is closed. 

Theorem 1 in Sch\"{o}lkopf et al. (2001) says that the solution to
the penalized problem takes the form $\mu_{n}\left(x\right)=\sum_{i=1}^{n}\alpha_{i}C\left(X_{i},x\right)$
for real valued coefficients $\alpha_{i}$. Hence, even if the parameter
space where the estimator lies is infinite dimensional, $\mu_{n}$
is not. This fact will be used without further mention in the matrix
implementation of the testing problem.

\subsection{The Testing Problem \label{Section_TestingProblem}}

Inference needs to be conducted on the estimator in (\ref{EQ_theEstimator}).
To this end, consider inference on functional restrictions possibly
allowing $\mu$ not to be fully specified under the null. Within this
framework, tests based on the moment equation $P_{n}\partial\ell_{\mu}h$
for suitable test functions $h$ are natural (recall (\ref{EQ_intro2testStat})).
Let $\mathcal{R}_{0}\subset\mathcal{H}^{K}$ be the RKHS with kernel
$C_{\mathcal{R}_{0}}$. Suppose that we can write $C_{\mathcal{H}^{K}}=C_{\mathcal{R}_{0}}+C_{\mathcal{R}_{1}}$,
where $C_{\mathcal{R}_{1}}$ is some suitable covariance function.
Under the null hypothesis we suppose that $\mu_{0}\in\mathcal{R}_{0}$
($\mu_{0}$ as in (\ref{EQ_targetFunction})). Under the alternative,
$\mu_{0}\notin\mathcal{R}_{0}$. Define 
\begin{equation}
\mu_{n0}:=\arg\inf_{\mu\in\mathcal{R}_{0}\left(B\right)}P_{n}\ell_{\mu},\label{EQ_restrictedEstimator}
\end{equation}
where $\mathcal{R}_{0}\left(B\right)=\mathcal{R}_{0}\cap\mathcal{H}^{K}\left(B\right)$.
This is the estimator under the null hypothesis. For this estimation,
we use the kernel $C_{\mathcal{R}_{0}}$. The goal is to consider
the quantity in (\ref{EQ_sampleMomentEquation}) with suitable $h\in\mathcal{R}_{1}$. 

\subsubsection{Matrix Implementation\label{Section_matrixImplementation}}

We show how to construct the statistic in (\ref{EQ_intro2testStat})
using matrix notation. Consider the regression problem under the square
error loss: nonlinear least squares. Let $\mathbf{C}$ be the $n\times n$
matrix with $\left(i,j\right)$ entry equal to $C_{\mathcal{H}^{K}}\left(X_{i},X_{j}\right)$,
$\mathbf{y}$ the $n\times1$ vector with $i^{th}$ entry equal to
$Y_{i}$. The penalized estimator is the $n\times1$ vector $\mathbf{a}:=\left(\mathbf{C}+\rho\mathbf{I}\right)^{-1}\mathbf{y}$.
Here, $\rho$ can be chosen such that $\mathbf{a}^{T}\mathbf{C}\mathbf{a}\leq B^{2}$
so that the constraint is satisfied: $\mu_{n}\left(\cdot\right)=\sum_{i=1}^{n}a_{i}C_{\mathcal{H}^{K}}\left(X_{i},\cdot\right)$
is in $\mathcal{H}^{K}\left(B\right)$; here $a_{i}$ is the $i^{th}$
entry in $\mathbf{a}$ and the superscript $^{T}$ is used for transposition.
For other problems the solution is still linear, but the coefficients
usually do not have a closed form. For the regression problem under
the square error loss, if the constraint $\left\{ \mu\in\mathcal{H}^{K}\left(B\right)\right\} $
is binding, the $\rho$ that satisfies the constraint is given by
the solution of 
\[
\sum_{i=1}^{n}\left(\mathbf{y}^{T}\mathbf{Q}_{i}\right)^{2}\frac{\kappa_{i}}{\kappa_{i}+\rho}=B^{2}
\]
where $\mathbf{Q}_{i}$ is the $i^{th}$ eigenvector of $\mathbf{C}$
and here $\kappa_{i}$ is the corresponding eigenvalue. 

The restricted estimator has the same solution with $\mathbf{C}$
replaced by $\mathbf{C}_{0}$ which is the matrix with $\left(i,j\right)$
entry $C_{\mathcal{R}_{0}}\left(X_{i},X_{j}\right)$. For the square
error loss, let $\mathbf{e}_{0}=\mathbf{y}-\mathbf{C}_{0}\mathbf{a}_{0}$
be the vector or residuals under the null. (For other problems, $\mathbf{e}_{0}$
is the vector of generalized residuals, i.e. the $i^{th}$ entry in
$\mathbf{e}_{0}$ is $\partial\ell_{\mu_{0,n}}\left(Z_{i}\right)$.)
Under the alternative we have the covariance kernel $C_{\mathcal{R}_{1}}$.
Denote by $\mathbf{C}_{1}$ the matrix with $\left(i,j\right)$ entry
$C_{\mathcal{R}_{1}}\left(X_{i},X_{j}\right)$. Let $\mathbf{S}$
be the diagonal matrix with $\left(i,i\right)$ diagonal entry equal
to $\partial^{2}\ell_{\mu_{0n}}\left(Z_{i}\right)$. In our case,
this entry can be taken to be one, as the second derivative of the
square error loss is a constant. However, the next step is the same
regardless of the loss function, as we only need to project the functions
in $\mathcal{R}_{1}$ onto $\mathcal{R}_{0}$ and consider the orthogonal
part. This ensures that the sample version of the orthogonality condition
(\ref{EQ_orthogonalityCondition}) is satisfied. We regress each column
of $\mathbf{C}_{1}$ on the columns of $\mathbf{C}_{0}$. We denote
by $\mathbf{C}_{1}^{\left(r\right)}$ the $r^{th}$ column in $\mathbf{C}_{1}$.
We approximately project $\mathbf{C}_{1}^{\left(r\right)}$ onto the
column space spanned by $\mathbf{C}_{0}$ minimizing the loss function
\[
\left(\mathbf{C}_{1}^{\left(r\right)}-\mathbf{C}_{0}\mathbf{b}^{\left(r\right)}\right)^{T}\mathbf{S}\left(\mathbf{C}_{1}^{\left(r\right)}-\mathbf{C}_{0}\mathbf{b}^{\left(r\right)}\right)+\rho\left(\mathbf{b}^{\left(r\right)}\right)^{T}\mathbf{C}_{0}\mathbf{b}^{\left(r\right)}.
\]
Here $\rho$ is chosen to go to zero with the sample size (Theorem
\ref{Theorem_projectionScoreTest} and Corollary \ref{Corollary_feasibleStatKnownHessian}).
In applications, we may just use a subset of $R$ columns from $\mathbf{C}_{1}$
and to avoid notational trivialities, say the first $R$. The solution
for all $r=1,2,...,R$ is 
\[
\mathbf{b}^{\left(r\right)}=\left(\mathbf{C}_{0}+\rho\mathbf{S}^{-1}\right)^{-1}\mathbf{C}_{1}^{\left(r\right)},
\]
and can be verified substituting it in the first order conditions.
Let the residual vector from this regression be $\mathbf{e}_{1}^{\left(r\right)}$.
In sample, this is orthogonal to the column space of $\mathbf{C}_{0}$
when $\rho=0$. We define the $r^{th}$ instruments by $\hat{\mathbf{h}}^{\left(r\right)}=\mathbf{e}_{1}^{\left(r\right)}$.
The test statistic is $\hat{S}_{n}=\sum_{r=1}^{R}\left(\mathbf{e}_{0}^{T}\mathbf{\hat{\mathbf{h}}}^{\left(r\right)}\right)^{2}/R$.
Under regularity conditions, if the true parameter $\mu_{0}$ lies
inside $\mathcal{R}_{0}\cap\mathcal{H}^{K}\left(B\right)$, the $R\times1$
vector $\mathbf{s}=\left(\mathbf{e}_{0}^{T}\hat{\mathbf{h}}^{\left(1\right)},\mathbf{e}_{0}^{T}\hat{\mathbf{h}}^{\left(2\right)},...,\mathbf{e}_{0}^{T}\hat{\mathbf{h}}^{\left(R\right)}\right)^{T}$
is asymptotically Gaussian for any $R$ and its covariance matrix
is consistently estimated by $\left(n^{-1}\mathbf{e}_{0}^{T}\mathbf{e}_{0}\right)\sum_{k,l=1}^{R}\left[n^{-1}\left(\mathbf{\hat{\mathbf{h}}}^{\left(k\right)}\right)^{T}\mathbf{\hat{\mathbf{h}}}^{\left(l\right)}\right]$.
The distribution of $\hat{S}_{n}$ can be simulated from the process
$\sum_{l=1}^{R}\omega_{n,l}N_{l}^{2}$, where the random variables
$N_{l}$ are i.i.d. standard normal and the real valued coefficients
$\omega_{n,l}$ are eigenvalues of the estimated covariance matrix. 

\paragraph{Operational remarks.}
\begin{enumerate}
\item If $C_{\mathcal{R}_{1}}$ is not explicitly given, we can set $C_{\mathcal{R}_{1}}=C_{\mathcal{H}^{K}}$
in the projection step. 
\item Instead of $\mathbf{C}_{1}$ $n\times n$ we can use a subset of the
columns of $\mathbf{C}_{1}$, e.g. $R<n$ columns. Each column is
an instrument. 
\item The $r^{th}$column of $\mathbf{C}_{1}$ can be replaced by an $n\times1$
vector with $i^{th}$ entry $C_{\mathcal{R}_{1}}\left(X_{i},z_{r}\right)$
where $z_{r}$ is an arbitrary element in $\mathcal{X}^{K}$.
\item To keep the test functions homogeneous, we can set the $r^{th}$ column
of $\mathbf{C}_{1}$ to have $i^{th}$ entry equal to $C_{\mathcal{H}^{K}}\left(X_{i},z_{r}\right)/\sqrt{C_{\mathcal{H}^{K}}\left(z_{r},z_{r}\right)}$;
note that $h^{\left(r\right)}\left(\cdot\right):=C_{\mathcal{H}^{K}}\left(\cdot,z_{r}\right)/\sqrt{C_{\mathcal{H}^{K}}\left(z_{r},z_{r}\right)}$
satisfies $\left|h^{\left(r\right)}\right|_{\mathcal{H}^{K}}=1$ by
the reproducing kernel property.
\item When the series expansion (\ref{EQ_covSeriesRepresentation}) for
the covariance is known, we can use the elements in the expansion.
For example, suppose $\mathcal{V}_{0}$ and $\mathcal{V}_{1}$ are
mutually exclusive subsets of the natural numbers such that $C_{\mathcal{R}_{j}}\left(s,t\right)=\sum_{v\in\mathcal{V}_{j}}\lambda_{v}\varphi_{v}\left(s\right)\varphi_{v}\left(t\right)$
for $j\in\left\{ 0,1\right\} $. We can directly ``project'' the
elements in $\left\{ \lambda_{v}^{1/2}\varphi_{v}:v\in\mathcal{V}_{1}\right\} $
onto the linear span of $\left\{ \lambda_{v}^{1/2}\varphi_{v}:v\in\mathcal{V}_{0}\right\} $
by ridge regression with penalty $\rho$. For $\mathcal{V}_{j}$ of
finite but increasing cardinality, the procedure covers sieve estimators
with restricted coefficients. Note that $h^{\left(r\right)}=\lambda_{r}^{1/2}\varphi_{r}$
satisfies $\left|h^{\left(r\right)}\right|_{\mathcal{H}^{K}}=1$ for
$r\in\mathcal{V}_{1}$.
\end{enumerate}

\paragraph{Additional remarks.}

The procedure can be seen as a J-Test where the instruments are given
by the $\hat{\mathbf{h}}^{\left(r\right)}$'s. Given that the covariance
matrix of the vector $\mathbf{s}$ can be high dimensional (many instruments
for large $R$) we work directly with the unstandardized statistic.
This is common in some high dimensional problems, as it is the case
in functional data analysis. 

We could replace $\hat{S}_{n}$ with $\max_{r\leq R}\mathbf{e}_{0}^{T}\hat{\mathbf{h}}^{\left(r\right)}$.
The maximum of correlated Gaussian random variables can be simulated
or approximated but it might be operationally challenging (Hartigan,
2014, Theorem 3.4). 

The rest of the paper provides details and justification for the estimation
and testing procedure. The theoretical justification beyond simple
heuristics is technically involved. Section \ref{Section_simulations}
(Tables \ref{Table_simulationSubset} and \ref{Table_simulationSubsetCov})
will show that failing to use the projection procedure discussed in
this paper leads to poor results. Additional details can be found
in Appendix 2 (supplementary material).

\section{Asymptotic Analysis\label{Section_asymptoticAnalysis}}

\subsection{Conditions for Basic Analysis\label{Section_Conditions}}

Throughout the paper, $\lesssim$ means that the l.h.s. is bounded
by an absolute constant times the r.h.s.. 

\begin{condition}\label{Condition_kernel}The set $\mathcal{H}$
is a RKHS on a compact subset of a separable Banach space $\mathcal{X}$,
with continuous uniformly bounded kernel $C$ admitting an expansion
(\ref{EQ_covSeriesRepresentation}), where $\lambda_{v}^{2}\lesssim v^{-2\eta}$
with exponent $\eta>1$ and with linearly independent continuous uniformly
bounded functions $\varphi_{v}:\mathcal{X}\rightarrow\mathbb{R}$.
If each additive component has a different covariance kernel, the
condition is meant to apply to each of them individually.\end{condition}

Attention is restricted to loss functions satisfying the following,
though generalizations will be considered in Section \ref{Section_lossExtensions}.
Recall the loss $L\left(z,t\right)$ from Section \ref{Section_motivation}.
Let $\bar{B}:=c_{K}B$ where $c_{K}:=\max_{s\in\mathcal{X}^{K}}\sqrt{C_{\mathcal{H}^{K}}\left(s,s\right)}$.
Define $\Delta_{k}\left(z\right):=\max_{\left|t\right|\leq\bar{B}}\left|\partial^{k}L\left(z,t\right)/\partial t^{k}\right|$
for $k=0,1,2,..$. 

\begin{condition}\label{Condition_Loss} The loss $L\left(z,t\right)$
is non-negative, twice continuously differentiable for real $t$ in
an open set containing $\left[-\bar{B},\bar{B}\right]$, and $\inf_{z,t}d^{2}L\left(z,t\right)/dt^{2}>0$
for $z\in\mathcal{Z}$ and $t\in\left[-\bar{B},\bar{B}\right]$. Moreover,
$P\left(\Delta_{0}+\Delta_{1}^{p}+\Delta_{2}^{p}\right)<\infty$ for
some $p>2$. \end{condition} 

The data are allowed to be weakly dependent, but restricted to uniform
regularity.

\begin{condition}\label{Condition_Dependence}The sequence $\left(Z_{i}\right)_{i\in\mathbb{Z}}$
($Z_{i}=\left(Y_{i},X_{i}\right)$) is stationary, with beta mixing
coefficient $\beta\left(i\right)\lesssim\left(1+i\right)^{-\beta}$
for $\beta>p/\left(p-2\right)$, where $p$ is as in Condition \ref{Condition_Loss}.\end{condition}

Remarks on the conditions can be found in Section \ref{Section_RemarksConditions}. 

\subsection{Basic Results}

This section shows the consistency and some basic convergence in distribution
of the estimator. These results can be viewed as a review, except
for the fact that we allow for dependent random variables. We also
provide details regarding the relation between constrained, and penalized
estimators and convergence rates. The usual penalized estimator is
defined as 
\begin{equation}
\mu_{n,\rho}=\arg\inf_{\mu\in\mathcal{H}^{K}}P_{n}\ell_{\mu}+\rho\left|\mu\right|_{\mathcal{H}^{K}}^{2}\label{EQ_penalizedEstimator.}
\end{equation}
for $\rho\geq0$. As mentioned in Example \ref{Example_sqaureErrorLoss},
suitable choice of $\rho$ leads to the constrained estimator. Throughout,
${\rm int}\left(\mathcal{H}^{K}\left(B\right)\right)$ will denote
the interior of $\mathcal{H}^{K}\left(B\right)$. 

\begin{theorem}\label{Theorem_consistency}Suppose that Conditions
\ref{Condition_kernel}, \ref{Condition_Loss}, and \ref{Condition_Dependence}
hold. The population minimizer in (\ref{EQ_targetFunction}) is unique
up to an equivalence class in $L_{2}$. 
\begin{enumerate}
\item There is a random $\rho=\rho_{B,n}$ such that $\rho=O_{p}\left(n^{-1/2}\right)$,
$\mu_{n,\rho}=\mu_{n}$ and if $\mu_{0}\in\mathcal{H}^{K}\left(B\right)$,
$\left|\mu_{n}-\mu_{0}\right|_{\infty}\rightarrow0$ in probability
where $\mu_{n}$ and $\mu_{n,\rho}$ are as in (\ref{EQ_theEstimator})
and (\ref{EQ_penalizedEstimator.}).
\item Consider (\ref{EQ_theEstimator}). We also have that $\left|\mu_{n}-\mu_{0}\right|_{2}=O_{p}\left(n^{-\left(2\eta-1\right)/\left(4\eta\right)}\right)$,
and if $\mathcal{H}^{K}$ is finite dimensional the r.h.s. is $O_{p}\left(n^{-1/2}\right)$.
\item Consider possibly random $\rho=\rho_{n}$ such that $\rho\rightarrow0$
and $\rho n^{1/2}\rightarrow\infty$ in probability. Suppose that
there is a finite $B$ such that $\mu_{0}\in{\rm int}\left(\mathcal{H}^{K}\left(B\right)\right)$.
Then, $\left|\mu_{n,\rho}-\mu_{0}\right|_{\mathcal{H}^{K}}\rightarrow0$
in probability, and in consequence $\left|\mu_{n,\rho}\right|_{\mathcal{H}^{K}}<B$
with probability going to one. 
\item If $\mathcal{H}^{K}$ is infinite dimensional, there is a $\rho=\rho_{n}$
such that $\rho\rightarrow0$, $\rho n^{1/2}\nrightarrow\infty$,
and $\left|\mu_{n,\rho}-\mu_{0}\right|_{\infty}\rightarrow0$ in probability,
but $\left|\mu_{n,\rho}-\mu_{0}\right|_{\mathcal{H}^{K}}$ does not
converge to zero in probability. 
\end{enumerate}
All the above consistency statements also hold if $\mu_{n}$ and $\mu_{n,\rho}$
in (\ref{EQ_theEstimator}) and (\ref{EQ_penalizedEstimator.}) are
approximate minimizers in the sense that the following hold
\[
P_{n}\ell_{\mu_{n}}\leq\inf_{\mu\in\mathcal{H}^{K}\left(B\right)}P_{n}\ell_{\mu}+o_{p}\left(1\right)
\]
and 
\[
P_{n}\ell_{\mu_{n,\rho}}+\rho\left|\mu_{n,\rho}\right|_{\mathcal{H}^{K}}\leq\inf_{\mu\in\mathcal{H}^{K}}\left\{ P_{n}\ell_{\mu}+\rho\left|\mu\right|_{\mathcal{H}^{K}}\right\} +o_{p}\left(\rho\right).
\]
\end{theorem}

The above result establishes the connection between the constrained
estimator $\mu_{n}$ in (\ref{EQ_theEstimator}) and the penalized
estimator $\mu_{n,\rho}$ in (\ref{EQ_penalizedEstimator.}). It is
worth noting that whether $\mathcal{H}^{K}$ is finite or infinite
dimensional, the estimator $\mu_{n}$ is equivalent to a penalized
estimator with penalty parameter $\rho$ going to zero relatively
fast (i.e. $\rho n^{1/2}\rightarrow\infty$ does not hold). However,
this only ensures uniform consistency and not consistency under the
RKHS norm $\left|\cdot\right|_{\mathcal{H}^{K}}$ (Point 3 in Theorem
\ref{Theorem_consistency}). For the testing procedure discussed in
this paper, we need the estimator to be equivalent to a penalized
one with penalty that converges to zero fast enough. This is achieved
working with the constrained estimator $\mu_{n}$.

Having established consistency, interest also lies in the distribution
of the estimator. We shall only consider the constrained estimator
$\mu_{n}$. To ease notation, for any arbitrary, but fixed real valued
functions $g$ and $g'$ on $\mathcal{Z}$ define $P_{1,j}\left(g,g'\right)=\mathbb{E}g\left(Z_{1}\right)g'\left(Z_{1+j}\right)$.
For suitable $g$ and $g'$, the quantity $\sum_{j\in\mathbb{Z}}P_{1,j}\left(g,g'\right)$
will be used as short notation for sums of population covariances.
We shall also use the additional condition $\left|\Delta_{3}\right|_{\infty}<\infty$,
where $\Delta_{k}\left(z\right)$ is as in Section \ref{Section_Conditions}. 

\begin{theorem}\label{Theorem_weakConvergence}Suppose Conditions
\ref{Condition_kernel}, \ref{Condition_Loss}, and \ref{Condition_Dependence}
hold. If $\mu_{0}\in{\rm int}\left(\mathcal{H}^{K}\left(B\right)\right)$,
then
\[
\sqrt{n}P_{n}\partial\ell_{\mu_{0}}h\rightarrow G\left(h\right),\,h\in\mathcal{H}^{K}\left(1\right)
\]
weakly, where $\left\{ G\left(h\right):h\in\mathcal{H}^{K}\left(1\right)\right\} $
is a mean zero Gaussian process with covariance function
\[
\mathbb{E}G\left(h\right)G\left(h'\right)=\sum_{j\in\mathbb{Z}}P_{1,j}\left(\partial\ell_{\mu_{0}}h,\partial\ell_{\mu_{0}}h'\right)
\]
for any $h,h'\in\mathcal{H}^{K}\left(1\right)$. 

Now, in addition to the above, also suppose that $\left|\Delta_{3}\right|_{\infty}<\infty$.
If $\mu_{n}\in\mathcal{H}^{K}\left(B\right)$ is an asymptotic minimizer
such that $P_{n}\ell_{\mu_{n}}\leq\inf_{\mu\in\mathcal{H}^{K}\left(B\right)}P_{n}\ell_{\mu}+o_{p}\left(n^{-1}\right)$,
and $\sup_{h\in\mathcal{H}^{K}\left(1\right)}P_{n}\partial\ell_{\mu_{n}}h=o_{p}\left(n^{-1/2}\right)$,
then, 
\[
\sqrt{n}P\partial^{2}\ell_{\mu_{0}}\left(\mu_{n}-\mu_{0}\right)h=\sqrt{n}P_{n}\partial\ell_{\mu_{0}}h+o_{p}\left(1\right),\,h\in\mathcal{H}^{K}\left(1\right).
\]
\end{theorem}

The second statement in Theorem \ref{Theorem_weakConvergence} cannot
be established for the penalized estimator with penalty satisfying
$\rho n^{1/2}\rightarrow\infty$. The restriction $\sup_{h\in\mathcal{H}^{K}\left(1\right)}P_{n}\partial\ell_{\mu_{n}}h=o_{p}\left(n^{-1/2}\right)$
holds for finite dimensional models as long as $\mu_{0}\in{\rm int}\left(\mathcal{H}^{K}\left(B\right)\right)$.
When testing restrictions, this is often of interest. However, for
infinite dimensional models this is no longer true as the constraint
is binding even if $\mu_{0}\in{\rm int}\left(\mathcal{H}^{K}\left(B\right)\right)$.
Then, it can be shown that the $o_{p}\left(n^{-1/2}\right)$ term
has to be replaced with $O_{p}\left(n^{-1/2}\right)$ (Lemma \ref{Lemma_lagrangeMultiplierSize},
in the Appendix). This has implications for testing. Additional remarks
can be found in Section \ref{Section_remarksTheoremWeakConvergence}. 

\subsection{Testing Functional Restrictions\label{Section_testingFunctionalRestrictions}}

This section considers tests on functional restrictions possibly allowing
$\mu$ not to be fully specified under the null. As previously discussed,
we write $C_{\mathcal{H}^{K}}=C_{\mathcal{R}_{0}}+C_{\mathcal{R}_{1}}$
as in Section \ref{Section_theGoal}. It is not necessary that $\mathcal{R}_{0}\cap\mathcal{R}_{1}=\emptyset$,
but $\mathcal{R}_{0}$ must be a proper subspace of $\mathcal{H}^{K}$
as otherwise there is no restriction to test. Hence, $\mathcal{R}_{1}$
is not necessarily the complement of $\mathcal{R}_{0}$ in $\mathcal{H}^{K}$.
A few examples clarify the framework. We shall make use of the results
reviewed in Section \ref{Section_reviewRKHS} when constructing the
covariance functions and in consequence the restrictions.

\subsubsection{Examples\label{Section_examples}}

\begin{example}\label{Example_H0ZeroComponent} Let $C_{\mathcal{H}^{K}}\left(s,t\right)=\sum_{k=1}^{K}C\left(s^{\left(k\right)},t^{\left(k\right)}\right)$
so that $\mu\left(x\right)=\sum_{k=1}^{K}f^{\left(k\right)}\left(x^{\left(k\right)}\right)$
as in (\ref{EQ_additiveFunction}), though $x^{\left(k\right)}$ could
be $d$-dimensional as in Example \ref{Example_d-Dimensional}. Consider
the subspace $\mathcal{R}_{0}$ such that $f^{\left(1\right)}=0$.
This is equivalent to $C_{\mathcal{R}_{0}}\left(s,t\right)=\sum_{k=2}^{K}C\left(s^{\left(k\right)},t^{\left(k\right)}\right)$.
In consequence, we can set $C_{\mathcal{R}_{1}}\left(s,t\right)=C\left(s^{\left(1\right)},t^{\left(1\right)}\right)$.\end{example}

Some functional restrictions can also be naturally imposed.

\begin{example}\label{Example_H0LinearityComponent} Suppose that
$\mathcal{H}^{K}$ is an additive space of functions, where each univariate
function is an element in the Sobolev Hilbert space of index $V$
on $\left[0,1\right]$, i.e. functions with $V$ square integrable
weak derivatives. Then, $C_{\mathcal{H}^{K}}\left(s,t\right)=\sum_{k=1}^{K}C\left(s^{\left(k\right)},t^{\left(k\right)}\right)$
where $C\left(s^{\left(k\right)},t^{\left(k\right)}\right)=\sum_{v=1}^{V-1}\lambda_{v}^{2}\left(s^{\left(k\right)}t^{\left(k\right)}\right)^{v}+H_{V}\left(s^{\left(k\right)},t^{\left(k\right)}\right)$
and where $H_{V}$ is the covariance function of the $V$-fold integrated
Brownian motion (see Section \ref{Section_detailsExamples}, in the
supplementary material, or Wahba, 1990, p.7-8, for the details). Consider
the subspace $\mathcal{R}_{0}$ that restricts the univariate RKHS
for the first covariate to be the set of linear functions, i.e. $f^{\left(1\right)}\left(x^{\left(1\right)}\right)=cx^{\left(1\right)}$
for real $c$. Then, $C_{\mathcal{R}_{0}}=\lambda_{1}^{2}s^{\left(1\right)}t^{\left(1\right)}+\sum_{k=2}^{K}C\left(s^{\left(k\right)},t^{\left(k\right)}\right)$.
Hence we can choose $C_{\mathcal{R}_{1}}=\sum_{v=2}^{V-1}\lambda_{v}^{2}\left(s^{\left(1\right)}t^{\left(1\right)}\right)^{v}+H_{V}\left(s^{\left(1\right)},t^{\left(1\right)}\right)$.\end{example}

In both examples above, $\mathcal{R}_{1}$ is the complement of $\mathcal{R}_{0}$
in $\mathcal{H}^{K}$. However, we can just consider spaces $\mathcal{R}_{0}$
and $\mathcal{R}_{1}$ to define the model under the null and the
space of instruments under the alternative. 

\begin{example}\label{Example_H0AdditiveH1General}Suppose $C_{K}\left(s,t\right)$
is a universal kernel on $\left[0,1\right]^{K}\times\left[0,1\right]^{K}$
(see Example \ref{Example_d-Dimensional}). We suppose that $C_{\mathcal{R}_{0}}=\sum_{k=1}^{K}C\left(s^{\left(k\right)},t^{\left(k\right)}\right)$,
while $C_{\mathcal{R}_{1}}=C_{K}\left(s,t\right)$. If $C$ is continuous
and bounded on $\left[0,1\right]\times\left[0,1\right]$, then, $\mathcal{R}_{0}\subset\mathcal{R}_{1}$.
In this case we are testing an additive model against a general nonlinear
one.\end{example}

It is worth noting that Condition \ref{Condition_kernel} restricts
the individual covariances in $C_{\mathcal{H}^{K}}$. The same condition
is inherited by the individual covariances that comprise $C_{\mathcal{R}_{0}}$
(i.e. Condition \ref{Condition_kernel} applies to each individual
component of $C_{\mathcal{R}_{0}}$). In a similar vein, in Example
\ref{Example_H0AdditiveH1General}, the covariance $C_{\mathcal{R}_{1}}$
can be seen as the individual covariance of a multivariate variable
$X^{\left(K+1\right)}:=\left(X^{\left(1\right)},...,X^{\left(K\right)}\right)$
and $C_{\mathcal{R}_{1}}$ will have to satisfy (\ref{EQ_covSeriesRepresentation})
where the features $\varphi_{v}$'s are functions of the variable
$X^{\left(K+1\right)}$. Hence, also Example \ref{Example_H0AdditiveH1General}
fits into our framework, though additional notation is required (see
Section \ref{Section_detailsExamples}, in the supplementary material
for more details).

The examples above can be extended to test more general models.

\begin{example}\label{Example_CAPM}Consider the varying coefficients
regression function $\mu\left(X_{i}\right)=bX_{i}^{\left(1\right)}+\beta\left(X_{i}^{\left(2\right)},...,X_{i}^{\left(K\right)}\right)X_{i}^{\left(1\right)}$.
The function $\beta\left(X_{i}^{\left(2\right)},...,X_{i}^{\left(K\right)}\right)$
can be restricted to linear or additive under the null $\mu\in\mathcal{R}_{0}$.
In the additive case, $C_{\mathcal{R}_{0}}\left(s,t\right)=\lambda_{0}^{2}+s^{\left(1\right)}t^{\left(1\right)}+\sum_{k=1}^{K}C\left(s^{\left(k\right)},t^{\left(k\right)}\right)s^{\left(1\right)}t^{\left(1\right)}$.
In finance, this model can be used to test the conditional Capital
Asset Pricing Model and includes the semiparametric model discussed
in Connor et al. (2012).\end{example}

\subsubsection{Correction for Nuisance Parameters}

Recall that $\mathcal{R}_{0}\left(B\right):=\mathcal{R}_{0}\cap\mathcal{H}^{K}\left(B\right)$
for any $B>0$ and similarly for $\mathcal{R}_{1}\left(B\right)$.
Suppose that $\mu_{0}$ in (\ref{EQ_targetFunction}) lies in the
interior of $\mathcal{R}_{0}\left(B\right)$. Then, the moment equation
$P\partial\ell_{\mu_{0}}h=0$ holds for any $h\in\mathcal{R}_{1}$.
This is because, by definition of (\ref{EQ_targetFunction}), $\partial\ell_{\mu_{0}}$
is orthogonal to all elements in $\mathcal{H}^{K}$. By linearity,
one can restrict attention to $h\in\mathcal{R}_{1}\left(1\right)$
(i.e. $\mathcal{R}_{1}\left(B\right)$ with $B=1$). For such functions
$h$, the statistic $P_{n}\partial\ell_{\mu_{0}}h$ is normally distributed
by Theorem \ref{Theorem_weakConvergence}. In practice, $\mu_{0}$
is rarely known and it is replaced by $\mu_{n0}$ in (\ref{EQ_restrictedEstimator}).
The estimator $\mu_{n0}$ does not need to satisfy $P_{n}\partial\ell_{\mu_{n0}}h=0$
for any $h$ in $\mathcal{H}^{K}\left(1\right)$ under the null. Moreover,
the nuisance parameter affects the asymptotic distribution because
it affects the asymptotic covariance. From now on, we suppose that
the restriction is true, i.e. $\mu_{0}$ in (\ref{EQ_targetFunction})
lies inside $\mathcal{R}_{0}\left(B\right)$, throughout. 

For fixed $\rho\geq0$, let $\Pi_{\rho}$ be the penalized population
projection operator such that 
\begin{equation}
\Pi_{\rho}h=\arg\inf_{\nu\in\mathcal{R}_{0}}P\partial\ell_{\mu_{0}}^{2}\left(h-\nu\right)^{2}+\rho\left|\nu\right|_{\mathcal{H}^{K}}^{2}\label{EQ_projectionOperator}
\end{equation}
for any $h\in\mathcal{H}^{K}$. Let the population projection operator
be $\Pi_{0}$, i.e. (\ref{EQ_projectionOperator}) with $\rho=0$.
We need the following conditions to ensure that we con construct a
test statistic that is not affected by the estimator $\mu_{n0}$. 

\begin{condition}\label{ConditionProjectionTheorem} On top of Conditions
\ref{Condition_kernel}, \ref{Condition_Loss} and \ref{Condition_Dependence},
the following are also satisfied: 
\begin{enumerate}
\item $P\Delta_{1}^{2p}<\infty$, $\left|\Delta_{2}\right|_{\infty}+\left|\Delta_{3}\right|_{\infty}<\infty$
with $p$ as in Conditions \ref{Condition_Loss} and \ref{Condition_Dependence};
\item Under the null, the sequence of scores at the true value is uncorrelated
in the sense that 
\[
\sup_{j>1,h\in\mathcal{H}^{K}\left(1\right)}\left|P_{1,j}\left(\partial\ell_{\mu_{0}}h,\partial\ell_{\mu_{0}}h\right)\right|=0;
\]
 
\item Using the notation in (\ref{EQ_additiveFunction}), for any $\mu\in\mathcal{H}^{K}\left(B\right)$
such that $\left|\mu\right|_{2}^{2}>0$, there is a constant $c>0$
independent of $\mu=\sum_{k=1}^{K}f^{\left(k\right)}$ such that $\left|\mu\right|_{2}^{2}\geq c\sum_{k=1}^{K}\left|f^{\left(k\right)}\right|_{2}^{2}$.
\end{enumerate}
\end{condition}

Remarks on these conditions can be found in Section \ref{Section_RemarksConditions}.
The following holds. 

\begin{theorem}\label{Theorem_projectionScoreTest}Suppose that Condition
\ref{ConditionProjectionTheorem} holds and that $\mu_{n0}\in\mathcal{R}_{0}\left(B\right)$
is such that $P_{n}\ell_{\mu_{n0}}\leq\inf_{\mu\in\mathcal{R}_{0}\left(B\right)}P_{n}\ell_{\mu}+o_{p}\left(n^{-1}\right)$.
Under the null $\mu_{0}\in{\rm int}\left(\mathcal{R}_{0}\left(B\right)\right)$,
we have that
\[
P_{n}\partial\ell_{\mu_{n0}}\left(h-\Pi_{0}h\right)\rightarrow G\left(h-\Pi_{0}h\right),\,h\in\mathcal{H}^{K}\left(1\right),
\]
weakly, where the r.h.s. is a mean zero Gaussian process with covariance
function
\[
\Sigma\left(h,h'\right):=\mathbb{E}G\left(h-\Pi_{0}h\right)G\left(h'-\Pi_{0}h'\right)=P\partial\ell_{\mu_{0}}^{2}\left(h-\Pi_{0}h\right)\left(h'-\Pi_{0}h'\right)
\]
for any $h,h'\in\mathcal{H}^{K}\left(1\right)$. \end{theorem}

Theorem \ref{Theorem_projectionScoreTest} says that if we knew the
projection (\ref{EQ_projectionOperator}), we could derive the asymptotic
distribution of the moment equation. Additional comments on Theorem
\ref{Theorem_projectionScoreTest} are postponed to Section \ref{Section_remarksTheoremWeakConvergence}. 

Considerable difficulties arise when the projection is not known.
In this case, we need to find a suitable estimator for the projection
and construct a test statistic using the moment conditions, whose
number does not need to be bounded. Next we show that it is possible
to do so as if we knew the true projection operator.

\subsubsection{The Test Statistic\label{Section_testStatistic}}

For the moment, to avoid distracting technicalities, suppose that
the projection $\Pi_{0}h$ and the covariance $\Sigma$ are known.
Then, Theorem \ref{Theorem_projectionScoreTest} suggests the construction
of the test statistic for any finite set $\tilde{\mathcal{R}}_{1}\subseteq\mathcal{R}_{1}\cap\mathcal{H}^{K}\left(1\right)$.
Let the cardinality of $\tilde{\mathcal{R}}_{1}$ be $R$, for definiteness.
For the sake of clarity in what follows, fix an order on $\tilde{\mathcal{R}}_{1}$.
Define the test statistic 
\[
S_{n}:=\frac{1}{R}\sum_{h\in\tilde{\mathcal{R}}_{1}}\left[P_{n}\partial\ell_{\mu_{n0}}\left(h-\Pi_{0}h\right)\right]^{2}.
\]
Let $\omega_{k}$ be the $k^{th}$ scaled eigenvalue of the covariance
matrix $\left\{ \Sigma\left(h,h'\right):h,h'\in\mathcal{R}_{1}\right\} $,
i.e., $\omega_{k}\psi_{k}\left(h\right)=\frac{1}{R}\sum_{h'\in\tilde{\mathcal{R}}_{1}}\Sigma\left(h,h'\right)\psi_{k}\left(h'\right)$,
where the $k^{th}$ eigenvector $\left\{ \psi_{k}\left(h\right):h\in\tilde{\mathcal{R}}_{1}\right\} $
satisfies $\frac{1}{R}\sum_{h\in\tilde{\mathcal{R}}_{1}}\psi_{k}\left(h\right)\psi_{l}\left(h\right)=1$
if $k=l$ and zero otherwise.

\begin{remark}\label{Remark_eigenvalues}Given that $R$ is finite,
we can just compute the eigenvalues (in the usual sense) of the matrix
with entries $\Sigma\left(h,h'\right)/R$, $h'h'\in\tilde{\mathcal{R}}_{1}$.\end{remark}

\begin{corollary}\label{Corollary_testStat}Let $\left\{ \omega_{k}:k>1\right\} $
be the set of scaled eigenvalues of the covariance with entries $\Sigma\left(h,h'\right)$
for $h,h'\in\tilde{\mathcal{R}}_{1}$, from Theorem \ref{Theorem_projectionScoreTest}.
Suppose that they are ordered in descending value. Under Condition
\ref{ConditionProjectionTheorem}, $S_{n}\rightarrow S$, in distribution,
where $S=\sum_{k\geq1}\omega_{k}N_{k}^{2}$, and the random variables
$N_{k}$ are independent standard normal.\end{corollary}

To complete this section, it remains to consider an estimator of the
projection $\Pi_{0}h$ and of the covariance function $\Sigma$. The
population projection operator can be replaced by a sample version
\begin{equation}
\Pi_{n,\rho}h=\arg\inf_{\nu\in\mathcal{R}_{0}}P_{n}\partial^{2}\ell_{\mu_{n0}}\left(h-\nu\right)^{2}+\rho\left|\nu\right|_{\mathcal{H}^{K}}^{2},\label{EQ_empiricalProjection}
\end{equation}
which depends on $\rho=\rho_{n}\rightarrow0$. To ease notation, write
$\Pi_{n}=\Pi_{n,\rho}$ for $\rho=\rho_{n}$. By the Representer Theorem,
$\Pi_{n,\rho}h$ is a linear combination of the finite set of functions
$\left\{ C_{\mathcal{R}_{0}}\left(\cdot,X_{j}\right):j=1,2,...,n\right\} $,
as discussed in Section \ref{Section_matrixImplementation}. 

The estimator of $\Sigma$ at $h,h'\in\mathcal{R}_{1}$ is given by
$\Sigma_{n}$ such that 
\begin{equation}
\Sigma_{n}\left(h,h'\right)=P_{n}\partial^{2}\ell_{\mu_{n0}}\left(h-\Pi_{n}h\right)\left(h'-\Pi_{n}h'\right).\label{EQ_covFunction}
\end{equation}
It is not at all obvious that we can effectively use the estimated
projection for all $h\in\mathcal{R}_{1}$, in place of the population
one. The following shows that this is the case.

\begin{theorem}\label{Theorem_feasibleTestStat}In Condition \ref{Condition_kernel},
let $\eta>3/2$, and in (\ref{EQ_empiricalProjection}), choose $\rho$
such that $\rho n^{1/\left(2\lambda\right)}\rightarrow0$ and $\rho n^{\left(2\lambda-1\right)/\left(4\lambda\right)}\rightarrow\infty$,
and define
\begin{equation}
\hat{S}_{n}:=\frac{1}{R}\sum_{h\in\tilde{\mathcal{R}}_{1}}\left[P_{n}\partial\ell_{\mu_{n0}}\left(h-\Pi_{n}h\right)\right]^{2}.\label{EQ_TestStatistic}
\end{equation}
Let $\hat{S}:=\sum_{k\geq1}\omega_{nk}N_{k}^{2}$ where $\omega_{nk}$
is the $k^{th}$ scaled eigenvalue of the covariance matrix $\left\{ \Sigma_{n}\left(h,h'\right):h,h^{'}\in\tilde{\mathcal{R}}_{1}\right\} $
(see Remark \ref{Remark_eigenvalues}). Under Condition \ref{ConditionProjectionTheorem},
$\hat{S}_{n}$ and $\hat{S}$ converge in distribution to $S$, where
the latter is as given in Corollary \ref{Corollary_testStat}. \end{theorem}

Note that the condition on $\rho$ can only be satisfied if in Condition
\ref{Condition_kernel}, $\eta>3/2$, as otherwise the condition on
$\rho$ is vacuous. 

Let $P\left(y|x\right)$ be the distribution of $Y_{i}$ given $X_{i}$.
Define the function $w:\mathcal{X}^{K}\rightarrow\mathbb{R}$ such
that $w\left(x\right):=\int\partial^{2}\ell_{\mu_{0}}\left(\left(y,x\right)\right)dP\left(y|x\right)$.
The function $w$ might be known under the null. In this case, $\partial^{2}\ell_{\mu_{n0}}$
in (\ref{EQ_empiricalProjection}) can be replaced by $w$, i.e.,
define the empirical projection as the $\arg\inf$ of 
\begin{equation}
P_{n}w\left(h-\nu\right)^{2}+\rho\left|\nu\right|_{\mathcal{H}^{K}}^{2}=\frac{1}{n}\sum_{i=1}^{n}w\left(X_{i}\right)\left(h\left(X_{i}\right)-\nu\left(X_{i}\right)\right)^{2}+\rho\left|\nu\right|_{\mathcal{H}^{K}}^{2}\label{EQ_simplifiedEmpiricalProjection}
\end{equation}
w.r.t. $\nu\in\mathcal{R}_{0}$. For example, for the regression problem,
using the square error loss, $w=1$.

\begin{corollary}\label{Corollary_feasibleStatKnownHessian}Suppose
$w$ is known. Replace $\Pi_{n}h$ with the minimizer of (\ref{EQ_simplifiedEmpiricalProjection})
in the construction of the test statistic $\hat{S}_{n}$ and $\Sigma_{n}$.
Suppose Condition \ref{ConditionProjectionTheorem} and $\rho$ such
that $\rho n^{1/\left(2\lambda\right)}\rightarrow0$ and $\rho n^{1/2}\rightarrow\infty$.
Then, the conclusion of Theorem \ref{Theorem_feasibleTestStat} continues
to hold.\end{corollary}

Corollary \ref{Corollary_feasibleStatKnownHessian} improves on Theorem
\ref{Theorem_feasibleTestStat} as it imposes less restrictions on
the exponent $\eta$ and the penalty $\rho$. Despite the technicalities
required to justify the procedure, the implementation shown in Section
\ref{Section_matrixImplementation} is straightforward. In fact $\partial\ell_{\mu_{n0}}$
evaluated at $\left(Y_{i},X_{i}\right)$ is the score for the $i^{th}$
observation and it is the $i^{th}$ entry in $\mathbf{e}_{0}$. On
the other hand the vector $\hat{\mathbf{h}}^{\left(r\right)}$has
$i^{th}$ entry $\left(h^{\left(r\right)}\left(X_{i}\right)-\Pi_{n}h^{\left(r\right)}\left(X_{i}\right)\right)$
and $\tilde{\mathcal{R}}_{1}=\left\{ h^{\left(1\right)},....,h^{\left(R\right)}\right\} $,
for example $\left\{ C_{\mathcal{R}_{1}}\left(\cdot,z_{r}\right):z_{r}\in\mathcal{X}^{K},r=1,2,...,R\right\} $. 

\section{Discussion\label{Section_discussion}}

\subsection{Remarks on Conditions\label{Section_RemarksConditions}}

A minimal condition for the coefficients $\lambda_{v}$ would be $\lambda_{v}\lesssim v^{-\eta}$
with $\eta>1/2$ as this is essentially required for $\sum_{v=1}^{\infty}\lambda_{v}^{2}\varphi_{v}^{2}\left(s\right)<\infty$
for any $s\in\mathcal{X}$. Mendelson (2002) derives consistency under
this minimal condition in the i.i.d. case, but no convergence rates.
Here, the condition is strengthened to $\eta>1$, but it is not necessarily
so restrictive. The covariance in Example \ref{Example_d-Dimensional}
satisfies Condition \ref{Condition_kernel} with exponentially decaying
coefficients $\lambda_{v}$ (e.g. Rasmussen and Williams, 2006, Ch.
4.3.1); the covariance in Example \ref{Example_H0LinearityComponent}
satisfies $\lambda_{v}\lesssim v^{-\eta}$ with $\eta\geq V$ (see
Ritter et al., 1995, Corollary 2, for this and more general results). 

It is not difficult to see that many loss functions (or negative log-likelihoods)
of interest satisfy Condition \ref{Condition_Loss} using the fact
that $\left|\mu\right|_{\infty}\leq\bar{B}$ (square error loss, logistic,
negative log-likelihood of Poisson, etc.). (Recall that $\bar{B}$
was defined just before Condition \ref{Condition_Loss}.) Nevertheless,
interesting loss functions such as absolute deviation for conditional
median estimation do not satisfy Condition \ref{Condition_Loss}.
The extension to such loss functions requires arguments that are specific
to the problem together with additional restrictions to compensate
for the lack of smoothness. Some partial extension to the absolute
loss will be discuss in Section \ref{Section_lossExtensions}. 

Condition \ref{Condition_Dependence} is standard in the literature.
More details and examples can be found in Section \ref{Section_appendixRemarksConditions}
in the Appendix.

In Condition \ref{ConditionProjectionTheorem}, the third derivative
of the loss function and the strengthening of the moment conditions
(Point 1) are used to control the error in the expansion of the moment
equation. The moment conditions are slightly stronger than needed.
The proofs show that we use the following in various places $P\left(\Delta_{1}^{2p}+\Delta_{1}^{p}\Delta_{2}^{p}\right)<\infty$,
$\left|\partial^{2}\ell_{\mu_{0}}\right|_{\infty}+\left|\Delta_{3}\right|_{\infty}<\infty$,
and these can be weakened somehow, but at the cost of introducing
dependence on the exponent $\eta$ ($\eta$ as in Condition \ref{Condition_kernel}).
The condition is satisfied by various loss functions. For example,
the following loss functions have bounded second and third derivative
w.r.t. $t\in\left[-\bar{B},\bar{B}\right]$: $\left(y-t\right)^{2}$
$y\in\mathbb{R}$ (regression), $\ln\left(1+\exp\left\{ -yt\right\} \right)$
$y\in\left\{ -1,1\right\} $ (classification), $-yt+\exp\left\{ t\right\} $
$y\in\left\{ 0,1,2,...\right\} $ (counting).

Time uncorrelated moment equations in Condition \ref{ConditionProjectionTheorem}
are needed to keep the computations feasible. This condition does
not imply that the data are independent. The condition is satisfied
in a variety of situations. In the Poisson example given in the introduction
this is the case as long as $\mathbb{E}_{i-1}Y_{i}=\exp\left\{ \mu_{0}\left(X_{i}\right)\right\} $
(which implicitly requires $X_{i}$ being measurable at time $i-1$).
In general, we still allow for misspecification as long as the conditional
expectation is not misspecified. 

If the scores at the true parameter are correlated, the estimator
of $\Sigma$ needs to be modified to include additional covariance
terms (e.g., Newey-West estimator). Also the projection operator $\Pi_{0}$
has to be modified such that 
\[
\Pi_{0}h=\arg\inf_{v\in\mathcal{R}_{0}}\sum_{j\in\mathbb{Z}}P_{1,j}\left(\partial\ell_{\mu_{0}}\left(h-v\right),\partial\ell_{\mu_{0}}\left(h-v\right)\right).
\]
 This can make the procedure rather involved and it is not discussed
further.

It is simple to show that Point 4 in Condition \ref{ConditionProjectionTheorem}
means that for all pairs $k,l\leq K$ such that $k\neq l$, and for
all $f,g\in\mathcal{H}$ such that $\mathbb{E}\left|f\left(X^{\left(k\right)}\right)\right|^{2}=\mathbb{E}\left|g\left(X^{\left(l\right)}\right)\right|^{2}=1$,
then $\mathbb{E}f\left(X^{\left(k\right)}\right)g\left(X^{\left(l\right)}\right)<1$
(i.e. no perfect correlation when the functions are standardized).

\subsection{Remarks on Theorem \ref{Theorem_weakConvergence}\label{Section_remarksTheoremWeakConvergence}}

The asymptotic distribution of the estimator is immediately derived
if $\mathcal{H}^{K}\left(B\right)$ is finite dimensional. 

\begin{example}\label{Example_covarianceEstimatorSquareLoss}Consider
the rescaled square error loss so that $\partial^{2}\ell_{\mu_{0}}=1$.
Defining $\nu=\lim_{n}\sqrt{n}\left(\mu_{n}-\mu_{0}\right)$, Theorem
\ref{Theorem_weakConvergence} gives 
\[
G\left(h\right)=P\nu h,
\]
in distribution, where $G$ is as in Theorem \ref{Theorem_weakConvergence}
as long as $\mu_{0}\in{\rm int}\left(\mathcal{H}^{K}\left(B\right)\right)$.
The distribution of $\nu$ is then given by the solution to the above
display when $\mathcal{H}^{K}\left(B\right)$ is finite dimensional.
\end{example}

In the infinite dimensional case, Hable (2012) has shown that $\sqrt{n}\left(\mu_{n,\rho}\left(x\right)-\mu_{0,\rho}\left(x\right)\right)$
converges to a Gaussian process whose covariance function would require
the solution of some Fredholm equation of the second type. Recall
that $\mu_{n,\rho}$ is as in (\ref{EQ_penalizedEstimator.}), while
we use $\mu_{0,\rho}$ to denote its population version. The penalty
$\rho=\rho_{n}$ needs to satisfy $\sqrt{n}\left(\rho_{n}-\rho_{0}\right)=o_{p}\left(1\right)$
for some fixed constant $\rho_{0}>0$. When $\mu_{0}\in{\rm int}\left(\mathcal{H}^{K}\left(B\right)\right)$,
we have $\mu_{0}=\arg\min_{\mu\in\mathcal{\mathcal{H}}}P\ell_{\mu}$.
Hence, there is no $\rho_{0}>0$ such that $\mu_{0}=\mu_{0,\rho_{0}}$.
The two estimators are both of interest with different properties.
When the penalty does not go to zero the approximation error is non-negligible,
e.g. for the square loss the estimator is biased.

Theorem \ref{Theorem_weakConvergence} requires $\mu_{0}\in{\rm int}\left(\mathcal{H}^{K}\left(B\right)\right)$.
In the finite dimensional case, the distribution of the estimator
when $\mu_{0}$ lies on the boundary of $\mathcal{H}^{K}\left(B\right)$
is not standard (e.g., Geyer, 1994). In consequence the p-values are
not easy to find.

\subsection{Alternative Constraints\label{Section_alternativeConstraints}}

As an alternative to the norm $\left|\cdot\right|_{\mathcal{H}^{K}}$,
define the norm $\left|f\right|_{\mathcal{L}^{K}}:=\sum_{k=1}^{K}\left|f^{\left(k\right)}\right|_{\mathcal{H}}$.
Estimation in $\mathcal{L}^{K}\left(B\right):=\left\{ f\in\mathcal{H}^{K}:\left|f\right|_{\mathcal{L}^{K}}\leq B\right\} $
is also of interest for variable screening. The following provides
some details on the two different constraints.

\begin{lemma}\label{Lemma_propertiesOfL} Suppose an additive kernel
$C_{\mathcal{H}^{K}}$ as in Section \ref{Section_reviewRKHS}. The
following hold.\\
1. $\left|\cdot\right|_{\mathcal{H}^{K}}$ and $\left|\cdot\right|_{\mathcal{L}^{K}}$
are norms on $\mathcal{H}^{K}$. \\
2. We have the inclusion 
\[
K^{-1/2}\mathcal{H}^{K}\left(1\right)\subset\mathcal{L}^{K}\left(1\right)\subset\mathcal{H}^{K}\left(1\right).
\]
3. For any $B>0$, $\mathcal{H}^{K}\left(B\right)$ and $\mathcal{L}^{K}\left(B\right)$
are convex sets. \\
4. Let $c:=\max_{s\in\mathcal{X}}\sqrt{C\left(s,s\right)}$. If $\mu\in\mathcal{H}^{K}\left(B\right)$,
then, $\sup_{\mu\in\mathcal{H}^{K}\left(B\right)}\left|\mu\right|_{p}\leq c\sqrt{K}B$
for any $p\in\left[1,\infty\right]$, while $\sup_{\mu\in\mathcal{L}^{K}\left(B\right)}\left|\mu\right|_{p}\leq cB$.\end{lemma}

By the inclusion in Lemma \ref{Lemma_propertiesOfL}, all the results
derived for $\mathcal{H}^{K}\left(B\right)$ also apply to $\mathcal{L}^{K}\left(K^{1/2}B\right)$.
In this case, we still need to suppose that $\mu_{0}\in{\rm int}\left(\mathcal{H}^{K}\left(B\right)\right)$.
Both norms are of interest. When interest lies in variable screening
and consistency only, estimation in $\mathcal{L}^{K}\left(B\right)$
inherits the properties of the $l_{1}$ norm (as for LASSO). The estimation
algorithms discussed in Section \ref{Section_estimationAlgo} cover
estimation in both subsets of $\mathcal{H}^{K}$.

\section{Computational Algorithm\label{Section_estimationAlgo}}

By duality, when $\mu\in\mathcal{H}^{K}$ and the constraint is $\left|\mu\right|_{\mathcal{H}^{K}}\leq B$
the estimator is the usual one obtained from the Representer Theorem
(e.g., Steinwart and Christmann, 2008). Estimation in an RKHS poses
computational difficulties when the sample size $n$ is large. Simplifications
are possible when the covariance $C_{\mathcal{H}^{K}}$ admits a series
expansion as in (\ref{EQ_covSeriesRepresentation}) (e.g., L\'{a}zaro-Gredilla
et al., 2010). 

Estimation for functions in $\mathcal{L}^{K}\left(B\right)$ rather
than in $\mathcal{H}^{K}\left(B\right)$ is even more challenging.
Essentially, in the case of the square error loss, estimation in $\mathcal{L}^{K}\left(B\right)$
resembles LASSO, while estimation in $\mathcal{H}^{K}\left(B\right)$
resembles ridge regression. 

A greedy algorithm can be used to solve both problems. In virtue of
Lemma \ref{Lemma_propertiesOfL} and the fact that estimation in $\mathcal{H}^{K}\left(B\right)$
has been considered extensively, only estimation in $\mathcal{L}^{K}\left(B\right)$
will be address in details. The minor changes required for estimation
in $\mathcal{H}^{K}\left(B\right)$ will be discussed in Section \ref{Section_estimationAlgoForH^K}.

\subsection{Estimation in $\mathcal{L}^{K}\left(B\right)$}

Estimation of $\mu_{n}$ in $\mathcal{L}^{K}\left(B\right)$ is carried
out according the following Frank-Wolfe algorithm. Let $f_{m}^{\left(s\left(m\right)\right)}$
be the solution to 
\begin{equation}
\min_{k\leq K}\min_{f^{\left(k\right)}\in\mathcal{H}\left(1\right)}P_{n}\partial\ell_{F_{m-1}}f^{\left(k\right)}\label{EQ_fkAlgoMaximization}
\end{equation}
where $F_{0}=0$, $F_{m}=\left(1-\tau_{m}\right)F_{m-1}+c_{m}f_{m}^{\left(s\left(m\right)\right)}$,
and $c_{m}=B\tau_{m}$, where $\tau_{m}$ is the solution to the line
search 
\begin{equation}
\min_{\tau\in\left[0,1\right]}P_{n}\ell\left(\left(1-\tau\right)F_{m-1}+\tau Bf_{m}^{\left(s\left(m\right)\right)}\right),\label{EQ_lineSearchFrankWolfeAlgo}
\end{equation}
writing $\ell\left(\mu\right)$ instead of $\ell_{\mu}$ for typographical
reasons. Details on how to solve (\ref{EQ_fkAlgoMaximization}) will
be given in Section \ref{Section_solveAdditiveFunctions}; the line
search in (\ref{EQ_fkAlgoMaximization}) is elementary. The algorithm
produces functions $\left\{ f_{j}^{\left(s\left(j\right)\right)}:j=1,2,...,m\right\} $
and coefficients $\left\{ c_{j}:j=1,2,...,m\right\} $. Note that
$s\left(j\right)\in\left\{ 1,2,...,K\right\} $ identifies which of
the $K$ additive functions will be updated at the $j^{th}$ iteration. 

To map the results of the algorithm into functions with representation
in $\mathcal{H}^{K}$, one uses trivial algebraic manipulations. A
simpler variant of the algorithm sets $\tau_{m}=1/m$. In this case,
the solution at the $m^{th}$ iteration, takes the particularly simple
form $F_{m}=\sum_{j=1}^{m}\frac{B}{m}f_{j}^{\left(s\left(j\right)\right)}$
(e.g., Sancetta, 2016) and the $k^{th}$ additive function can be
written as $\tilde{f}^{\left(k\right)}=\frac{B}{m}\sum_{j\leq m:s\left(j\right)=k}f_{j}^{\left(s\left(j\right)\right)}$.

To avoid cumbersome notation, the dependence on the sample size $n$
has been suppressed in the quantities defined in the algorithm. The
algorithm can find a solution with arbitrary precision as the number
of iterations $m$ increases. 

\begin{theorem}\label{Theorem_algoErrorBound}For $F_{m}$ derived
from the above algorithm, 
\[
P_{n}\ell_{F_{m}}\leq\inf_{\mu\in\mathcal{L}^{K}\left(B\right)}P_{n}\ell_{\mu}+\epsilon_{m}
\]
where, 
\[
\epsilon_{m}\lesssim\begin{cases}
\frac{B^{2}\sup_{\left|t\right|\leq B}P_{n}d^{2}L\left(\cdot,t\right)/dt^{2}}{m} & \text{ if }\tau_{m}=\frac{2}{m+2}\text{ or line search in (\ref{EQ_lineSearchFrankWolfeAlgo})}\\
\frac{B^{2}\sup_{\left|t\right|\leq B}\left[P_{n}d^{2}L\left(\cdot,t\right)/dt^{2}\right]\ln\left(1+m\right)}{m} & \text{ if }\tau_{m}=\frac{1}{m}
\end{cases}.
\]
\end{theorem}

For the sake of clarity, recall that $P_{n}d^{2}L\left(\cdot,t\right)/dt^{2}=\frac{1}{n}\sum_{i=1}^{n}d^{2}L\left(Z_{i},t\right)/dt^{2}$.

\subsubsection{Solving for the Additive Functions\label{Section_solveAdditiveFunctions} }

The solution to (\ref{EQ_fkAlgoMaximization}) is found by minimizing
the Lagrangian 
\begin{equation}
P_{n}\partial\ell_{F_{m-1}}f^{\left(k\right)}+\rho\left|f^{\left(k\right)}\right|_{\mathcal{H}}^{2}.\label{EQ_objectiveGreedyUnivariate}
\end{equation}
Let $\Phi^{\left(k\right)}\left(x^{\left(k\right)}\right)=C\left(\cdot,x^{\left(k\right)}\right)$
be the canonical feature map (Lemma 4.19 in Steinwart and Christmann,
2008); $\Phi^{\left(k\right)}$ has image in $\mathcal{H}$ and the
superscript $k$ is only used to stress that it corresponds to the
$k^{th}$ additive component. The first derivative w.r.t. $f^{\left(k\right)}$
is $P_{n}\partial\ell_{F_{m-1}}\Phi^{\left(k\right)}+2\rho f^{\left(k\right)}$,
using the fact that $f^{\left(k\right)}\left(x^{\left(k\right)}\right)=\left\langle f^{\left(k\right)},\Phi^{\left(k\right)}\left(x^{\left(k\right)}\right)\right\rangle _{\mathcal{H}}$,
by the reproducing kernel property. Then, the solution is 
\[
f^{\left(k\right)}=-\frac{1}{2\rho}P_{n}\partial\ell_{F_{m-1}}\Phi^{\left(k\right)},
\]
where $\rho$ is such that $\left|f^{\left(k\right)}\right|_{\mathcal{H}}^{2}=1$.
If $P_{n}\partial\ell_{F_{m-1}}\Phi^{\left(k\right)}=0$, set $\rho=1$.
Explicitly, using the properties of RKHS (see (\ref{EQ_RHKS_normsEquivalences}))
\[
\left|f^{\left(k\right)}\right|_{\mathcal{H}}^{2}=\frac{1}{\left(2\rho\right)^{2}}\sum_{i,j=1}^{n}\frac{\partial\ell_{F_{m-1}}\left(Z_{i}\right)}{n}\frac{\partial\ell_{F_{m-1}}\left(Z_{j}\right)}{n}C\left(X_{i}^{\left(k\right)},X_{j}^{\left(k\right)}\right)
\]
which is trivially solved for $\rho$. With this choice of $\rho$,
the constraint $\left|f^{\left(k\right)}\right|_{\mathcal{H}}\leq1$
is satisfied for all integers $k$, and the algorithm, simply selects
$k$ such that $P_{n}\partial\ell_{F_{m-1}}f^{\left(k\right)}$ is
minimized. Additional practical computational aspects are discussed
in Section \ref{Section_additionalRepresentation} in the Appendix
(supplementary material). 

The above calculations together with Theorem \ref{Theorem_algoErrorBound}
imply the following, which for simplicity, it is stated using the
update $\tau_{m}=m^{-1}$ instead of the line search.

\begin{theorem}Let $\rho_{j}$ be the Lagrange multiplier estimated
at the $j^{th}$ iteration of the algorithm in (\ref{EQ_fkAlgoMaximization})
with $\tau_{m}=m^{-1}$ instead of the line search (\ref{EQ_lineSearchFrankWolfeAlgo}).
Then, 
\[
\mu_{n}=\lim_{m\rightarrow\infty}\sum_{j=1}^{m}\left(-\frac{B}{2m\rho_{j}}\right)P_{n}\partial\ell_{F_{j-1}}\Phi^{\left(s\left(j\right)\right)},
\]
is the solution in $\mathcal{L}^{K}\left(B\right)$.\end{theorem}

\subsection{The Algorithm for Estimation in $\mathcal{H}^{K}\left(B\right)$\label{Section_estimationAlgoForH^K}}

When estimation is constrained in $\mathcal{H}^{K}\left(B\right)$,
the algorithm has to be modified. Let $\Phi\left(x\right)=C_{\mathcal{H}^{K}}\left(\cdot,x\right)$
be the canonical feature map of $\mathcal{H}^{K}$ (do not confuse
$\Phi$ with $\Phi^{\left(k\right)}$ in the previous section). Then,
(\ref{EQ_fkAlgoMaximization}) is replaced by 
\[
\min_{f\in\mathcal{H}^{K}\left(B\right)}P_{n}\partial\ell_{F_{m-1}}f,
\]
 and we denote by $f_{m}\in\mathcal{H}^{K}\left(B\right)$ the solution
at the $m^{th}$ iteration. This solution can be found replacing the
minimization of (\ref{EQ_objectiveGreedyUnivariate}) with minimization
of $P_{n}\partial\ell_{F_{m-1}}f+\rho\left|f\right|_{\mathcal{H}^{K}}^{2}$.
The solution is then $f_{m}=-\frac{1}{2\rho}P_{n}\partial\ell_{F_{m-1}}\Phi$
where $\rho$ is chosen to satisfy the constraint $\left|f\right|_{\mathcal{H}^{K}}^{2}\leq1$.
No other change in the algorithm is necessary and the details are
left to the reader. 

\paragraph{Empirical illustration.}

To gauge the rate at which the algorithm converges to a solution,
we consider the SARCOS data set (\url{http://www.gaussianprocess.org/gpml/data/}),
which comprises a test sample of 44484 observations with $21$ input
variables and a continuous response variable. We standardize the variables
by their Euclidean norm, use the square error loss and the Gaussian
covariance kernel of Example \ref{Example_d-Dimensional} with $d=21$
and $a^{-1}=0.75$. Hence for this example, the kernel is not additive.
Given that the kernel is universal, we shall be able to interpolate
the data if $B$ is chosen large enough: we choose $B=1000$. The
aim is not to find a good statistical estimator, but to evaluate the
computational algorithm. Figure \ref{Figure_R2ApproximationsAlgoSarcos},
plots the $R^{2}$ as a function of the number of iterations $m$.
After approximately $20$ iterations, the algorithm starts to fit
the data better than a constant, and after about 80-90 iterations
the $R^{2}$ is very close to one. The number of operations per iteration
is $O\left(n^{2}\right)$. 

\begin{figure}
\caption{Estimation Algorithm $R^{2}$ as Function of Number of Iterations.
The $R^{2}$ is computed for each iteration $m$ of the estimation
algorithm. Negative $R^{2}$ have been set to zero. }
\label{Figure_R2ApproximationsAlgoSarcos}

\includegraphics[scale=0.7]{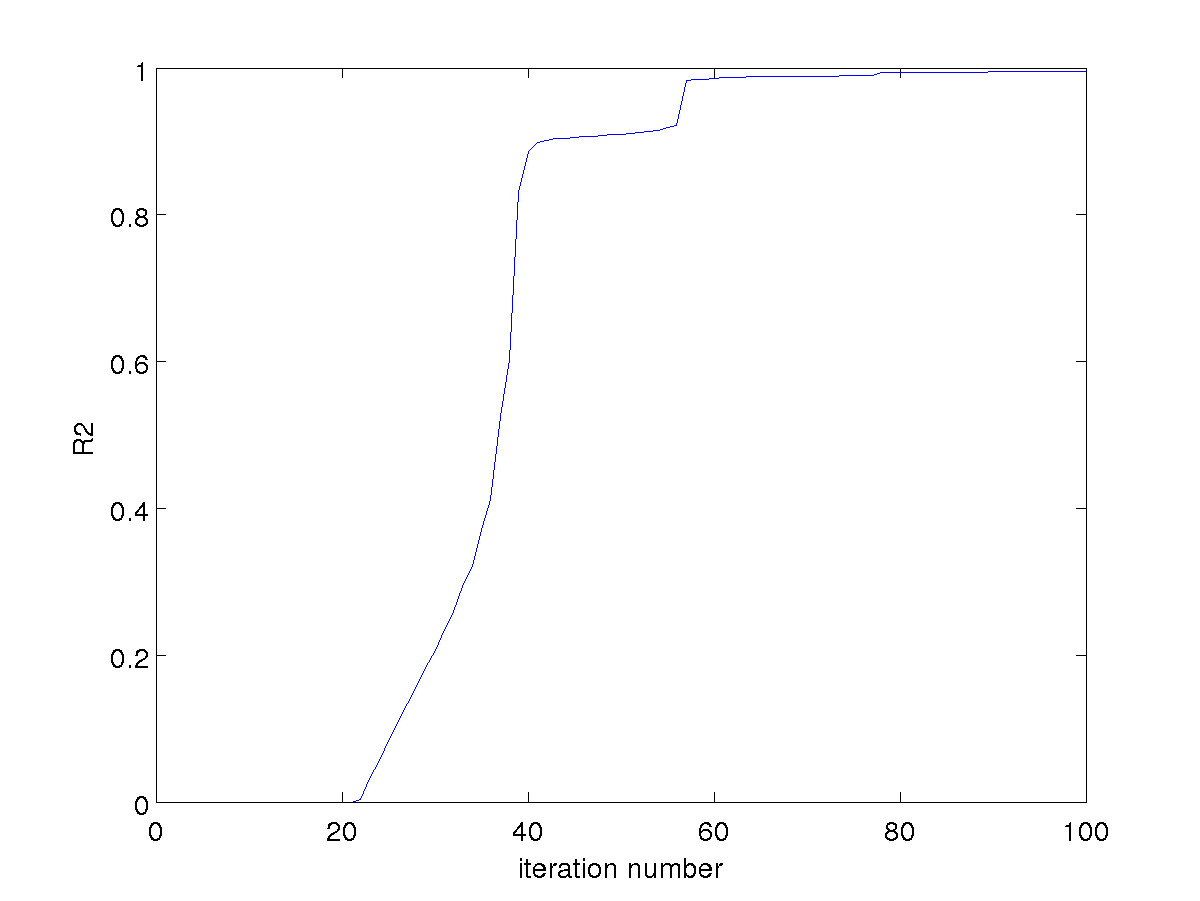}
\end{figure}

\section{Further Remarks\label{Section_furtherRemarks}}

In this last section additional remarks of various nature are included.
A simulation example is used to shed further light on the importance
of the projection procedure.  The paper will conclude with an example
on how Condition \ref{Condition_Loss} can be weakened in order to
accommodate other loss functions, such as the absolute loss. 

\subsection{Some Finite Sample Evidence via Simulation Examples\label{Section_simulations}}

\subsubsection{High Dimensional Model\label{Section_simulationHighDim}}

\paragraph{Simulation Design: True Models.}

Consider the regression problem where $Y_{i}=\mu_{0}\left(X_{i}\right)+\varepsilon_{i}$,
the number of covariates $X^{\left(k\right)}$ is $K=10$, and the
sample size is $n=100,\text{ and }1000$. The covariates are i.i.d.
standard Gaussian random variables that are then truncated to the
interval $\mathcal{X}=\left[-2,2\right]$. Before truncation, the
cross-sectional correlation between $X^{\left(k\right)}$ and $X^{\left(l\right)}$
is $\rho^{\left|k-l\right|}$ with $\rho=0,\text{ and }0.75$, $k,l=1,2,...,K$.
The error terms are i.i.d. mean zero, Gaussian with variance such
that the signal to noise ratio $\sigma_{\mu/\varepsilon}^{2}$ is
equal to $1$ and $0.2$. This is equivalent to an $R^{2}$ of $0.5$
and $0.167$, i.e. a moderate and low $R^{2}$. The following specifications
for $\mu_{0}$ are used: $\mu_{0}\left(X\right)=\sum_{k=1}^{3}b_{k}X^{\left(k\right)}$
with $b_{k}=1/3$ (Lin3:); $\mu_{0}\left(X\right)=\sum_{k=1}^{10}b_{k}X_{i}^{\left(k\right)}$
with $b_{k}=1/10$ (LinAll); $\mu_{0}\left(X\right)=X^{\left(1\right)}+\sum_{v=1}^{9}b_{4,v}\left(X^{\left(4\right)}/2\right)^{v}$
where the $b_{4,v}$'s are uniformly distributed in $\left[-20/v,20/v\right]$
(NonLinear). In NonLinear the first variable enters the model linearly,
the forth variable enters it in a nonlinear fashion, while the remaining
variables do not enter the model. The choice of random coefficient
for NonLinear is along the lines of Friedman (2001) to mitigate the
dependence on a specific nonlinear functional form. The number of
simulations is 1000.

\paragraph{Estimation Details.}

We let $\mathcal{H}^{10}$ be generated by the polynomial additive
kernel $C_{\mathcal{H}^{10}}=\sum_{k=1}^{10}C\left(s^{\left(k\right)},t^{\left(k\right)}\right)$,
where $C\left(s^{\left(k\right)},t^{\left(k\right)}\right)=\sum_{v=1}^{10}v^{-2.2}\left(s^{\left(k\right)}t^{\left(k\right)}\right)^{v}$
. For such kernel, the true models in the simulation design all lie
in a strict subset of $\mathcal{H}^{10}$. Estimation is carried out
in $\mathcal{L}^{10}\left(B\right)$ using the algorithm in Section
\ref{Section_estimationAlgo} with number of iterations $m$ equal
to $500$. This should also allow us to assess whether there is a
distortion in the test results when the estimator minimizes the objective
function on $\mathcal{L}^{10}\left(B\right)$ only approximately.
The parameter $B$ is chosen equal to $10\hat{\sigma}_{Y}$ where
$\hat{\sigma}_{Y}$ is the sample standard deviation of $Y$, which
is a crude approach to keep simulations manageable. The eigenvalues
from the sample covariance were used to simulate the limiting process
(see Lemma \ref{Lemma_covarianceConsistency}), from which the p-values
were derived using $10^{4}$ simulations.

\paragraph{Hypotheses.}

Hypotheses are tested within the framework of Section \ref{Section_testingFunctionalRestrictions}.
We estimate Lin1, Lin2, Lin3 and LinAll, using the restricted kernel
$C_{\mathcal{R}_{0}}\left(s,t\right)=\sum_{k=1}^{J}s^{\left(k\right)}t^{\left(k\right)}$
with $J=1,2,3,10$, i.e., a linear model with 1,2,3 and 10 variables
respectively. We also estimate LinPoly using the restricted kernel
$C_{\mathcal{R}_{0}}\left(s,t\right)=s^{\left(1\right)}t^{\left(1\right)}+\sum_{k=2}^{10}C\left(s^{\left(k\right)},t^{\left(k\right)}\right)$
with $C\left(s^{\left(k\right)},t^{\left(k\right)}\right)$ as defined
in the previous paragraph, i.e., the first variable enters the model
linearly, all other functions are unrestricted. In all cases we test
against the full unrestricted model with kernel $C_{\mathcal{H}^{10}}\left(s,t\right)$.

\paragraph{Test functions. }

We exploit the structure of the covariance kernels. Let the function
$h^{\left(v,k\right)}:\mathcal{X}^{K}\rightarrow\mathbb{R}$ be such
that $h^{\left(v,k\right)}\left(s\right)=v^{-1.1}\left(s^{\left(k\right)}\right)^{v}$,
$v=1,2,...,10,$ $k=1,2,...,K$. For the Lin1, Lin2, Lin3, LinAll
models, we set the test functions as elements in $\left\{ h^{\left(v,k\right)}:v=2,3,...,10,k\leq J\right\} $
with $J=1$ for model Lin1, and so on. We project on the span of $\left\{ h^{\left(1,k\right)}:k\leq J\right\} $.
For LinPoly, we set the test functions as elements in $\left\{ h^{\left(v,1\right)}:v=2,3,...,10\right\} $,
and project on the span of $\left\{ h^{\left(1,1\right)}\right\} \cup\left\{ h^{\left(v,k\right)}:v\leq10,k=2,3,...,K\right\} $.

\paragraph{Results. }

Table \ref{Table_simulationSubset} reports the frequency of rejections
for a given nominal size of the test. Here, results are for $n=1000$,
a signal to noise level $\sigma_{\mu/\varepsilon}^{2}=1$, and $\rho=0$
under the three different true designs: Lin3, LinAll, and NonLin.
The column heading ``No $\Pi$'' means that no correction was used
in estimating the test statistic (i.e. test statistic ignoring the
presence of nuisance parameters). The results for the other configurations
of sample size, signal to noise ratio and correlation in the variables
were similar. The LinPoly model is only estimated when the true model
is NonLin. Here, we only report a subset of the tested hypotheses
(Lin3 and LinAll, only). The complete set of results is in Section
\ref{Section_AppendixNumericalResults} in the Appendix (supplementary
material). Without using the projection adjustment, the size of the
test can be highly distorted, as expected. The results reported in
Table \ref{Table_simulationSubset} show that the test (properly constructed
using the projection adjustment) has coverage probability relatively
close to the nominal one when the null holds, and that the test has
a good level of power.

\begin{table}
\caption{Simulated frequency of rejections for $n=1000$, $\sigma_{\mu/\varepsilon}^{2}=1$,
$\rho=0$. The column heading ``Size'' stands for the nominal size.}
\label{Table_simulationSubset}

\begin{tabular}{cccccccc}
Size &  & Lin3 &  & LinAll &  & LinPoly & \tabularnewline
 &  & No $\Pi$ & $\Pi$ & No $\Pi$ & $\Pi$ & No $\Pi$ & $\Pi$\tabularnewline
 &  & \multicolumn{6}{c}{True model: Lin3}\tabularnewline
0.10 &  & 0.09 & 0.11 & 0.07 & 0.10 & - & -\tabularnewline
0.05 &  & 0.05 & 0.05 & 0.04 & 0.06 & - & -\tabularnewline
 &  & \multicolumn{6}{c}{True model: LinAll}\tabularnewline
0.10 &  & 1.00 & 1.00 & 0.49 & 0.08 & - & -\tabularnewline
0.05 &  & 1.00 & 1.00 & 0.23 & 0.05 & - & -\tabularnewline
 &  & \multicolumn{6}{c}{True model: NonLin}\tabularnewline
0.10 &  & 1.00 & 1.00 & 0.92 & 0.91 & 0.03 & 0.1\tabularnewline
0.05 &  & 1.00 & 1.00 & 0.9 & 0.88 & 0.02 & 0.05\tabularnewline
\end{tabular}
\end{table}

\subsubsection{Infinite Dimensional Estimation\label{Sections_simulationsCov}}

\paragraph{Simulation Design: True Model.}

Consider a bivariate regression model with independent standard normal
errors. The regression function is 
\[
\mu_{0}\left(x\right)=b\left(\frac{1}{2}x^{\left(1\right)}+\frac{3}{2}x^{\left(2\right)}-4\left(x^{\left(2\right)}\right)^{2}+3\left(x^{\left(2\right)}\right)^{3}\right),
\]
where the scalar coefficient $b$ is chosen so that the signal to
noise ratio is $1$ and $0.2$ and $x\in\mathcal{X}^{2}$ where $\mathcal{X}=\left[-2,2\right]$.
The covariates $X_{i}$ and the errors $\varepsilon_{i}$ together
with the other details are as in Section \ref{Section_simulationHighDim}.

\paragraph{Estimation Details and Hypotheses.}

We consider two hypotheses. For hypothesis one, $C_{\mathcal{R}_{0}}\left(s,t\right)=0.5\left(1+\sum_{k=1}^{2}s^{\left(k\right)}t^{\left(k\right)}\right)+0.5\exp\left\{ \frac{1}{2}\left(\frac{s^{\left(2\right)}-t^{\left(2\right)}}{0.75}\right)^{2}\right\} $
(Lin1NonLin) and $C_{\mathcal{R}_{1}}\left(s,t\right)=0.5\exp\left\{ \frac{1}{2}\left(\frac{s^{\left(1\right)}-t^{\left(1\right)}}{0.75}\right)^{2}\right\} $.
This means that we postulate a linear model for the first covariate
and a nonlinear for the second. The true model $\mu_{0}$ is in $\mathcal{R}_{0}$,
hence this hypothesis allows us to verify the size of a Type I error.
For hypothesis two, $C_{\mathcal{R}_{0}}\left(s,t\right)=0.5\left(1+\sum_{k=1}^{2}s^{\left(k\right)}t^{\left(k\right)}\right)$
(LinAll) and $C_{\mathcal{R}_{1}}\left(s,t\right)=0.5\exp\left\{ -\frac{1}{2}\left[\sum_{k=1}^{2}\left(\frac{s^{\left(k\right)}-t^{\left(k\right)}}{0.75}\right)^{2}\right]\right\} $.
In this case, the true model is not in $\mathcal{R}_{0}$ and this
hypothesis allows us to verify the power of the test. All the other
details are as in Section \ref{Section_simulationHighDim}.

\paragraph{Test functions.}

Let the function $h^{\left(r\right)}:\mathcal{X}^{2}\rightarrow\mathbb{R}$
be such that $h^{\left(r\right)}\left(s\right)=C_{\mathcal{R}_{j}}\left(s,X_{r}\right)/\sqrt{C_{\mathcal{R}_{j}}\left(X_{r},X_{r}\right)}$,
$r=1,2,...,n$. For Lin1NonLin, and LinAll, the test functions are
in $\left\{ h^{\left(r\right)}:r=1,2,...,n\right\} $. We project
on the functions $\left\{ C_{\mathcal{R}_{0}}\left(\cdot,X_{r}\right):i=1,2,...,n\right\} $. 

\paragraph{Results. }

Table \ref{Table_simulationSubsetCov} reports the frequency of rejections
for $n=1000$, a signal to noise level $\sigma_{\mu/\varepsilon}^{2}=1$,
and $\rho=0$. The detailed and complete set of results is in Section
\ref{Section_AppendixNumericalResults} in the Appendix (supplementary
material). The results still show a considerable improvement relative
to the naive test. 

\begin{table}
\caption{Simulated frequency of rejections for $n=1000$, and various combinations
of signal to noise $\sigma_{\mu/\varepsilon}^{2}$, and variables
correlation $\rho=0$. The true model is linear in the first variable
and nonlinear in the second variable. The column heading ``Size''
stands for the nominal size.}
\label{Table_simulationSubsetCov}

\begin{tabular}{cccccc}
 &  & \multicolumn{2}{c}{Lin1NonLin} & \multicolumn{2}{c}{LinAll}\tabularnewline
$\left(\sigma_{\mu/\varepsilon}^{2},\rho\right)$ & Size & No $\Pi$ & $\Pi$ & No $\Pi$ & $\Pi$\tabularnewline
\hline 
$\left(1,0\right)$ & 0.10 & \textsf{0.00} & \textsf{0.09} & \textsf{0.99} & \textsf{1.00}\tabularnewline
$\left(1,0\right)$ & 0.05 & \textsf{0.00} & \textsf{0.04} & \textsf{0.83} & \textsf{1.00}\tabularnewline
$\left(.2,0\right)$ & 0.10 & \textsf{0.00} & \textsf{0.09} & \textsf{0.00} & \textsf{1.00}\tabularnewline
$\left(.2,0\right)$ & 0.05 & \textsf{0.00} & \textsf{0.04} & \textsf{0.00} & \textsf{1.00}\tabularnewline
$\left(1,.75\right)$ & 0.10 & \textsf{0.00} & \textsf{0.09} & \textsf{1.00} & \textsf{1.00}\tabularnewline
$\left(1,.75\right)$ & 0.05 & \textsf{0.00} & \textsf{0.03} & \textsf{1.00} & \textsf{1.00}\tabularnewline
$\left(.2,.75\right)$ & 0.10 & \textsf{0.00} & \textsf{0.09} & \textsf{0.13} & \textsf{1.00}\tabularnewline
$\left(.2,.75\right)$ & 0.05 & \textsf{0.00} & \textsf{0.03} & \textsf{0.01} & \textsf{1.00}\tabularnewline
\end{tabular}
\end{table}

\subsection{Weakening Condition \ref{Condition_Loss}: Partial Extension to the
Absolute Loss\label{Section_lossExtensions}}

Some loss functions are continuous and convex, but they are not differentiable
everywhere. An important case is the absolute loss and its variations
used for quantile estimation. The following considers an alternative
to Condition \ref{Condition_Loss} that can be used in this case.
Condition \ref{Condition_Dependence} can be weakened, but Condition
\ref{Condition_kernel} has to be slightly tightened. The details
are stated next, but for simplicity for the absolute loss only. More
general losses such as the one used for quantile estimation can be
studied in a similar way. 

\begin{condition}\label{Condition_absoluteLoss}Suppose that $\ell_{\mu}\left(z\right)=\left|y-\mu\left(x\right)\right|$,
$P\ell_{\mu}<\infty$, and that $P\left(y,x\right)=P\left(y|x\right)P\left(x\right)$
where $P\left(y|x\right)$ (the conditional distribution of $Y$ given
$X$) has a bounded density $pdf\left(y|x\right)$ w.r.t. the Lebesgue
measure on $\mathcal{Y}$, and $P\left(x\right)$ is the distribution
of $X\in\mathcal{X}^{K}$. Moreover, $pdf\left(y|x\right)$ has derivative
w.r.t. $y$ which is uniformly bounded for any $x\in\mathcal{X}^{K}$,
and $\min_{\left|t\right|\leq\bar{B},x\in\mathcal{X}^{K}}pdf\left(t|x\right)>0$
($\bar{B}$ as in Section \ref{Section_Conditions}). The sequence
$\left(Z_{i}\right)_{i\in\mathbb{Z}}$ is stationary with summable
beta mixing coefficients. Finally, Condition \ref{Condition_kernel}
holds with $\lambda>3/2$.\end{condition}

\begin{theorem}\label{Theorem_absoluteLossWeakConvergence} Under
Condition \ref{Condition_absoluteLoss}, Theorems \ref{Theorem_consistency}
and \ref{Theorem_weakConvergence} hold, where $\partial\ell_{\mu_{0}}\left(z\right)=2\times1_{\left\{ y-\mu_{0}\left(x\right)\geq0\right\} }-1$
($1_{\left\{ \cdot\right\} }$ is the indicator function) and 
\[
P\partial^{2}\ell_{\mu_{0}}\left(\mu_{n}-\mu_{0}\right)h=2\int pdf\left(\mu_{0}\left(x\right)|x\right)\sqrt{n}\left(\mu_{n}\left(x\right)-\mu_{0}\left(x\right)\right)h\left(x\right)dP\left(x\right).
\]
\end{theorem}

The result depends on knowledge of the probability density function
of $Y$ conditioning on $X$. Hence, inference in the presence of
nuisance parameters is less feasible within the proposed methodology. 

\newpage{}

\part*{Supplementary Material }

\setcounter{figure}{0} \renewcommand{\thefigure}{A.\arabic{figure}}

\setcounter{equation}{0} \renewcommand{\theequation}{A.\arabic{equation}} 

\setcounter{page}{1} 

\setcounter{section}{0} \renewcommand{\thesection}{A.\arabic{section}}

\section{Appendix 1: Proofs \label{Section_proofs}}

Recall that $\ell_{\mu}\left(Z\right)=L\left(Z,\mu\left(X\right)\right)$
and $\partial^{k}\ell_{\mu}\left(Z\right)=\left.\partial^{k}L\left(Z,t\right)/\partial t^{k}\right|_{t=\mu\left(X\right)}$,
$k\geq1$. Condition \ref{Condition_Loss} implies Fr\'{e}chet differentiability
of $P\ell_{\mu}$ and $P\partial\ell_{\mu}$ (w.r.t. $\mu\in\mathcal{H}^{K}$)
at $\mu$ in the direction of $h\in\mathcal{H}^{K}$. It can be shown
that these two derivatives are $P\partial\ell_{\mu}h$ and $P\partial^{2}\ell_{\mu}hh$,
respectively. For this purpose, we view $P\ell_{\mu}$ as a map from
the space of uniformly bounded functions on $\mathcal{X}^{K}$ ($L_{\infty}\left(\mathcal{X}^{K}\right)$)
to $\mathbb{R}$. The details can be derived following the steps in
the proof of Lemma 2.21 in Steinwart and Christmann (2008) or the
proof of Lemma A.4 in Hable (2012). The application of those proofs
to the current scenario, essentially requires that the loss function
$L\left(Z,t\right)$ is differentiable w.r.t. real $t$, and that
$\mu$ is uniformly bounded, together with integrability of the quantities
$\Delta_{0}$, and $\Delta_{1}$, as implied by Condition \ref{Condition_Loss}.
It will also be necessary to take the Fr\'{e}chet derivative of $P_{n}\ell_{\mu}$
and $P_{n}\partial\ell_{\mu}h$ conditioning on the sample data. By
Condition \ref{Condition_Loss} this will also hold because $\Delta_{0}$,
and $\Delta_{1}$ are finite. This will also allow us to apply Taylor's
Theorem in Banach spaces. Following the aforementioned remarks, when
the loss function is three times differentiable, we also have that
for any $h\in\mathcal{H}^{K}$, the Fr\'{e}chet derivative of $P\partial^{2}\ell_{\mu}h$
in the direction of $h'\in\mathcal{H}^{K}$ is $P\partial^{3}\ell_{\mu}hh'$.
These facts will be used throughout the proofs with no further mention.
Moreover, throughout, for notational simplicity, we tacitly suppose
that $\sup_{x\in\mathcal{X}^{K}}\sqrt{C_{\mathcal{H}^{K}}\left(x,x\right)}=1$
so that $h\in\mathcal{H}^{K}\left(B\right)$ implies that $\left|h\right|_{\infty}\leq B$
for any $B>0$. 

\subsection{Complexity and Gaussian Approximation\label{Section_GaussianApproximation}}

The reader can skip this section and refer to it when needed. Recall
that the $\epsilon$-covering number of a set $\mathcal{F}$ under
the $L_{p}$ norm (denoted by $N\left(\epsilon,\mathcal{F},\left|\cdot\right|_{p}\right)$)
is the minimum number of balls of $L_{p}$ radius $\epsilon$ needed
to cover $\mathcal{F}$. The entropy is the logarithm of the covering
number. The $\epsilon$-bracketing number of the set $\mathcal{F}$
under the $L_{p}$ norm is the minimum number of $\epsilon$-brackets
under the $L_{p}$ norm needed to cover $\mathcal{F}$. Given two
functions $f_{L}\leq f_{U}$ such that $\left|f_{L}-f_{U}\right|_{p}\leq\epsilon$,
an $L_{p}$ $\epsilon$-bracket $\left[f_{L},f_{U}\right]$ is the
set of all functions $f\in\mathcal{F}$ such that $f_{L}\leq f\leq f_{U}$.
Denote the $L_{p}$ $\epsilon$-bracketing number of $\mathcal{F}$
by $N_{\left[\right]}\left(\epsilon,\mathcal{F},\left|\cdot\right|_{p}\right)$.
Under the uniform norm, $N\left(\epsilon,\mathcal{F},\left|\cdot\right|_{\infty}\right)=N_{\left[\right]}\left(\epsilon,\mathcal{F},\left|\cdot\right|_{\infty}\right)$.

In this section, let $\left(G\left(x\right)\right)_{x\in\mathcal{X}}$
be a centered Gaussian process on $\mathcal{X}$ with covariance $C$
as in (\ref{EQ_covSeriesRepresentation}). For any $\epsilon>0$,
let 
\[
\phi\left(\epsilon\right)=-\ln\Pr\left(\left|G\right|_{\infty}<\epsilon\right).
\]

The space $\mathcal{H}$ is generated by the measure of the Gaussian
process $\left(G\left(x\right)\right)_{x\in\mathcal{X}}$ with covariance
function $C$. In particular, $G\left(x\right)=\sum_{v=1}^{\infty}\lambda_{v}\xi_{v}\varphi_{v}\left(x\right)$,
where the $\left(\xi_{v}\right)_{v\geq1}$ is a sequence of i.i.d.standard
normal random variables, and the equality holds in distribution. For
any positive integer $V$, the $l$-approximation number $l_{V}\left(G\right)$
w.r.t.  $\left|\cdot\right|_{\infty}$ (e.g., Li and Linde, 1999,
see also Li and Shao, 2001) is bounded above by $\left(\mathbb{E}\left|\sum_{v>V}\lambda_{v}\xi_{v}\varphi_{v}\right|_{\infty}^{2}\right)^{1/2}$.
Under Condition \ref{Condition_kernel}, deduce that 
\begin{equation}
l_{V}\left(G\right)\lesssim\sum_{v>V}\lambda_{v}\lesssim V^{-\left(\eta-1\right)}.\label{EQ_lApproximationNumber}
\end{equation}
There is a link between the $l_{V}\left(G\right)$ approximating number
of the centered Gaussian process $G$ with covariance $C$ and the
$L_{\infty}$ $\epsilon$-entropy number of the class of functions
$\mathcal{H}\left(1\right)$, which is denoted by $\ln N\left(\epsilon,\mathcal{H}\left(1\right),\left|\cdot\right|_{\infty}\right)$.
These quantities are also related to the small ball probability of
$G$ under the sup norm (results hold for other norms, but will not
be used here). We have the following bound on the $\epsilon$-entropy
number of $\mathcal{H}\left(1\right)$.

\begin{lemma}\label{Lemma_B_H_entropy}Under Condition \ref{Condition_kernel},
$\ln N\left(\epsilon,\mathcal{H}\left(1\right),\left|\cdot\right|_{\infty}\right)\lesssim\epsilon^{-2/\left(2\eta-1\right)}$.\end{lemma}

\begin{proof}As previously remarked, the space $\mathcal{H}\left(1\right)$
is generated by the law of the Gaussian process $G$ with covariance
function $C$. For any integer $V<\infty$, the $l$-approximation
number of $G$, $l_{V}\left(G\right)$ is bounded as in (\ref{EQ_lApproximationNumber}).
In consequence, $\phi\left(\epsilon\right)\lesssim\epsilon^{-1/\left(\eta-1\right)}$,
by Proposition 4.1 in Li and Linde (1999). Then, Theorem 1.2 in Li
and Linde (1999) implies that $\ln N\left(\epsilon,\mathcal{H}\left(1\right),\left|\cdot\right|_{\infty}\right)\lesssim\epsilon^{-2/\left(2\eta-1\right)}$.\end{proof}

\begin{lemma}\label{Lemma_entropyOfL}Under Condition \ref{Condition_kernel},
\[
\ln N\left(\epsilon,\mathcal{H}^{K}\left(B\right),\left|\cdot\right|_{\infty}\right)\lesssim\left(B/\epsilon\right)^{2/\left(2\eta-1\right)}+K\ln\left(\frac{B}{\epsilon}\right).
\]

\end{lemma}

\begin{proof}Functions in $\mathcal{H}^{K}\left(B\right)$ can be
written as $\mu\left(x\right)=\sum_{k=1}^{K}b_{k}f^{\left(k\right)}\left(x^{\left(k\right)}\right)$
where $f^{\left(k\right)}\in\mathcal{H}\left(1\right)$. Hence, the
covering number of $\left\{ \mu\in\mathcal{H}^{K}\left(B\right)\right\} $
is bounded by the product of the covering number of the sets $\mathcal{F}_{1}:=\left\{ \left(b_{1},b_{2},...,b_{K}\right)\in\mathbb{R}^{K}:\sum_{k=1}^{K}b_{k}^{2}\leq B^{2}\right\} $
and $\mathcal{F}_{2}:=\left\{ f^{\left(k\right)}\in\mathcal{H}\left(B\right)\right\} $.
The $\epsilon$-covering number of $\mathcal{F}_{1}$ is bounded by
a constant multiple of $\left(B/\epsilon\right)^{K}$ under the supremum
norm. The $\epsilon$-covering number of $\mathcal{F}_{2}$ is given
by Lemma \ref{Lemma_B_H_entropy}, i.e. $\exp\left\{ \left(B/\epsilon\right)^{2/\left(2\eta-1\right)}\right\} $.
The lemma follows by taking logs of these quantities.\end{proof}

Next, link the entropy of $\mathcal{H}\left(1\right)$ to the entropy
with bracketing of $\ell_{\mu}h$. 

\begin{lemma}\label{Lemma_entropy_dLh} Suppose Condition \ref{Condition_kernel}
holds. For the set $\mathcal{F}:=\left\{ \partial\ell_{\mu}h:\mu\in\mathcal{H}^{K}\left(B\right),h\in\mathcal{H}^{K}\left(1\right)\right\} $,
for any $p\in\left[1,\infty\right]$ satisfying Condition \ref{Condition_Loss},
the $L_{p}$ $\epsilon$-entropy with bracketing is 
\[
\ln N_{\left[\right]}\left(\epsilon,\mathcal{F},\left|\cdot\right|_{p}\right)\lesssim\left(B/\epsilon\right)^{2/\left(2\eta-1\right)}+K\ln\left(\frac{B}{\epsilon}\right).
\]

The same exact result holds for $\mathcal{F}:=\left\{ \ell_{\mu}:\mu\in\mathcal{H}^{K}\left(B\right)\right\} $
under Condition \ref{Condition_Loss}. \end{lemma}

\begin{proof}In the interest of conciseness, we only prove the result
for 
\[
\mathcal{F}:=\left\{ \partial\ell_{\mu}h:\mu\in\mathcal{H}^{K}\left(B\right),h\in\mathcal{H}^{K}\left(1\right)\right\} .
\]
To this end, note that by Condition \ref{Condition_Loss} and the
triangle inequality, 
\[
\left|\partial\ell_{\mu}h-\partial\ell_{\mu'}h'\right|\leq\left|\partial\ell_{\mu}-\partial\ell_{\mu'}\right|\sup_{h\in\mathcal{H}^{K}\left(1\right)}\left|h\right|+\sup_{\mu\in\mathcal{H}^{K}\left(B\right)}\left|\partial\ell_{\mu}\right|\left|h-h'\right|.
\]
 By Condition \ref{Condition_Loss}, $\left|\partial\ell_{\mu}\left(z\right)\right|\leq\Delta_{1}\left(z\right)$,
and $\left|\partial\ell_{\mu}\left(z\right)-\partial\ell_{\mu'}\left(z\right)\right|\leq\Delta_{2}\left(z\right)\left|\mu\left(x\right)-\mu'\left(x\right)\right|$,
and $P\left(\Delta_{1}^{p}+\Delta_{2}^{p}\right)<\infty$. By Lemma
\ref{Lemma_propertiesOfL}, $\left|h\left(x\right)\right|\lesssim1$.
By these remarks, the previous display is bounded by 
\[
\Delta_{2}\left(z\right)\left|\mu-\mu'\right|_{\infty}+\Delta_{1}\left(z\right)\left|h-h'\right|_{\infty}.
\]
Theorem 2.7.11 in van der Vaart and Wellner (2000) says that the $L_{p}$
$\epsilon$-bracketing number of class of functions satisfying the
above Lipschitz kind of condition is bounded by the $L_{\infty}$
$\epsilon'$-covering number of $\mathcal{H}^{K}\left(B\right)\times\mathcal{H}^{K}\left(1\right)$
with $\epsilon'=\epsilon/\left[2\left(P\left|\Delta_{1}+\Delta_{2}\right|^{p}\right)^{1/p}\right]$.
Using Lemma \ref{Lemma_entropyOfL}, the statement of the lemma is
deduced because the product of the covering numbers is the sum of
the entropy numbers.\end{proof}

We shall also need the following.

\begin{lemma}\label{Lemma_entropy_dL2hh} Suppose Condition \ref{Condition_kernel}
holds. For the set $\mathcal{F}:=\left\{ \partial\ell_{\mu}^{2}hh':\mu\in\mathcal{H}^{K}\left(B\right),h,h'\in\mathcal{H}^{K}\left(1\right)\right\} $,
and any $p\in\left[1,\infty\right]$ satisfying Condition \ref{Condition_Loss}
with the addition that $P\left(\Delta_{1}^{2p}+\Delta_{1}^{p}\Delta_{2}^{p}\right)<\infty$,
the $L_{p}$ $\epsilon$-entropy with bracketing is 
\[
\ln N_{\left[\right]}\left(\epsilon,\mathcal{F},\left|\cdot\right|_{p}\right)\lesssim\left(B/\epsilon\right)^{2/\left(2\eta-1\right)}+K\ln\left(\frac{B}{\epsilon}\right).
\]

If also $P\left(\Delta_{2}^{p}+\Delta_{3}^{p}\right)<\infty$, $\left\{ \partial^{2}\ell_{\mu}hh':\mu\in\mathcal{H}^{K}\left(B\right),h,h'\in\mathcal{H}^{K}\left(1\right)\right\} $
has $L_{p}$ $\epsilon$-entropy with bracketing as in the above display.
\end{lemma}

\begin{proof}The proof is the same as the one of Lemma \ref{Lemma_entropy_dLh}.
By Condition \ref{Condition_Loss} and the triangle inequality, for
$g,g'\in\mathcal{H}^{K}\left(1\right)$ 
\[
\left|\partial\ell_{\mu}^{2}hh'-\partial\ell_{\mu'}^{2}gg'\right|\leq\left|\partial\ell_{\mu}^{2}-\partial\ell_{\mu'}^{2}\right|\sup_{h\in\mathcal{H}^{K}\left(1\right)}\left|h\right|^{2}+\sup_{\mu\in\mathcal{H}^{K}\left(B\right)}\left|\partial\ell_{\mu}^{2}\right|\left|hh'-gg'\right|.
\]
By Condition \ref{Condition_Loss}, $\left|\partial\ell_{\mu}^{2}\left(z\right)\right|\leq\Delta_{1}^{2}\left(z\right)$,
$\left|\partial\ell_{\mu}^{2}\left(z\right)-\partial\ell_{\mu'}^{2}\left(z\right)\right|\leq2\Delta_{1}\left(z\right)\Delta_{2}\left(z\right)\left|\mu\left(x\right)-\mu'\left(x\right)\right|$,
and $P\left(\Delta_{1}^{2p}+\Delta_{1}^{p}\Delta_{2}^{p}\right)<\infty$.
By Lemma \ref{Lemma_propertiesOfL}, $\left|h\left(x\right)\right|\lesssim1$.
By these remarks, the previous display is bounded by 
\[
2\Delta_{1}\left(z\right)\Delta_{2}\left(z\right)\left|\mu-\mu'\right|_{\infty}+\Delta_{1}^{2}\left(z\right)\left|h-h'\right|_{\infty}.
\]
Theorem 2.7.11 in van der Vaart and Wellner (2000) says that the $L_{p}$
$\epsilon$-bracketing number of class of functions satisfying the
above Lipschitz kind of condition is bounded by the $L_{\infty}$
$\epsilon'$-covering number of $\mathcal{H}^{K}\left(B\right)\times\mathcal{H}^{K}\left(1\right)$
with $\epsilon'=\epsilon/\left[2\left(P\left|\Delta_{1}^{2p}+\Delta_{1}^{p}\Delta_{2}^{p}\right|\right)^{1/p}\right]$.

The last statement in the lemma is proved following step by step the
proof of Lemma \ref{Lemma_entropy_dLh} with $\partial\ell_{\mu}$
replaced by $\partial^{2}\ell_{\mu}$ and $h$ by $hh'$. \end{proof}

\begin{lemma}\label{Lemma_gaussianConvergence} Under Conditions
\ref{Condition_kernel}, \ref{Condition_Loss}, and \ref{Condition_Dependence},
\[
\sqrt{n}\left(P_{n}-P\right)\partial\ell_{\mu}h\rightarrow G\left(\partial\ell_{\mu},h\right)
\]
weakly, where $G\left(\partial\ell_{\mu},h\right)$ is a mean zero
Gaussian process indexed by $\left(\partial\ell_{\mu},h\right)\in\left\{ \partial\ell_{\mu}:\mu\in\mathcal{H}^{K}\left(B\right)\right\} \times\mathcal{H}^{K}\left(1\right)$,
with a.s. continuous sample paths and covariance function
\[
\mathbb{E}G\left(\partial\ell_{\mu},h\right)G\left(\partial\ell_{\mu'},h'\right)=\sum_{j\in\mathbb{Z}}P_{1,j}\left(\partial\ell_{\mu}h,\partial\ell_{\mu}h'\right)
\]
\end{lemma}

\begin{proof}The proof shall use the main result in Doukhan et al.
(1995). Let $\mathcal{F}:=\left\{ \partial\ell_{\mu}h:\mu\in\mathcal{H}^{K}\left(B\right),h\in\mathcal{H}^{K}\left(1\right)\right\} $.
The elements in $\mathcal{F}$ have finite $L_{p}$ norm because $P\left|\partial\ell_{\mu}\right|^{p}\leq P\Delta_{1}^{p}$
by Condition \ref{Condition_Loss}, and $\left|h\right|_{\infty}\lesssim1$
by Lemma \ref{Lemma_propertiesOfL}. To avoid extra notation, it is
worth noting that the entropy integrability condition in Doukhan et
al. (1995, Theorem 1, eq. 2.10) is implied by 
\begin{equation}
\int_{0}^{1}\sqrt{\ln N_{\left[\right]}\left(\epsilon,\mathcal{F},\left|\cdot\right|_{p}\right)}d\epsilon<\infty.\label{EQ_entropyIntegral}
\end{equation}
and $\beta\left(i\right)\lesssim\left(1+i\right)^{-\beta}$ with $\beta>p/\left(p-2\right)$
and $p>2$. Then, Theorem 1 in Doukhan et al. (1995) shows that the
empirical process indexed in $\mathcal{F}$ converges weakly to the
Gaussian one given in the statement of the present lemma. By Condition
\ref{Condition_Dependence}, it is sufficient to show (\ref{EQ_entropyIntegral}).
By Lemma \ref{Lemma_entropy_dLh}, the integral is finite because
$\lambda>1$ by Condition \ref{Condition_kernel}. \end{proof}

\subsection{Proof of Theorem \ref{Theorem_consistency}}

The proof is split into the part concerned with the constrained estimator
and the one that studies the penalized estimator.

\subsubsection{Consistency of the Constrained Estimator}

At first we show Point 1 verifying the conditions of Theorem 3.2.5
van der Vaart and Wellner (2000) which we will refer to as VWTh herein.
To this end, by Taylor's Theorem in Banach spaces, 
\begin{eqnarray*}
P\ell_{\mu} & -P\ell_{\mu_{0}}= & P\partial\ell_{\mu_{0}}\left(\mu-\mu_{0}\right)+\frac{1}{2}P\partial^{2}\ell_{\mu_{t}}\left(\mu-\mu_{0}\right)^{2}
\end{eqnarray*}
for $\mu_{t}=\mu+t\left(\mu_{0}-\mu\right)$ with some $t\in\left[0,1\right]$
and arbitrary $\mu\in\mathcal{H}^{K}\left(B\right)$. The variational
inequality $P\partial\ell_{\mu_{0}}\left(\mu-\mu_{0}\right)\geq0$
holds by definition of $\mu_{0}$ and the fact that $\mu\in\mathcal{H}^{K}\left(B\right)$.
Therefore, the previous display implies that $P\ell_{\mu}-P\ell_{\mu_{0}}\gtrsim P\left(\mu-\mu_{0}\right)^{2}$
because $P\partial^{2}\ell_{\mu_{t}}\left(\mu-\nu\right)^{2}\geq P\left(\mu-\nu\right)^{2}\geq0$
by Condition \ref{Condition_Loss}. The right hand most inequality
holds with equality if and only if $\mu=\mu_{0}$ in $L_{2}$. This
verifies the first condition in VWTh. Given that the loss function
is convex and coercive and that $\mathcal{H}^{K}\left(B\right)$ is
a closed convex set, this also shows that the population minimizer
$\mu_{0}$ exists and is unique up to an $L_{2}$ equivalence class,
as stated in the theorem. Moreover, given that $\mu,\mu_{0}\in\mathcal{H}^{K}\left(B\right)$,
then both $\mu$ and $\mu_{0}$ are uniformly bounded by a constant
multiple $B$, hence for simplicity suppose they are bounded by $B$.
This implies the following relation 
\[
B^{2-p}\left|\mu-\mu_{0}\right|_{p}\leq\left|\mu-\mu_{0}\right|_{2}\leq\left|\mu-\mu_{0}\right|_{p}
\]
for any $p\in\left(2,\infty\right)$. Hence, for any finite real $\delta$,
\[
\sup_{\left|\mu-\mu_{0}\right|_{2}<\delta}\mathbb{E}\left|\left(P_{n}-P\right)\left(\ell_{\mu}-\ell_{\mu_{0}}\right)\right|\leq\sup_{\left|\mu-\mu_{0}\right|_{p}<B^{p-2}\delta}\mathbb{E}\left|\left(P_{n}-P\right)\left(\ell_{\mu}-\ell_{\mu_{0}}\right)\right|
\]
To verify the second condition in VWTh, we need to find a function
$\phi\left(\delta\right)$ that grows slower than $\delta^{2}$ such
that the r.h.s. of the above display is bounded above by $n^{-1/2}\phi\left(\delta\right)$.
To this end, note that we are interested in the following class of
functions $\mathcal{F}:=\left\{ \ell_{\mu}-\ell_{\mu_{0}}:\left|\mu-\mu_{0}\right|_{p}\leq\delta'\right\} $
with $\delta'=B^{p-2}\delta$. This class of functions satisfies $\left|\ell_{\mu}-\ell_{\mu_{0}}\right|_{p}\leq\left(\Delta_{1}^{p}\right)^{1/p}\delta'$
using the differentiability and the bounds implied by Condition \ref{Condition_Loss}.
Theorem 3 in Doukhan et al. (1995) says that that for large enough
$n$, eventually (see their page 410), 
\[
\phi\left(\delta\right)\lesssim\int_{0}^{B^{p-2}\delta}\sqrt{\ln N_{\left[\right]}\left(\epsilon,\mathcal{F},\left|\cdot\right|_{p}\right)}d\epsilon.
\]
Note that we have $L_{p}$ balls of size $B^{p-2}\delta$ rather than
$\delta$ and for this reason we have modified the limit in the integral.
Moreover, as remarked in the proof of Lemma \ref{Lemma_gaussianConvergence},
the entropy integral in Doukhan et al. (1995) uses the bracketing
number based on another norm. However, their norm is bounded by the
$L_{p}$ norm used here under the restrictions we impose on the mixing
coefficients via Condition \ref{Condition_Dependence}. To compute
the integral we use Lemma \ref{Lemma_entropy_dLh}, so that the l.h.s.
of the display is a constant multiple of $B^{\left(1-\alpha\right)+\alpha\left(p-2\right)}\delta^{\alpha}$
with $\alpha=\left(2\eta-2\right)/\left(2\eta-1\right)$. The third
condition in VWTh requires to find a sequence $r_{n}$ such that $r_{n}^{2}\phi\left(r_{n}^{-1}\right)\leq n^{1/2}$.
Given that $\phi\left(\delta\right)\lesssim A\delta^{\alpha}$ with
$A:=B^{\left(1-\alpha\right)+\alpha\left(p-2\right)}$, deduce that
we can set $r_{n}\asymp n^{\left(2\eta-1\right)/\left(4\eta\right)}$.
Then VWTh states that $\left|\mu_{n}-\mu_{0}\right|_{2}=O_{p}\left(r_{n}^{-1}\right)$.
Of course, if $\mathcal{H}^{K}$ is finite dimensional, it is not
difficult to show that $r_{n}\asymp n^{1/2}$. The space $\mathcal{H}^{K}$
is finite dimensional if (\ref{EQ_covSeriesRepresentation}) has a
finite number of terms. 

We also show that $\sup_{\mu\in\mathcal{H}^{K}\left(B\right)}\left|\left(P_{n}-P\right)\ell_{\mu}\right|\rightarrow0$
a.s. which shall imply $\left|\mu_{n}-\mu_{0}\right|_{2}\rightarrow0$
a.s. (Corollary 3.2.3 in van der Vaart and Wellner, 2000, replacing
the in probability result with a.s.). This only requires the loss
function to be integrable (if the loss is positive), but does not
allow us to derive convergence rates. For any fixed $\mu$, $\left|\left(P_{n}-P\right)\ell_{\mu}\right|\rightarrow0$
a.s., by the ergodic theorem, because $P\left|\ell_{\mu}\right|<\infty$
by Condition \ref{Condition_Loss}. Hence, it is just sufficient to
show that $\left\{ \ell_{\mu}:\mu\in\mathcal{H}^{K}\left(B\right)\right\} $
has finite $\epsilon$-bracketing number under the $L_{1}$ norm (e.g.,
see the proof of Theorem 2.4.1 in van der Vaart and Wellner, 2000).
This is the case by Lemma \ref{Lemma_entropy_dLh}, because by Condition
\ref{Condition_kernel}, $\eta>1$. Hence, $\left|\mu_{n}-\mu_{0}\right|_{2}\rightarrow0$
a.s..

To turn the $L_{2}$ convergence into uniform, note that $\mathcal{H}^{K}\left(B\right)$
is compact under the uniform norm and functions in $\mathcal{H}^{K}\left(B\right)$
are defined on a compact domain $\mathcal{X}^{K}$. Hence, $\mathcal{H}^{K}\left(B\right)$
is a subset of the space of continuous bounded function equipped with
the uniform norm. In consequence, any convergent sequence in $\mathcal{H}^{K}\left(B\right)$
converges uniformly. 

We now turn to the relation between the constrained and penalized
estimator, which will also conclude the proof of Theorem \ref{Theorem_consistency}.

\paragraph{}

\subsubsection{The Constraint and the Lagrange Multiplier\label{Section_lagrangeMultiplierLemma}}

The following lemma puts together crucial results for estimation
in RKHS (Steinwart and Christmann, 2008, Theorems 5.9 and 5.17 for
a proof). The cited results make use of the definition of integrable
Nemitski loss of finite order $p$ (Steinwart and Christmann, 2008,
Def. 2.16). However, under Condition \ref{Condition_Loss}, the proofs
of those results still hold. 

\begin{lemma}\label{Lemma_SteinwartChirstman} Under Condition \ref{Condition_Loss},
\begin{equation}
\left|\mu_{0,\rho}-\mu_{n,\rho}\right|_{\mathcal{H}^{K}}\leq\frac{1}{\rho}\left|P\partial\ell_{\mu_{0,\rho}}\Phi-P_{n}\partial\ell_{\mu_{0,\rho}}\Phi\right|_{\mathcal{H}^{K}},\label{EQ_normLagrangeEstimatorDifference}
\end{equation}
where $\Phi\left(x\right)=C_{\mathcal{H}^{K}}\left(\cdot,x\right)$
is the canonical feature map. Moreover, if $\mu_{0,\rho}$ is bounded
for any $\rho\rightarrow0$, then $\left|\mu_{0,\rho}-\mu_{0}\right|_{\mathcal{H}^{K}}\rightarrow0$.

\end{lemma}

We apply Lemma \ref{Lemma_SteinwartChirstman} and the results in
Section \ref{Section_GaussianApproximation} to derive the following. 

\begin{lemma}\label{Lemma_lagrangeMultiplierSize}Suppose Conditions
\ref{Condition_kernel}, \ref{Condition_Loss} and \ref{Condition_Dependence}.
The following statements hold. 
\begin{enumerate}
\item There is a finite $B$ such that $\mu_{0}\in{\rm int}\left(\mathcal{H}^{K}\left(B\right)\right)$. 
\item For any $\rho>0$ possibly random, $\left|\mu_{n,\rho}-\mu_{0\rho}\right|_{\mathcal{H}^{K}}^{2}=O_{p}\left(\rho^{-2}n^{-1}\right)$,
and $\left|\mu_{n,\rho}\right|_{\mathcal{H}^{K}}\leq B$ eventually
in probability for any $\rho\rightarrow0$ such that $\rho n^{1/2}\rightarrow\infty$.
\item There is a $\rho=O_{p}\left(n^{-1/2}\right)$ such that $\left|\mu_{n,\rho}\right|_{\mathcal{H}^{K}}\leq B$
and 
\[
\sup_{h\in\mathcal{H}^{K}\left(1\right)}P_{n}\partial\ell_{\mu_{n,\rho}}h=O_{p}\left(n^{-1/2}B\right).
\]
\end{enumerate}
\end{lemma}

\begin{proof} Given that $K$ is finite and the kernel is additive,
there is no loss in restricting attention to $K=1$ in order to reduce
the notational burden. We shall need a bound for the r.h.s. of (\ref{EQ_normLagrangeEstimatorDifference}).
By (\ref{EQ_covSeriesRepresentation}), the canonical feature map
can be written as $\Phi\left(x\right)=\sum_{v=1}^{\infty}\lambda_{v}^{2}\varphi_{v}\left(\cdot\right)\varphi_{v}\left(x\right)$.
This implies that, 
\[
\left(P_{n}-P\right)\partial\ell_{\mu_{0,\rho}}\Phi\left(x\right)=\sum_{v=1}^{\infty}\left[\lambda_{v}^{2}\left(P_{n}-P\right)\partial\ell_{\mu_{0,\rho}}\varphi_{v}\right]\varphi_{v}\left(x\right).
\]
By Lemma \ref{Lemma_SteinwartChirstman}, (\ref{EQ_RHKS_normsEquivalences}),
and the above, 
\[
\left|\left(P_{n}-P\right)\partial\ell_{\mu_{0,\rho}}\Phi\right|_{\mathcal{H}^{K}}^{2}=\sum_{v=1}^{\infty}\frac{\left[\lambda_{v}^{2}\left(P_{n}-P\right)\partial\ell_{\mu_{0,\rho}}\varphi_{v}\right]^{2}}{\lambda_{v}^{2}}=\sum_{v=1}^{\infty}\lambda_{v}^{2}\left[\left(P_{n}-P\right)\partial\ell_{\mu_{0,\rho}}\varphi_{v}\right]^{2}.
\]
In consequence of the above display, by the triangle inequality, 
\begin{eqnarray*}
\left|\mu_{0,\rho}-\mu_{n,\rho}\right|_{\mathcal{H}^{K}} & \leq & \frac{1}{\rho}\left[\sum_{v=1}^{\infty}\lambda_{v}^{2}\left|\left(P_{n}-P\right)\partial\ell_{\mu_{0,\rho}}\varphi_{v}\right|^{2}\right]^{1/2}\\
 & \leq & \frac{1}{\rho}\sum_{v=1}^{\infty}\lambda_{v}\left|\left(P_{n}-P\right)\partial\ell_{\mu_{0,\rho}}\varphi_{v}\right|.
\end{eqnarray*}
Given that $\mu_{0}\in\mathcal{H}^{K}$, there is a finite $B$ such
that $\mu_{0}\in{\rm int}\left(\mathcal{H}^{K}\left(B\right)\right)$
(this proves Point 1 in the lemma). By this remark, it follows that,
uniformly in $\rho\geq0$, there is an $\epsilon>0$ such that $\left|\mu_{0,\rho}\right|_{\mathcal{H}^{K}}\leq B-\epsilon$.
Hence, the maximal inequality of Theorem 3 in Doukhan et al. (1995)
implies that 
\begin{equation}
\mathbb{E}\sup_{\mu\in\mathcal{H}^{K}\left(B\right)}\left|\sqrt{n}\left(P_{n}-P\right)\partial\ell_{\mu}\varphi_{v}\right|\leq c_{1}\label{EQ_doukhanL1MaxInequ}
\end{equation}
for some finite constant $c_{1}$, for any $v\geq1$, because the
entropy integral (\ref{EQ_entropyIntegral}) is finite in virtue of
Lemma \ref{Lemma_entropy_dLh}. Define
\[
L_{n}:=\sum_{v=1}^{\infty}\lambda_{v}\sup_{\mu\in\mathcal{H}^{K}\left(B\right)}\left|\sqrt{n}\left(P_{n}-P\right)\partial\ell_{\mu_{0,\rho}}\varphi_{v}\right|.
\]
Given that the coefficients $\lambda_{v}$ are summable by Condition
\ref{Condition_kernel}, deduce from (\ref{EQ_doukhanL1MaxInequ})
that $\left(L_{n}\right)$ is a tight random sequence. Using the above
display, we have shown that (\ref{EQ_normLagrangeEstimatorDifference})
is bounded by $L_{n}/\left(\rho n^{1/2}\right)$. This proves Point
2 in the lemma. For any fixed $\epsilon>0$, we can choose $\rho=\rho_{n}:=L_{n}/\left(\epsilon n^{1/2}\right)$
so that $\left|\mu_{0,\rho}-\mu_{n,\rho}\right|_{\mathcal{H}^{K}}\leq\epsilon$
in probability. By the triangle inequality and the above calculations,
deduce that, in probability, 
\[
\left|\mu_{n,\rho}\right|_{\mathcal{H}^{K}}\leq\left|\mu_{0,\rho}\right|_{\mathcal{H}^{K}}+\left|\mu_{0,\rho}-\mu_{n,\rho}\right|_{\mathcal{H}^{K}}\leq B
\]
for $\rho=\rho_{n}$. By tightness of $L_{n}$, deduce that $\rho_{n}=O_{p}\left(n^{-1/2}\right)$.
Also, the first order condition for the sample estimator $\mu_{n,\rho}$
reads 
\begin{equation}
P_{n}\partial\ell_{\mu_{n,\rho}}h=-2\rho\left\langle \mu_{n,\rho},h\right\rangle _{\mathcal{H}^{K}}\leq2\rho\left|\mu_{n,\rho}\right|_{\mathcal{H}^{K}}\left|h\right|_{\mathcal{H}^{K}}\label{EQ_FOCLagrangianBound}
\end{equation}
for any $h\in\mathcal{H}^{K}\left(1\right)$. In consequence, $\sup_{h\in\mathcal{H}^{K}\left(1\right)}P_{n}\partial\ell_{\mu_{n,\rho}}h\leq2\rho\left|\mu_{n,\rho}\right|_{\mathcal{H}^{K}}$.
These calculations prove Point 3 in the lemma when $\rho=O_{p}\left(n^{-1/2}\right)$.

\end{proof}

The penalized objective function is increasing with $\rho$. In the
Lagrangian formulation of the constrained minimization, interest lies
in finding the smallest value of $\rho$ such that the constraint
is still satisfied. When $\rho$ equals such smallest value $\rho_{B,n}$,
we have $\mu_{n}=\mu_{n,\rho}$. From Lemma \ref{Lemma_lagrangeMultiplierSize}
deduce that $\rho_{B,n}=O_{p}\left(n^{-1/2}\right)$. Also, if $\mathcal{H}^{K}$
is infinite dimensional, the constraint needs to be binding so that
$\left|\mu_{n}\right|_{\mathcal{H}^{K}}=B$. Hence, if $\mu_{0}\in{\rm int}\left(\mathcal{H}^{K}\left(B\right)\right)$
there is an $\epsilon>0$ such that $\left|\mu_{0}\right|_{\mathcal{H}^{K}}=B-\epsilon$.
Then, we must have 
\begin{align*}
\left|\mu_{n}-\mu_{0}\right|_{\mathcal{H}^{K}}^{2} & =\left|\mu_{n}\right|_{\mathcal{H}^{K}}^{2}+\left|\mu_{0}\right|_{\mathcal{H}^{K}}^{2}-2\left\langle \mu_{n},\mu_{0}\right\rangle _{\mathcal{H}^{K}}\\
 & =\left(B^{2}+\left(B-\epsilon\right)^{2}-2\left\langle \mu_{n},\mu_{0}\right\rangle _{\mathcal{H}^{K}}\right).
\end{align*}
But $\left\langle \mu_{n},\mu_{0}\right\rangle _{\mathcal{H}^{K}}\leq\left|\mu_{n}\right|_{\mathcal{H}^{K}}\left|\mu_{0}\right|_{\mathcal{H}^{K}}\leq B\left(B-\epsilon\right)$.
Hence, the above display is greater or equal than 
\[
B^{2}+\left(B-\epsilon\right)^{2}-2B\left(B-\epsilon\right)\geq\epsilon^{2}.
\]
This means that $\mu_{n}$ cannot converge under the norm $\left|\cdot\right|_{\mathcal{H}^{K}}$.

The statement concerning approximate minimizers will be proved in
Section \ref{Section_approximateMinimizers}.

\subsection{Proof of Theorem \ref{Theorem_weakConvergence}}

It is convenient to introduce additional notation and concepts that
will be used in the remaining of the paper. By construction the minimizer
of the population objective function is $\mu_{0}\in\mathcal{H}^{K}\left(B\right)$.
Let $l^{\infty}\left(\mathcal{H}^{K}\right)$ be the space of uniformly
bounded functions on $\mathcal{H}^{K}$. Let $\Psi\left(\mu\right)$
be the operator in $l^{\infty}\left(\mathcal{H}^{K}\right)$ such
that $\Psi\left(\mu\right)h=P\partial\ell_{\mu}h$, $h\in\mathcal{H}^{K}$.
If the objective function is Fr\'{e}chet differentiable, the minimizer
of the objective function $P\ell_{\mu}$ in $\mathcal{H}^{K}\left(B\right)$
satisfies the variational inequality: $\Psi\left(\mu\right)h\geq0$
for any $h$ in the tangent cone of $\mathcal{H}^{K}\left(B\right)$
at $\mu_{0}$. This tangent cone is defined as $\lim\sup_{t\downarrow0}\left(\mathcal{H}^{K}\left(B\right)-\mu_{0}\right)/t$.
If $\mu_{0}$ is in the interior of $\mathcal{H}^{K}\left(B\right)$,
this tangent cone is the whole of $\mathcal{H}^{K}$. Hence by linearity
of the operator $\Psi\left(\mu\right)$, attention can be restricted
to $h\in\mathcal{H}^{K}\left(1\right)$. When $\mu_{0}\in{\rm int}\left(\mathcal{H}^{K}\left(B\right)\right)$,
it also holds that $\Psi\left(\mu_{0}\right)h=0$, for any $h\in\mathcal{H}^{K}\left(1\right)$.
Then, in the following calculations, $\Psi\left(\mu\right)$ can be
restricted to be in $l^{\infty}\left(\mathcal{H}^{K}\left(1\right)\right)$.
The empirical counterpart of $\Psi\left(\mu\right)$ is the operator
$\Psi_{n}\left(\mu\right)$ such that $\Psi_{n}\left(\mu\right)h=P_{n}\partial\ell_{\mu}h$.
Finally, write $\dot{\Psi}_{\mu_{0}}\left(\mu-\mu_{0}\right)$ for
the Fr\'{e}chet derivative of $\Psi\left(\mu\right)$ at $\mu_{0}$
tangentially to $\left(\mu-\mu_{0}\right)$, where $\mu,\mu_{0}\in\mathcal{H}^{K}\left(B\right)$.
Then, $\dot{\Psi}_{\mu_{0}}$ is an operator from $\mathcal{H}^{K}$
to $l^{\infty}\left(\mathcal{H}^{K}\right)$. As for $\Psi\left(\mu\right)$,
the operator $\dot{\Psi}_{\mu_{0}}\left(\mu-\mu_{0}\right)$ can be
restricted to be in $l^{\infty}\left(\mathcal{H}^{K}\left(1\right)\right)$.
These facts will be used without further notice in what follows. Most
of these concepts are reviewed in van der Vaart and Wellner (2000,
ch.3.3) where this same notation is used. 

Deduce that $P\ell_{\mu}$is Fr\'{e}chet differentiable and its derivative
is the map $\Psi\left(\mu\right)$. By the conditions of Theorem \ref{Theorem_weakConvergence},
$\mu_{0}\in{\rm int}\left(\mathcal{H}^{K}\left(B\right)\right)$,
hence by the first order conditions, $\Psi\left(\mu_{0}\right)h=0$
for any $h\in\mathcal{H}^{K}\left(1\right)$. By this remark, and
basic algebra, 
\begin{eqnarray}
\sqrt{n}\Psi_{n}\left(\mu_{n}\right) & = & \sqrt{n}\Psi_{n}\left(\mu_{0}\right)+\sqrt{n}\left[\Psi\left(\mu_{n}\right)-\Psi\left(\mu_{0}\right)\right]\nonumber \\
 &  & +\sqrt{n}\left[\Psi_{n}\left(\mu_{n}\right)-\Psi\left(\mu_{n}\right)\right]-\sqrt{n}\left[\Psi_{n}\left(\mu_{0}\right)-\Psi\left(\mu_{0}\right)\right].\label{EQ_basicIdentity}
\end{eqnarray}
To bound the last two terms, verify that 
\[
\sup_{h\in\mathcal{H}^{K}\left(1\right)}\sqrt{n}\left[\left(\Psi_{n}\left(\mu_{n}\right)-\Psi\left(\mu_{n}\right)\right)-\left(\Psi_{n}\left(\mu_{0}\right)-\Psi\left(\mu_{0}\right)\right)\right]h=o_{p}\left(1\right).
\]
This follows if (i) $\sqrt{n}\left(\Psi_{n}\left(\mu\right)-\Psi\left(\mu\right)\right)h$
, $\mu\in\mathcal{H}^{K}\left(B\right)$, $h\in\mathcal{H}^{K}\left(1\right)$,
converges weakly to a Gaussian process with continuous sample paths,
(ii) $\mathcal{H}^{K}\left(B\right)$ is compact under the uniform
norm, and (iii) $\mu_{n}$ is consistent for $\mu_{0}$ in $\left|\cdot\right|_{\infty}$.
Point (i) is satisfied by Lemma \ref{Lemma_gaussianConvergence},
which also controls the first term on the r.h.s. of (\ref{EQ_basicIdentity}).
Point (ii) is satisfied by Lemma \ref{Lemma_entropyOfL}. Point (iii)
is satisfied by Theorem \ref{Theorem_consistency}. Hence, by continuity
of the sample paths of the Gaussian process, as $\mu_{n}\rightarrow\mu_{0}$
in probability (using Point iii), the above display holds true 

To control the second term on the r.h.s. of (\ref{EQ_basicIdentity}),
note that the Fr\'{e}chet derivative of $\Psi\left(\mu\right)$ at
$\mu_{0}$ is the linear operator $\dot{\Psi}_{\mu_{0}}$ such that
$\dot{\Psi}_{\mu_{0}}\left(\mu-\mu_{0}\right)h=P\partial^{2}\ell_{\mu_{0}}\left(\mu-\mu_{0}\right)h$,
which can be shown to exist based on the remarks at the beginning
of Section \ref{Section_proofs}. For any $h\in\mathcal{H}^{K}\left(1\right)$,

\begin{equation}
\left|\left[\Psi\left(\mu_{n}\right)-\Psi\left(\mu_{0}\right)\right]h-\dot{\Psi}_{\mu_{0}}\left(\mu_{n}-\mu_{0}\right)h\right|\leq\sup_{t\in\left(0,1\right)}\left|P\partial^{3}\ell_{\mu_{0}+t\left(\mu_{n}-\mu_{0}\right)}\left(\mu_{n}-\mu_{0}\right)^{2}h\right|\label{EQ_secondFrechetDerivativeBound}
\end{equation}
using differentiability of the loss function and Taylor's theorem
in Banach spaces. By Condition \ref{ConditionProjectionTheorem},
and the fact that $h$ is uniformly bounded, the r.h.s. is a constant
multiple of $P\left(\mu-\mu_{0}\right)^{2}$. By Theorem \ref{Theorem_consistency}
this quantity is $O_{p}\left(n^{-\left(2\eta-1\right)/\left(2\eta\right)}\right)$.
Given that $\eta>1$, these calculations show that
\[
\sqrt{n}\left[\Psi\left(\mu_{n}\right)-\Psi\left(\mu_{0}\right)\right]=\sqrt{n}\dot{\Psi}_{\mu_{0}}\left(\mu_{n}-\mu_{0}\right)+o_{p}\left(1\right).
\]
In consequence, from (\ref{EQ_basicIdentity}) deduce that
\begin{eqnarray}
\sqrt{n}\Psi_{n}\left(\mu_{n}\right)-\sqrt{n}\Psi_{n}\left(\mu_{0}\right) & = & \sqrt{n}\left(\Psi\left(\mu_{n}\right)-\Psi\left(\mu_{0}\right)\right)+o_{p}\left(1\right)\nonumber \\
 & = & \sqrt{n}\dot{\Psi}_{\mu_{0}}\left(\mu_{n}-\mu_{0}\right)+o_{p}\left(1\right).\label{EQ_basicIdentityFinal}
\end{eqnarray}
By Lemma \ref{Lemma_gaussianConvergence}, $\sqrt{n}\Psi_{n}\left(\mu_{0}\right)=O_{p}\left(1\right)$.
For the moment, suppose that $\mu_{n}$ is the exact solution to the
minimization problem, i.e. as in (\ref{EQ_theEstimator}). Hence,
by Lemma \ref{Lemma_lagrangeMultiplierSize}, $\sup_{h\in\mathcal{H}^{K}\left(1\right)}\sqrt{n}\Psi_{n}\left(\mu_{n}\right)h=O_{p}\left(1\right)$,
implying that $\sup_{h\in\mathcal{H}^{K}\left(1\right)}\sqrt{n}\dot{\Psi}_{\mu_{0}}\left(\mu_{n}-\mu_{0}\right)h=O_{p}\left(1\right)$.
Finally, if $\sup_{h\in\mathcal{H}^{K}\left(1\right)}\sqrt{n}\Psi_{n}\left(\mu_{n}\right)h=o_{p}\left(1\right)$,
(\ref{EQ_basicIdentityFinal}) together with the previous displays
imply that $-\lim_{n}\sqrt{n}\left(\Psi_{n}\left(\mu_{0}\right)-\Psi\left(\mu_{0}\right)\right)=\lim_{n}\dot{\Psi}_{\mu_{0}}\sqrt{n}\left(\mu_{n}-\mu_{0}\right)$
in probability, where the l.h.s. has same distribution as the Gaussian
process $G$ given in the statement of the theorem. It remains to
show that if we use an approximate minimizer say $\nu_{n}$ to distinguish
it here from $\mu_{n}$ in (\ref{EQ_theEstimator}), the result still
holds. The lemma in the next section shows that this is true, hence
completing the proof of Theorem \ref{Theorem_weakConvergence}.

\subsection{Asymptotic Minimizers\label{Section_approximateMinimizers}}

The following collects results on asymptotic minimizers. It proves
the last statement in Theorem \ref{Theorem_consistency} and also
allows us to use such minimisers in the test. 

\begin{lemma}Let $\left(\epsilon_{n}\right)$ be an $o_{p}\left(1\right)$
sequence. Suppose that $\nu_{n}$ satisfies $P_{n}\ell_{\nu_{n}}\leq P_{n}\ell_{\mu_{n}}-O_{p}\left(\epsilon_{n}\right)$,
where $\mu_{n}$ is as in (\ref{EQ_theEstimator}). Also suppose that
$\nu_{n,\rho}$ satisfies $P_{n}\ell_{\nu_{n,\rho}}+\rho\left|\nu_{n,\rho}\right|_{\mathcal{H}^{K}}^{2}\leq P_{n}\ell_{\mu_{n,\rho}}+\rho\left|\mu_{n,\rho}\right|_{\mathcal{H}^{K}}^{2}-O_{p}\left(\rho\epsilon_{n}\right)$,
where $\mu_{n,\rho}$ is as in (\ref{EQ_penalizedEstimator.}) and
$\rho n^{1/2}\rightarrow\infty$. 
\begin{enumerate}
\item Under the conditions of Theorem \ref{Theorem_consistency}, $\left|\mu_{n}-\nu_{n}\right|_{\infty}=o_{p}\left(1\right)$,
$\left|\mu_{n}-\nu_{n}\right|_{2}=o_{p}\left(\epsilon_{n}\right)$
and $\left|\mu_{n,\rho}-\nu_{n\rho}\right|_{\mathcal{H}^{K}}=O_{p}\left(\epsilon_{n}\right)$
in probability, and there is a finite $B$ such that $\left|\nu_{n,\rho}\right|_{\mathcal{H}^{K}}\leq B$
eventually in probability.
\item If $\epsilon_{n}=o_{p}\left(n^{-1/2}\right)$, under the Conditions
of Theorem \ref{Theorem_weakConvergence}, $\sup_{h\in\mathcal{H}^{K}\left(1\right)}\left|\Psi_{n}\left(\mu_{n}\right)h-\Psi_{n}\left(\nu_{n}\right)h\right|=o_{p}\left(n^{-1/2}\right)$.
\end{enumerate}
\end{lemma}

\begin{proof}At first consider the penalized estimator. To this end,
follow the same steps in the proof of 5.14 in Theorem 5.9 of Steinwart
and Christmann (2008). Mutatis mutandis, the argument in their second
paragraph on page 174 gives
\begin{align*}
 & \left\langle \nu_{n,\rho}-\mu_{n,\rho},P_{n}\partial\ell_{\mu_{n,\rho}}\Phi+2\rho\mu_{n,\rho}\right\rangle _{\mathcal{H}^{K}}+\rho\left|\mu_{n,\rho}-\nu_{n,\rho}\right|_{\mathcal{H}^{K}}\\
\leq & P_{n}\ell_{\nu_{n,\rho}}+\rho\left|\nu_{n,\rho}\right|_{\mathcal{H}^{K}}^{2}-\left(P_{n}\ell_{\mu_{n,\rho}}+\rho\left|\mu_{n,\rho}\right|_{\mathcal{H}^{K}}^{2}\right).
\end{align*}
 Derivation of this display requires convexity of $L\left(z,t\right)$
w.r.t. $t$, which is the case by Condition \ref{Condition_Loss}.
By assumption, the r.h.s. is $O_{p}\left(\rho\epsilon_{n}\right)$.
Note that $\mu_{n,\rho}$ is the exact minimizer of the penalized
empirical risk. Hence, eq. (5.12) in Theorem 5.9 of Steinwart and
Christmann (2008) says that $\mu_{n,\rho}=-\left(2\rho\right)^{-1}P_{n}\partial\ell_{\mu_{n,\rho}}\Phi$
for any $\rho>0$, implying that the inner product in the display
is zero. By these remarks, deduce that the above display simplifies
to $\rho\left|\mu_{n,\rho}-\nu_{n,\rho}\right|_{\mathcal{H}^{K}}=O_{p}\left(\rho\epsilon_{n}\right)$.
Deduce that $\left|\mu_{n}-\nu_{n}\right|_{\mathcal{H}^{K}}=o_{p}\left(1\right)$
so that by the triangle inequality, and Lemma \ref{Lemma_lagrangeMultiplierSize},
$\left|\nu_{n,\rho}\right|_{\mathcal{H}^{K}}\leq B$ eventually, in
probability for some $B<\infty$.

Now, consider the constrained estimator. Conditioning on the data,
by definition of $\mu_{n}$, the variational inequality $P_{n}\partial\ell_{\mu_{n}}\left(\nu_{n}-\mu_{n}\right)\geq0$
holds because $\nu_{n}-\mu_{n}$ is an element of the tangent cone
of $\mathcal{H}^{K}\left(B\right)$ at $\mu_{n}$. Conditioning on
the data, by Taylor's theorem in Banach spaces, and the fact that
$\inf_{z\in\mathcal{Z},\left|t\right|\leq B}\partial^{2}L\left(z,t\right)>0$
by Condition \ref{Condition_Loss}, deduce that$\left|P_{n}\ell_{\nu_{n}}-P_{n}\ell_{\mu_{n}}\right|\gtrsim P_{n}\left(\mu_{n}-\nu_{n}\right)^{2}$.
By the conditions of the lemma, and the previous inequality deduce
that $P_{n}\left(\mu_{n}-\nu_{n}\right)^{2}=O_{p}\left(\epsilon_{n}\right)$.
The $L_{2}$ convergence is then turned into uniform using the same
argument used in the proof of Theorem \ref{Theorem_consistency}.
Now, conditioning on the data, by Fr\'{e}chet differentiability,
\begin{eqnarray*}
\left|\Psi_{n}\left(\mu_{n}\right)h-\Psi_{n}\left(\nu_{n}\right)h\right| & = & \left|P_{n}\partial\ell_{\nu_{n}}-P_{n}\partial\ell_{\mu_{n}}\right|\\
 & \leq & P_{n}\left|\sup_{\mu\in\mathcal{H}^{K}\left(B\right)}\partial^{2}\ell_{\mu}\left(\nu_{n}-\mu_{n}\right)h\right|.
\end{eqnarray*}
By Holder's inequality, and the fact that $h\in\mathcal{H}^{K}\left(1\right)$
is bounded, the r.h.s. is bounded by a constant multiple of 
\[
\left[P_{n}\left|\sup_{\mu\in\mathcal{H}^{K}\left(B\right)}\partial^{2}\ell_{\mu}\right|^{2}\right]^{1/2}\left[P_{n}\left(\nu_{n}-\mu_{n}\right)^{2}\right]^{1/2}\lesssim\left[P_{n}\Delta_{2}^{2}\right]^{1/2}\left[P_{n}\left(\nu_{n}-\mu_{n}\right)^{2}\right]^{1/2}.
\]
By Condition \ref{Condition_Loss}, deduce that $P_{n}\Delta_{2}^{2}=O_{p}\left(1\right)$
so that, by the previous calculations, the result follows. \end{proof}

\subsection{Proof of Results in Section \ref{Section_testingFunctionalRestrictions}
\label{Section_proofTheoremProjection}}

We use the operators $\Pi_{\rho}$, $\Pi_{n,\rho}$, $\tilde{\Pi}_{n,\rho}$
such that for any $h\in\mathcal{H}^{K}$:
\begin{eqnarray}
\Pi_{\rho}h & := & \arg\inf_{\nu\in\mathcal{R}_{0}}P\partial^{2}\ell_{\mu_{0}}\left(h-\nu\right)^{2}+\rho\left|\nu\right|_{\mathcal{H}^{K}}^{2}\text{ as in (\ref{EQ_projectionOperator})}\nonumber \\
\Pi_{n,\rho}h & := & \arg\inf_{\nu\in\mathcal{R}_{0}}P_{n}^{2}\partial\ell_{\mu_{n0}}\left(h-\nu\right)^{2}+\rho\left|\nu\right|_{\mathcal{H}^{K}}^{2}\,\text{ as in (\ref{EQ_empiricalProjection})}\nonumber \\
\tilde{\Pi}_{n,\rho}h & := & \arg\inf_{\nu\in\mathcal{R}_{0}}P\partial\ell_{\mu_{n0}}^{2}\left(h-\nu\right)^{2}+\rho\left|\nu\right|_{\mathcal{H}^{K}}^{2}.\label{EQ_notationNus}
\end{eqnarray}
To ease notation, we may write $\Pi_{n}=\Pi_{n,\rho}$ when $\rho=\rho_{n}$. 

The proof uses some preliminary results. In what follows, we shall
assume that $K=1$. This is to avoid notational complexities that
could obscure the main steps in the derivations. Because of additivity,
this is not restrictive as long as $K$ is bounded.

\begin{lemma}\label{Lemma_nuTildeNormBound}Suppose that $h\in\mathcal{H}^{K}\left(1\right)$.
Then, $\left|\Pi_{0}h\right|_{\mathcal{H}^{K}}\leq1$.\end{lemma}

\begin{proof}By construction, the linear projection $\Pi_{0}h$ satisfies
$\Pi_{0}h\in\mathcal{R}_{0}$ and $\Pi_{0}\left(h-\Pi_{0}h\right)=0$.
Hence, the space $\mathcal{H}^{K}$ is the direct sum of the set $\mathcal{R}_{0}$
and its complement in $\mathcal{H}^{K}$, say $\mathcal{R}_{0}^{c}$.
These sets are orthogonal. Note that we do not necessarily have $\mathcal{R}_{0}^{c}=\mathcal{R}_{1}$
unless the basis that spans $\mathcal{R}_{1}$ is already linearly
independent of $\mathcal{R}_{0}$. By Lemma 9.1 in van der Vaart and
van Zanten (2008) $\left|h\right|_{\mathcal{H}^{K}}=\left|\Pi_{0}h\right|_{\mathcal{R}_{0}}+\left|h-\Pi_{0}h\right|_{\mathcal{R}_{0}^{c}}$.
The norms are the ones induced by the inner products in the respective
spaces. But, $\left|\Pi_{0}h\right|_{\mathcal{R}_{0}^{c}}=0$. Hence,
we have that $\left|\Pi_{0}h\right|_{\mathcal{R}_{0}}=\left|\Pi_{0}h\right|_{\mathcal{H}^{K}}\leq\left|h\right|_{\mathcal{H}^{K}}=1$.\end{proof}

\begin{lemma}\label{Lemma_populationPenalizedProjConvergence2}Under
Condition \ref{ConditionProjectionTheorem}, if $\rho n^{\left(2\eta-1\right)/\left(4\eta\right)}\rightarrow\infty$,
then, $\sup_{h\in\mathcal{H}^{K}\left(1\right)}\left|\left(\Pi_{\rho}-\tilde{\Pi}_{n,\rho}\right)h\right|_{\mathcal{H}^{K}}\rightarrow0$
in probability. \end{lemma}

\begin{proof}Let $\tilde{P}$ and $\tilde{P}_{n}$ be finite positive
measures such that $d\tilde{P}/dP=\partial^{2}\ell_{\mu_{0}}$ and
$d\tilde{P}_{n}/dP=\partial^{2}\ell_{\mu_{0,n}}$. By Lemma \ref{Lemma_SteinwartChirstman},
using the same arguments as in the proof of Lemma \ref{Lemma_lagrangeMultiplierSize},
\begin{equation}
\left|\left(\Pi_{\rho}-\tilde{\Pi}_{n,\rho}\right)h\right|_{\mathcal{H}^{K}}\leq\frac{1}{\rho}\sum_{v=1}^{\infty}\lambda_{v}\left|\left(\tilde{P}_{n}-\tilde{P}\right)\left(h-\Pi_{\rho}h\right)\varphi_{v}\right|.\label{EQ_SteinwarthChristmanIneqPopProject}
\end{equation}
Taking derivatives, we bound each term in the absolute value by 
\[
\left|P\left(\partial^{2}\ell_{\mu_{0,n}}-\partial^{2}\ell_{\mu_{0}}\right)\left(h-\Pi_{\rho}h\right)\varphi_{v}\right|\leq\left|P\sup_{\mu\in\mathcal{H}^{K}\left(B\right)}\left|\partial^{3}\ell_{\mu}\right|\left(\mu_{0,n}-\mu_{0}\right)\left(h-\Pi_{\rho}h\right)\varphi_{v}\right|.
\]
By Lemma \ref{Lemma_nuTildeNormBound} and the definition of penalized
estimation, $\left|\Pi_{\rho}h\right|_{\mathcal{H}^{K}}\leq\left|\Pi_{0}h\right|_{\mathcal{H}^{K}}\leq1$
independently of $\rho$. Hence, $\left|h-\Pi_{\rho}h\right|_{\infty}\leq2$.
Moreover, the $\varphi_{v}$'s are uniformly bounded. Therefore, the
r.h.s. of the above display is bounded by a constant multiple of 
\[
\sqrt{P\left|\mu_{0,n}-\mu_{0}\right|^{2}}\sqrt{P\sup_{\mu\in\mathcal{H}^{K}\left(B\right)}\left|\partial^{3}\ell_{\mu}\right|^{2}}=\left|\mu_{0,n}-\mu_{0}\right|_{2}\sqrt{P\Delta_{3}^{2}}.
\]
The term $P\Delta_{3}^{2}$ is finite by Condition \ref{Condition_Loss}.
By Theorem \ref{Theorem_consistency}, we have that $\left|\mu_{0,n}-\mu_{0}\right|_{2}=O_{p}\left(n^{-\left(2\eta-1\right)/\left(4\eta\right)}\right)$.
Using the above display to bound (\ref{EQ_SteinwarthChristmanIneqPopProject}),
deduce that the lemma holds true if $\rho^{-1}\left(n^{-\left(2\eta-1\right)/\left(4\eta\right)}\right)=o_{p}\left(1\right)$
as stated in the lemma. Taking supremum w.r.t. $h\in\mathcal{H}^{K}\left(1\right)$
in the above steps, deduce that the result holds uniformly in $h\in\mathcal{H}^{K}\left(1\right)$.\end{proof}

\begin{lemma}\label{Lemma_samplePenalizedProjConvergence}Under Condition
\ref{ConditionProjectionTheorem}, we have that $\sup_{h\in\mathcal{H}^{K}\left(1\right)}\left|\left(\Pi_{n,\rho}-\tilde{\Pi}_{n,\rho}\right)h\right|_{\mathcal{H}^{K}}\rightarrow0$
in probability for any $\rho$ such that $\rho n^{1/2}\rightarrow\infty$
in probability.\end{lemma}

\begin{proof}Following the same steps as in the proof of Lemma \ref{Lemma_populationPenalizedProjConvergence2},
deduce that 
\[
\left|\left(\Pi_{n,\rho}-\tilde{\Pi}_{n,\rho}\right)h\right|_{\mathcal{H}^{K}}\leq\frac{1}{\rho}\sum_{v=1}^{\infty}\lambda_{v}\left|\left(P_{n}-P\right)\partial^{2}\ell_{\mu_{0,n}}\left(h-\tilde{\Pi}_{n,\rho}h\right)\varphi_{v}\right|.
\]
Each absolute value term on the r.h.s. is bounded in $L_{1}$ by 
\begin{align*}
 & \mathbb{E}\sup_{h\in\mathcal{H}^{K}\left(1\right),\mu\in\mathcal{H}^{K}\left(B\right),\nu\in\mathcal{H}^{K}\left(1\right)}\left|\left(P_{n}-P\right)\partial^{2}\ell_{\mu}\left(h-\nu\right)\varphi_{v}\right|\\
\leq & 2\mathbb{E}\sup_{h\in\mathcal{H}^{K}\left(1\right),\mu\in\mathcal{H}^{K}\left(B\right)}\left|\left(P_{n}-P\right)\partial^{2}\ell_{\mu}h\varphi_{v}\right|.
\end{align*}
Define the class of functions $\mathcal{F}:=\left\{ \partial^{2}\ell_{\mu}h\varphi_{k}:\mu\in\mathcal{H}^{K}\left(B\right),h\in\mathcal{H}^{K}\left(1\right)\right\} $.
Given that $\varphi_{v}$ is uniformly bounded, it can be deduced
from Lemma \ref{Lemma_entropy_dL2hh} that $\ln N_{\left[\right]}\left(\epsilon,\mathcal{F},\left|\cdot\right|_{p}\right)\lesssim\left(B/\epsilon\right)^{2/\left(2\eta-1\right)}+K\ln\left(\frac{B}{\epsilon}\right)$.
Hence, to complete the proof of the lemma, we can follow the same
exact steps as in the proof of Lemma \ref{Lemma_lagrangeMultiplierSize}.
\end{proof}

\begin{lemma}\label{Lemma_projectionsConvergence}Suppose Conditions
\ref{ConditionProjectionTheorem}, and $\mu_{0}\in{\rm int}\left(\mathcal{H}^{K}\left(B\right)\right)$.
Then, for $\rho$ such that $\rho n^{\left(2\eta-1\right)/\left(4\eta\right)}\rightarrow\infty$
in probability, and for $n\rightarrow0$, the following hold
\[
\sup_{h\in\mathcal{H}^{K}\left(1\right)}\left|\left(\Pi_{\rho}-\Pi_{n,\rho}\right)h\right|_{\mathcal{H}^{K}}=o_{p}\left(1\right),
\]
and

\begin{equation}
\sup_{h\in\mathcal{H}^{K}\left(1\right)}\left|\sqrt{n}\Psi_{n}\left(\mu_{n0}\right)\left(\Pi_{\rho}-\Pi_{n,\rho}\right)h\right|=o_{p}\left(1\right).\label{EQ_momentEquationNuisanceProjectionError}
\end{equation}

Finally, if $w\left(x\right)=\int_{\mathcal{Y}}\partial^{2}\ell_{\mu_{0}}\left(y,x\right)dP\left(y|x\right)$
is a known function, the above displays hold for $\rho$ such that
$\rho n^{1/2}\rightarrow\infty$ in probability.\end{lemma}

\begin{proof}By the triangle inequality 
\begin{equation}
\sup_{h\in\mathcal{H}^{K}\left(1\right)}\left|\left(\Pi_{\rho}-\Pi_{n,\rho}\right)h\right|_{\mathcal{H}^{K}}\leq\sup_{h\in\mathcal{H}^{K}\left(1\right)}\left|\left(\tilde{\Pi}_{n,\rho}-\Pi_{\rho}\right)\right|^{2}+\sup_{h\in\mathcal{H}^{K}\left(1\right)}\left|\left(\Pi_{n,\rho}-\tilde{\Pi}_{n,\rho}\right)h\right|^{2}.\label{EQ_projectionsDecompositions}
\end{equation}
The first statement in the lemma follows by showing that the r.h.s.
of the above is $o_{p}\left(1\right)$. This is the case by application
of Lemmas \ref{Lemma_populationPenalizedProjConvergence2} and \ref{Lemma_samplePenalizedProjConvergence}. 

By the established convergence in $\left|\cdot\right|_{\mathcal{H}^{K}}$,
for any $h\in\mathcal{H}^{K}\left(1\right)$, $\left|\left(\Pi_{\rho}-\Pi_{n,\rho}\right)h\right|_{\mathcal{H}^{K}}\leq\delta$
with probability going to one for any $\delta>0$. Therefore, to prove
(\ref{EQ_momentEquationNuisanceProjectionError}), we can restrict
attention to a bound for 
\[
\lim_{\delta\rightarrow0}\sup_{\left|h\right|_{\mathcal{H}^{K}}\leq\delta}\sqrt{n}\Psi_{n}\left(\mu_{n0}\right)h=\lim_{\delta\rightarrow0}\sup_{\left|h\right|_{\mathcal{H}^{K}}\leq\delta}\sqrt{n}P_{n}\partial\ell_{\mu_{n}}h.
\]
From Lemma \ref{Lemma_lagrangeMultiplierSize} and (\ref{EQ_FOCLagrangianBound})
in its proof, deduce that the above is bounded by 
\[
\lim_{\delta\rightarrow0}\sup_{\left|h\right|_{\mathcal{H}^{K}}\leq\delta}\left|h\right|_{\mathcal{H}^{K}}\times O_{p}\left(B\right).
\]
The first term in the product is zero so that (\ref{EQ_momentEquationNuisanceProjectionError})
holds.

Finally, to show the last statement in the lemma, note that it is
Lemma \ref{Lemma_populationPenalizedProjConvergence2} that puts an
additional constraint on $\rho$. However, saying that the function
$w$ is known, effectively amounts to saying that we can replace $\mu_{0,n}$
with $\mu_{0}$ in the definition of $\tilde{\Pi}_{n,\rho}$ in (\ref{EQ_notationNus}).
This means that $\tilde{\Pi}_{n,\rho}=\Pi_{\rho}$ so that the second
term in (\ref{EQ_projectionsDecompositions}) is exactly zero and
we do not need to use Lemma \ref{Lemma_populationPenalizedProjConvergence2}.
Therefore, $\rho$ is only constrained for the application of Lemma
\ref{Lemma_samplePenalizedProjConvergence}.\end{proof}

We also need to bound the distance between $\Pi_{\rho}$ and $\Pi_{0}$,
but this cannot be achieved in probability under the operator norm. 

\begin{lemma}\label{Lemma_penalizedProjL2Convergence}Under Condition
\ref{ConditionProjectionTheorem}, we have that $\sup_{h\in\mathcal{H}^{K}\left(1\right)}\tilde{P}\left(\Pi_{\rho}h-\Pi_{0}h\right)^{2}\leq\rho$.\end{lemma}

\begin{proof}At first show that 
\begin{equation}
\tilde{P}\left(\Pi_{\rho}h-\Pi_{0}h\right)^{2}\leq\tilde{P}\left(h-\Pi_{\rho}h\right)^{2}-\tilde{P}\left(h-\Pi_{0}h\right)^{2}.\label{EQ_basicQuadraticDistanceLossIneq}
\end{equation}
To see this, expand the r.h.s. of (\ref{EQ_basicQuadraticDistanceLossIneq}),
add and subtract $2\tilde{P}\left(\Pi_{0}h\right)^{2}$, and verify
that the r.h.s. of (\ref{EQ_basicQuadraticDistanceLossIneq}) is equal
to 
\[
-2\tilde{P}\Pi_{\rho}h+2\tilde{P}\Pi_{0}h\left(h-\Pi_{0}h\right)+\tilde{P}\left[\left(\Pi_{\rho}h\right)^{2}+\left(\Pi_{0}h\right)^{2}\right]
\]
However, $\Pi_{0}h$ is the projection of $h\in\mathcal{H}^{K}\left(1\right)$
onto the subspace $\mathcal{R}_{0}$. Hence, the middle term in the
above display is zero. Then, add and subtract $2\tilde{P}\Pi_{\rho}h\Pi_{0}h$
and rearrange to deduce that the above display is equal to 
\[
2\tilde{P}\Pi_{\rho}h\left(\Pi_{0}h-h\right)+\tilde{P}\left(\Pi_{\rho}h-\Pi_{0}h\right)^{2}.
\]
Given that $\Pi_{\rho}h\in\mathcal{R}_{0}$ and $\left(\Pi_{0}h-h\right)$
is orthogonal to elements in $\mathcal{R}_{0}$ by definition of the
projection $\Pi_{0}$, we have shown that (\ref{EQ_basicQuadraticDistanceLossIneq})
holds true. Following the proof of Corollary 5.18 in Steinwart and
Christmann (2008), 
\begin{align*}
\tilde{P}\left(h-\Pi_{\rho}h\right)^{2}-\tilde{P}\left(h-\Pi_{0}h\right)^{2} & \leq\left[\tilde{P}\left(h-\Pi_{\rho}h\right)^{2}+\rho\left|\Pi_{\rho}h\right|_{\mathcal{H}^{K}}^{2}\right]-\tilde{P}\left(h-\Pi_{0}h\right)^{2}\\
 & \leq\left[\tilde{P}\left(h-\Pi_{0}h\right)^{2}+\rho\left|\Pi_{0}h\right|_{\mathcal{H}^{K}}^{2}\right]-\tilde{P}\left(h-\Pi_{0}h\right)^{2}=\rho\left|\Pi_{0}h\right|_{\mathcal{H}^{K}}^{2}
\end{align*}
because $\left|\Pi_{\rho}h\right|_{\mathcal{H}^{K}}$ is positive
and $\Pi_{\rho}h$ is the minimizer of the penalized population loss
function (see (\ref{EQ_notationNus})). Now note that the r.h.s. of
the above display is bounded by $\rho$ using Lemma \ref{Lemma_nuTildeNormBound}
and (\ref{EQ_notationNus}). Hence the r.h.s. of (\ref{EQ_basicQuadraticDistanceLossIneq})
is bounded by $\rho$ uniformly in $h\in\mathcal{H}^{K}\left(1\right)$,
and the lemma is proved. \end{proof}

\begin{lemma}\label{Lemma_bilinearOperatorProjectionBound}Under
Condition \ref{ConditionProjectionTheorem}, we have that $\dot{\Psi}_{\mu_{0}}\sqrt{n}\left(\mu_{n0}-\mu_{0}\right)\left(\Pi_{0}h-\Pi_{\rho}h\right)=o_{p}\left(1\right)$
for any $\rho$ such that $n^{1/\left(2\eta\right)}\rho\rightarrow0$
in probability. \end{lemma} 

\begin{proof}By definition, 
\[
\dot{\Psi}_{\mu_{0}}\sqrt{n}\left(\mu_{n0}-\mu_{0}\right)\left(\Pi_{0}h-\Pi_{\rho}h\right)=P\partial^{2}\ell_{\mu_{0}}\sqrt{n}\left(\mu_{n0}-\mu_{0}\right)\left(\Pi_{0}h-\Pi_{\rho}h\right).
\]
By Holder inequality, the absolute value of the display is bounded
by 
\begin{equation}
\sqrt{n}\left[P\partial^{2}\ell_{\mu_{0}}\left(\mu_{n0}-\mu_{0}\right)^{2}\right]^{1/2}\left[P\partial^{2}\ell_{\mu_{0}}\left(\Pi_{0}h-\Pi_{\rho}h\right)^{2}\right]^{1/2}.\label{EQ_bilinearOperatorHolderBound}
\end{equation}
By Condition \ref{ConditionProjectionTheorem}, $\left|\partial^{2}\ell_{\mu_{0}}\right|_{\infty}<\infty$,
so that 
\[
\sqrt{n}\left[P\partial^{2}\ell_{\mu_{0}}\left(\mu_{n0}-\mu_{0}\right)^{2}\right]^{1/2}\lesssim\sqrt{n}\left[P\left(\mu_{n0}-\mu_{0}\right)^{2}\right]^{1/2}=O_{p}\left(n^{1/\left(4\eta\right)}\right)
\]
using Point 2 in Theorem \ref{Theorem_consistency}. Hence, by Lemma
\ref{Lemma_penalizedProjL2Convergence}, deduce that (\ref{EQ_bilinearOperatorHolderBound})
is bounded above by $O_{p}\left(n^{1/\left(4\eta\right)}\rho^{1/2}\right)=o_{p}\left(1\right)$
for the given choice of $\rho$.\end{proof}

\begin{lemma}\label{Lemma_weakConvergenceProjectedProcess}Suppose
that $\mu_{0}\in{\rm int}\left(\mathcal{H}^{K}\left(B\right)\right)$.
Under Condition \ref{ConditionProjectionTheorem}, if $\rho\rightarrow0$, 

\[
\sqrt{n}\Psi_{n}\left(\mu_{0}\right)\left(h-\Pi_{\rho}h\right)\rightarrow G\left(h-\Pi_{0}h\right),\,h\in\mathcal{H}^{K}\left(1\right),
\]
weakly, where the r.h.s. is a mean zero Gaussian process with covariance
function
\[
\Sigma\left(h,h'\right):=\mathbb{E}G\left(h-\Pi_{0}h\right)G\left(h'-\Pi_{0}h'\right)=P\partial\ell_{\mu_{0}}^{2}\left(h-\Pi_{0}h\right)\left(h'-\Pi_{0}h'\right)
\]
for any $h,h'\in\mathcal{H}^{K}\left(1\right)$. \end{lemma}

\begin{proof}Any Gaussian process $G\left(h\right)$ - not necessarily
the one in the lemma - is continuous w.r.t. the pseudo norm $d\left(h,h'\right)=\sqrt{\mathbb{E}\left|G\left(h\right)-G\left(h'\right)\right|^{2}}$
(Lemma 1.3.1 in Adler and Taylor, 2007). Hence, $d\left(h,h'\right)\rightarrow0$
implies that $G\left(h\right)-G\left(h'\right)\rightarrow0$ in probability.
By Lemma \ref{Lemma_nuTildeNormBound}, deduce that $\left(h-\Pi_{\rho}h\right)\in\mathcal{H}^{K}\left(2\right)$.
Hence, consider the Gaussian process $G\left(h\right)$ in the lemma
with $h\in\mathcal{H}^{K}\left(2\right)$. By direct calculation,
\begin{align}
d^{2}\left(h,h'\right) & =P\partial\ell_{\mu_{0}}^{2}h\left(h-h'\right)+P\partial\ell_{\mu_{0}}^{2}h'\left(h'-h\right)\lesssim P\partial\ell_{\mu_{0}}^{2}\left|h'-h\right|.\label{EQ_GaussianMetric}
\end{align}

Recall the notation $d\tilde{P}/dP=\partial^{2}\ell_{\mu_{0}}$. Multiply
and divide the r.h.s. of the display by $\sqrt{\partial^{2}\ell_{\mu_{0}}}$
and use Holder inequality, to deduce that (\ref{EQ_GaussianMetric})
is bounded above by 
\[
\sqrt{P\left(\frac{\partial\ell_{\mu_{0}}^{4}}{\partial^{2}\ell_{\mu_{0}}}\right)}\sqrt{P\partial^{2}\ell_{\mu_{0}}\left(h-h'\right)^{2}}\lesssim\sqrt{\tilde{P}\left(h-h'\right)^{2}}
\]
using the fact that $\partial^{2}\ell_{\mu_{0}}$ is bounded away
from zero and $\partial\ell_{\mu_{0}}^{4}$ is integrable. Hence,
to check continuity of the Gaussian process at arbitrary $h\rightarrow h'$,
we only need to consider $\tilde{P}\left(h-h'\right)^{2}\rightarrow0$.
By Theorem \ref{Theorem_weakConvergence} which also holds for any
$h\in\mathcal{H}^{K}\left(2\right)$, $\sqrt{n}\Psi_{n}\left(\mu_{0}\right)h$
converges weakly to a Gaussian process $G\left(h\right)$, $h\in\mathcal{H}^{K}\left(2\right)$.
Hence, $\sqrt{n}\Psi_{n}\left(\mu_{0}\right)\left(h-\Pi_{\rho}h\right)$
converges weakly to $G\left(h-\Pi_{0}h\right)$ if for any $h\in\mathcal{H}^{K}\left(1\right)$
\[
\sup_{h\in\mathcal{H}^{K}\left(1\right)}\lim_{\rho\rightarrow0}\left|G\left(h-\Pi_{\rho}h\right)-G\left(h-\Pi_{0}h\right)\right|=0
\]
in probability. The above display holds true if $\sup_{h\in\mathcal{H}^{K}\left(1\right)}\tilde{P}\left(\Pi_{0}h-\Pi_{\rho}h\right)^{2}\rightarrow0$
in probability as $\rho\rightarrow0$. This is the case by Lemma \ref{Lemma_penalizedProjL2Convergence}.\end{proof}

Furthermore we need to estimate the eigenvalues $\omega_{k}$ in order
to compute critical values.

\begin{lemma}\label{Lemma_covarianceConsistency}Under the conditions
of Theorem \ref{Theorem_projectionScoreTest}, if in Condition \ref{Condition_Loss},
$P\Delta_{1}^{p}\left(\Delta_{1}^{p}+\Delta_{2}^{p}\right)<\infty$,
the following hold in probability:\\
1. $\sup_{h,h'\in\mathcal{H}^{K}\left(1\right)}\left|\Sigma_{n}\left(h,h'\right)-\Sigma\left(h,h'\right)\right|\rightarrow0$;\\
2. $\sup_{k>0}\left|\omega_{nk}-\omega_{k}\right|\rightarrow0$, where
$\omega_{nk}$ and $\omega_{k}$ are the $k^{th}$ eigenvalues of
the covariance functions with entries $\Sigma_{n}\left(h,h'\right)$
and $\Sigma\left(h,h'\right)$, $h,h'\in\tilde{\mathcal{R}}_{1}$;
moreover, both the sample and population eigenvalues are summable.
In particular, for any $c\geq1+\sum_{k=1}^{\infty}\omega_{k}$, $\Pr\left(\sum_{k=1}^{\infty}\omega_{n,k}>c\right)=o\left(1\right)$.
\end{lemma}

\begin{proof}To show Point 1, use the triangle inequality to deduce
that 
\begin{eqnarray}
\left|\Sigma_{n}\left(h,h'\right)-\Sigma\left(h,h'\right)\right| & \leq & \left|\left(P_{n}-P\right)\left(\partial\ell_{\mu_{n0}}^{2}\right)\left(h-\Pi_{n}h\right)\left(h'-\Pi_{n}h'\right)\right|\nonumber \\
 &  & +\left|P\left(\partial\ell_{\mu_{n0}}^{2}-\partial\ell_{\mu_{0}}^{2}\right)\left(h-\Pi_{n}h\right)\left(h'-\Pi_{n}h'\right)\right|\nonumber \\
 &  & +\left|P\partial\ell_{\mu_{0}}^{2}\left(\Pi_{0}h-\Pi_{n}h\right)\left(h'-\Pi_{n}h'\right)\right|\nonumber \\
 &  & +\left|P\partial\ell_{\mu_{0}}^{2}\left(h-\Pi_{0}h\right)\left(\Pi_{0}h'-\Pi_{n}h'\right)\right|.\label{EQ_SigmaDifference}
\end{eqnarray}
It is sufficient to bound each term individually uniformly in $h,h'\in\mathcal{H}^{K}\left(1\right)$. 

To bound the first term in (\ref{EQ_SigmaDifference}), note that,
with probability going to one, $\left|h-\Pi_{n}h\right|_{\mathcal{H}^{K}}\leq2+\epsilon$
for any $\epsilon>0$ uniformly in $h\in\mathcal{H}^{K}\left(1\right)$,
by Lemmas \ref{Lemma_nuTildeNormBound} and \ref{Lemma_projectionsConvergence},
as $n\rightarrow\infty$. By this remark, to bound the first term
in probability, it is enough to bound $\left|\left(P_{n}-P\right)\partial\ell_{\mu}^{2}hh'\right|$uniformly
in $\mu\in\mathcal{H}^{K}\left(B\right)$ and $h,h'\in\mathcal{H}^{K}\left(2+\epsilon\right)$.
By Lemma \ref{Lemma_entropy_dL2hh} and the same maximal inequality
used to bound (\ref{EQ_doukhanL1MaxInequ}), deduce that this term
is $O_{p}\left(n^{-1/2}\right)$. 

To bound the second term in (\ref{EQ_SigmaDifference}), note that
$P\partial\ell_{\mu}^{2}$ is Fr\'{e}chet differentiable w.r.t. $\mu$.
To see this, one can use the same arguments as in the proof of Lemma
2.21 in Steinwart and Christmann (2008) as long as $P\sup_{\mu\in\mathcal{H}^{K}\left(B\right)}\left|\partial\ell_{\mu}\partial^{2}\ell_{\mu}\right|<\infty$,
which is the case by the assumptions in the lemma. Hence, 
\begin{eqnarray*}
 &  & \left|P\left(\partial\ell_{\mu_{n0}}^{2}-\partial\ell_{\mu_{0}}^{2}\right)\left(h-\Pi_{n}h\right)\left(h'-\Pi_{n}h'\right)\right|\\
 & \leq & 2\left|P\partial\ell_{\mu_{0}}\partial^{2}\ell_{\mu_{0}}\left(\mu_{n0}-\mu_{0}\right)\left(h-\Pi_{n}h\right)\left(h'-\Pi_{n}h'\right)\right|+o_{p}\left(1\right)
\end{eqnarray*}
using the fact that $\left|\mu_{n0}-\mu_{0}\right|_{\infty}=o_{p}\left(1\right)$
by Theorem \ref{Theorem_consistency}. By an application of Lemma
\ref{Lemma_projectionsConvergence}, again, a bound in probability
for the above is given by a bound for
\[
2\sup_{h,h'\in\mathcal{H}^{K}\left(2+\epsilon\right)}\left|P\partial\ell_{\mu_{0}}\partial^{2}\ell_{\mu_{0}}\left(\mu_{n0}-\mu_{0}\right)hh'\right|.
\]
By Theorem \ref{Theorem_consistency} and $P\left|\partial\ell_{\mu_{0}}\partial^{2}\ell_{\mu_{0}}\right|\leq P\Delta_{1}\Delta_{2}<\infty$,
implying that the above is $o_{p}\left(1\right)$. 

The third term in (\ref{EQ_SigmaDifference}) is bounded by 

\begin{eqnarray}
 &  & P\left|\partial\ell_{\mu_{0}}^{2}\left(\Pi_{0}h-\Pi_{n}h\right)\left(h'-\Pi_{n}h'\right)\right|\nonumber \\
 & \leq & \left|\Pi_{0}h-\Pi_{n}h\right|_{2}\times\left|\partial\ell_{\mu_{0}}^{2}\left(h'-\Pi_{n}h'\right)\right|_{2}.\label{EQ_ineqCovSecondTerm}
\end{eqnarray}
By the triangle inequality 
\[
\left|\Pi_{0}h-\Pi_{n}h\right|_{2}\leq\left|\Pi_{0}h-\Pi_{\rho}h\right|_{2}+\left|\Pi_{\rho}h-\Pi_{n}h\right|_{2}.
\]
By Lemma \ref{Lemma_penalizedProjL2Convergence}, and the fact that
$d\tilde{P}/dP=\partial^{2}\ell_{\mu_{0}}$ is bounded away from zero
and infinity, the first term on the r.h.s. goes to zero as $\rho\rightarrow0$.
By Lemma \ref{Lemma_projectionsConvergence}, the second term on the
r.h.s. is $o_{p}\left(1\right)$. Using the triangle inequality, the
second term in the product in (\ref{EQ_ineqCovSecondTerm}), is bounded
by $\left|\partial\ell_{\mu_{0}}^{2}\left(h'-\Pi_{0}h'\right)\right|_{2}+\left|\partial\ell_{\mu_{0}}^{2}\left(\Pi_{0}-\Pi_{n}\right)h'\right|_{2}$
and it is not difficult to see that this is $O_{p}\left(1\right)$.
These remarks imply that (\ref{EQ_ineqCovSecondTerm}) is $o_{p}\left(1\right)$.
The last term in (\ref{EQ_SigmaDifference}) is bounded similarly.
The uniform convergence of the covariance is proved because all the
bounds converge to zero uniformly in $h,h'\in\mathcal{H}^{K}\left(1\right)$. 

It remains to show Point 2. This follows from the inequality 
\[
\sup_{k>0}\left|\omega_{nk}-\omega_{k}\right|\leq\frac{1}{R}\sum_{h\in\mathcal{\tilde{R}}_{1}}\left|\Sigma_{n}\left(h,h\right)-\Sigma\left(h,h\right)\right|,
\]
which uses Lemma 4.2 in Bosq (2000) together with the fact that the
operator norm of a covariance function is bounded by the nuclear norm
( Bosq, 2000). Clearly, the r.h.s. is bounded by $\sup_{h\in\mathcal{H}^{K}\left(1\right)}\left|\Sigma_{n}\left(h,h\right)-\Sigma\left(h,h\right)\right|$
which converges to zero in probability. Finally, by definition of
the eigenvalues and eigenfunctions, $\Sigma\left(h,h\right)=\sum_{k=1}^{\infty}\omega_{k}\psi_{k}\left(h\right)\psi_{k}\left(h\right)$
so that 
\[
\frac{1}{R}\sum_{h\in\mathcal{\tilde{R}}_{1}}\Sigma\left(h,h\right)=\sum_{k=1}^{\infty}\omega_{k}\leq\sup_{h\in\mathcal{H}^{K}\left(1\right)}\Sigma\left(h,h\right)<\infty
\]
implying that the eigenvalues are summable. The sum of the sample
eigenvalues is equal to 
\begin{eqnarray*}
\frac{1}{R}\sum_{h\in\mathcal{\tilde{R}}_{1}}\Sigma_{n}\left(h,h\right) & \leq & \frac{1}{R}\sum_{h\in\mathcal{\tilde{R}}_{1}}\Sigma\left(h,h\right)+\frac{1}{R}\sum_{h\in\mathcal{\tilde{R}}_{1}}\left|\Sigma_{n}\left(h,h\right)-\Sigma\left(h,h\right)\right|\\
 & \leq & \sup_{h\in\mathcal{H}^{K}\left(1\right)}\Sigma\left(h,h\right)+\sup_{h\in\mathcal{H}^{K}\left(1\right)}\left|\Sigma_{n}\left(h,h\right)-\Sigma\left(h,h\right)\right|.
\end{eqnarray*}
As shown above, the first term on the r.h.s. is finite and the second
term converges to zero in probability. Hence, the sample eigenvalues
are summable in probability. In particular, from these remarks deduce
that for any $c<\infty$ such that $c\geq1+\sum_{k=1}^{\infty}\omega_{k}$,
$\Pr\left(\sum_{k=1}^{\infty}\omega_{n,k}>c\right)=o\left(1\right)$.
\end{proof}\\

To avoid repetition, the results in Section \ref{Section_testingFunctionalRestrictions}
are proved together. Mutatis mutandis, from (\ref{EQ_basicIdentityFinal}),
we have that 
\begin{equation}
\sqrt{n}\Psi_{n}\left(\mu_{n0}\right)=\sqrt{n}\Psi_{n}\left(\mu_{0}\right)+\dot{\Psi}_{\mu_{0}}\sqrt{n}\left(\mu_{n0}-\mu_{0}\right)+o_{p}\left(1\right).\label{EQ_nuisanceParmExpansion}
\end{equation}
Trivially, any $h\in\mathcal{H}^{K}\left(1\right)$ can be written
as $h=\Pi_{\rho}h+\left(h-\Pi_{\rho}h\right)$. By Lemma \ref{Lemma_nuTildeNormBound},
replace $h\in\mathcal{H}^{K}\left(1\right)$ with $\left(h-\Pi_{\rho}h\right)\in\mathcal{H}^{K}\left(2\right)$
in Lemma \ref{Lemma_gaussianConvergence}. Then, $\sqrt{n}\left(\Psi_{n}\left(\mu\right)-\Psi\left(\mu\right)\right)\left(h-\Pi_{\rho}h\right)$
for $\mu\in\mathcal{H}^{K}\left(B\right),\,h\in\mathcal{H}^{K}\left(1\right)$
converges weakly to a Gaussian process with a.s. continuous sample
paths. Therefore, (\ref{EQ_nuisanceParmExpansion}) also applies to
$\Psi_{n}\left(\mu\right)$ as an element in the space of uniformly
bounded functions on $\mathcal{H}^{K}\left(2\right)$. Now, for $\rho=\rho_{n}$,
\[
\sqrt{n}\Psi_{n}\left(\mu_{n0}\right)\left(h-\Pi_{n,\rho}h\right)=\sqrt{n}\Psi_{n}\left(\mu_{n0}\right)\left(h-\Pi_{\rho}h\right)+\sqrt{n}\Psi_{n}\left(\mu_{n0}\right)\left(\Pi_{\rho}-\Pi_{n,\rho}\right)h
\]
adding and subtracting $\sqrt{n}\Psi_{n}\left(\mu_{n0}\right)\Pi_{\rho}h$.
Using Lemma \ref{Lemma_projectionsConvergence}, this is equal to
\[
\sqrt{n}\Psi_{n}\left(\mu_{n0}\right)\left(h-\Pi_{\rho}h\right)+o_{p}\left(1\right)
\]
which by (\ref{EQ_nuisanceParmExpansion}) is equal to
\[
\sqrt{n}\Psi_{n}\left(\mu_{0}\right)\left(h-\Pi_{\rho}h\right)+\sqrt{n}\dot{\Psi}_{\mu_{0}}\left(\mu_{n0}-\mu_{0}\right)\left(h-\Pi_{\rho}h\right)+o_{p}\left(1\right).
\]
Using linearity, rewrite 
\begin{align*}
\dot{\Psi}_{\mu_{0}}\sqrt{n}\left(\mu_{n0}-\mu_{0}\right)\left(h-\Pi_{\rho}h\right)= & \dot{\Psi}_{\mu_{0}}\sqrt{n}\left(\mu_{n0}-\mu_{0}\right)\left(h-\Pi_{0}h\right)\\
 & +\dot{\Psi}_{\mu_{0}}\sqrt{n}\left(\mu_{n0}-\mu_{0}\right)\left(\Pi_{0}h-\Pi_{\rho}h\right).
\end{align*}
The first term on the r.h.s. is $P\partial^{2}\ell_{\mu_{0}}\sqrt{n}\left(\mu_{n0}-\mu_{0}\right)\left(h-\Pi_{0}h\right)$.
This is zero because $\left(\mu_{n0}-\mu_{0}\right)$ is in the linear
span of elements in $\mathcal{R}_{0}$, and $\left(h-\Pi_{0}\right)$
is orthogonal to any element in $\mathcal{R}_{0}$ (w.r.t. $\tilde{P}$
by (\ref{EQ_projectionOperator}) with $\rho=0$). Lemma \ref{EQ_bilinearOperatorHolderBound}
shows that the absolute value of the second term on the r.h.s. of
the display is $o_{p}\left(1\right)$. 

We deduce that the asymptotic distribution of $\sqrt{n}\Psi_{n}\left(\mu_{n0}\right)\left(h-\Pi_{n}h\right)$
is given by the one of $\sqrt{n}\Psi_{n}\left(\mu_{0}\right)\left(h-\Pi_{\rho}h\right)$
for $\rho\rightarrow0$ at a suitable rate. By Lemma \ref{Lemma_weakConvergenceProjectedProcess},
the latter converges weakly to a centered Gaussian process as in the
statement of Theorem \ref{Theorem_projectionScoreTest}. 

The test statistic $\hat{S}_{n}$ is the square of $\sqrt{n}\Psi_{n}\left(\mu_{n0}\right)\left(h-\Pi_{n,\rho}h\right)$
averaged over a finite number of functions $h$. This is the average
of squared asymptotically centered Gaussian random variables. By the
singular value decomposition, its distribution is given by $S$. The
distribution of the approximation to $S$ when the sample eigenvalues
are used is $\hat{S}$ as given in Theorem \ref{Theorem_feasibleTestStat}.
By the triangle inequality, 
\begin{equation}
\left|\hat{S}-S\right|\leq\sum_{k=1}^{\infty}\left|\omega_{nk}-\omega_{k}\right|N_{k}^{2}.\label{EQ_limitProcessApproximationError}
\end{equation}
The sum can be split into two parts, one for $k\leq L$ plus one for
$k>L$ where here $L$ is a positive integer. Hence, deduce that the
above is bounded by 
\[
L\sup_{k\leq L}\left|\omega_{nk}-\omega_{k}\right|N_{k}^{2}+\sum_{k>L}\left(\omega_{nk}+\omega_{k}\right)N_{k}^{2}
\]
Using Lemma \ref{Lemma_covarianceConsistency}, the first term is
$o_{p}\left(1\right)$ for any fixed integer $L$. By Lemma \ref{Lemma_covarianceConsistency},
again, there is a positive summable sequence $\left(a_{k}\right)_{k\geq1}$
such that, as $n\rightarrow\infty$, the event $\left\{ \sup_{k\geq1}\omega_{nk}a_{k}^{-1}=\infty\right\} $
is contained in the event $\left\{ \sum_{k=1}^{\infty}\omega_{n,k}>c\right\} $
for some some finite constant $c$. However, the latter event has
probability going to zero. Hence, the second term in the display is
bounded with probability going to one by 
\[
\left(\sup_{k>0}\omega_{nk}a_{k}^{-1}\right)\sum_{k>L}a_{k}N_{k}^{2}+\sum_{k>L}\omega_{k}N_{k}^{2},
\]
where $\sup_{k>0}\omega_{nk}a_{k}^{-1}=O_{p}\left(1\right)$. Given
that 
\[
\mathbb{E}\left[\sum_{k>L}a_{k}N_{k}^{2}+\sum_{k>L}\omega_{k}N_{k}^{2}\right]\lesssim\sum_{k>L}\left(a_{k}+\omega_{k}\right)\rightarrow0
\]
as $L\rightarrow\infty$, deduce that letting $L\rightarrow\infty$
slowly enough, (\ref{EQ_limitProcessApproximationError}) is $o_{p}\left(1\right)$.

\subsection{Proof of Theorem \ref{Theorem_absoluteLossWeakConvergence}}

The proof of Theorem \ref{Theorem_absoluteLossWeakConvergence} follows
the proof of Theorems \ref{Theorem_consistency} and \ref{Theorem_weakConvergence}.
The following lemmas can be used for this purpose. 

The estimator $\mu_{n}$ satisfies $P_{n}\partial\ell_{\mu_{n}}h=-2\rho\left\langle \mu_{n},h\right\rangle _{\mathcal{H}^{K}}$
for $h\in\mathcal{H}^{K}\left(1\right)$, and $\rho=\rho_{B,n}$where
$\partial\ell_{\mu}\left(z\right)=sign\left(y-\mu\left(x\right)\right)$.
This is a consequence of the following results together with the fact
that $\mu_{0}\in{\rm int}\left(\mathcal{H}^{K}\left(B\right)\right)$
minimizes the expected loss $P\ell_{\mu}$. Hence, first of all, find
an expression for $P\ell_{\mu}$. 

\begin{lemma}\label{Lemma_expectedAbsoluteLossRepresentation} Suppose
that $\mathbb{E}\left|Y\right|<\infty$ and $\mu$ has bounded range.
Then, for $\ell_{\mu}\left(z\right)=\left|y-\mu\left(x\right)\right|$,
\[
P\ell_{\mu}=\int\left[\int_{\mu\left(x\right)}^{\infty}\Pr\left(Y\geq s|x\right)ds+\int_{-\infty}^{\mu\left(x\right)}\Pr\left(Y<s|x\right)ds\right]dP\left(x\right),
\]
where $\Pr\left(Y\leq s|x\right)$ is the distribution of $Y$ conditional
on $X=x$, and similarly for $\Pr\left(Y\geq s|x\right)$.\end{lemma}

\begin{proof}Note that for any positive variable $a$, $a=\int_{0}^{a}ds=\int_{0}^{\infty}1_{\left\{ s\leq a\right\} }ds$
and $a=\int_{-\infty}^{0}1_{\left\{ s>-a\right\} }ds$. Since $\left|y-\mu\left(x\right)\right|=\left(y-\mu\left(x\right)\right)1_{\left\{ y-\mu\left(x\right)>0\right\} }+\left(\mu\left(x\right)-y\right)1_{\left\{ \mu\left(x\right)-y\geq0\right\} }$,
by the aforementioned remark 
\[
\left|y-\mu\left(x\right)\right|=\int_{0}^{\infty}1_{\left\{ s\leq y-\mu\left(x\right)\right\} }ds+\int_{-\infty}^{0}1_{\left\{ s>y-\mu\left(x\right)\right\} }ds.
\]
Write $P\left(y,x\right)=P\left(y|x\right)P\left(x\right)$ and take
expectation of the above to find that 
\begin{eqnarray*}
 &  & \int\int\left|y-\mu\left(x\right)\right|dP\left(y|x\right)dP\left(x\right)\\
 & = & \int\int\left[\int_{0}^{\infty}1_{\left\{ s\leq y-\mu\left(x\right)\right\} }+\int_{\infty}^{0}1_{\left\{ s>y-\mu\left(x\right)\right\} }\right]dsdP\left(y|x\right)dP\left(x\right).
\end{eqnarray*}
By the conditions of the lemma, the expectation is finite. Hence,
apply Fubini's Theorem to swap integration w.r.t. $s$ and $y$. Integrating
w.r.t. $y$, the the above display is equal to 
\[
\int\left[\int_{0}^{\infty}\Pr\left(Y\geq\mu\left(x\right)+s|x\right)ds+\int_{-\infty}^{0}\Pr\left(Y<\mu\left(x\right)+s|x\right)ds\right]dP\left(x\right).
\]
By change of variables, this is equal to the statement in the lemma.\end{proof}

The population loss function is Fr\'{e}chet differentiable and strictly
convex. This will also ensure uniqueness of $\mu_{0}$. 

\begin{lemma}\label{Lemma_frechetDerivativesAbsolutNorm}Under Condition
\ref{Condition_absoluteLoss}, the first three Fr\'{e}chet derivatives
of $P\ell_{\mu}$ are
\[
\partial\ell_{\mu}Ph=\int\left[2\Pr\left(Y\leq\mu\left(x\right)|x\right)-1\right]h\left(x\right)dP\left(x\right)
\]
\[
\partial^{2}P\ell_{\mu}h^{2}=2\int pdf\left(\mu\left(x\right)|x\right)h^{2}\left(x\right)dP\left(x\right),
\]
\[
\partial^{3}P\ell_{\mu}h^{3}=2\int pdf'\left(\mu\left(x\right)|x\right)h^{3}\left(x\right)dP\left(x\right),
\]
where $pdf'\left(y|x\right)$ is the first derivative of $pdf\left(y|x\right)$
w.r.t. $y$. Moreover, $P\partial^{2}\ell_{\mu}h^{2}\gtrsim Ph^{2}$
and 
\[
\sup_{\mu\in\mathcal{H}^{K}\left(B\right)}\left|\partial^{3}P\ell_{\mu}h\left(\mu-\mu_{0}\right)^{2}\right|\lesssim P\left(\mu-\mu_{0}\right)^{2}.
\]
\end{lemma}

\begin{proof}Define 
\[
I\left(t\right):=\int_{t}^{\infty}\Pr\left(Y\geq s|x\right)ds+\int_{-\infty}^{t}\Pr\left(Y<s|x\right)ds.
\]
By Lemma \ref{Lemma_expectedAbsoluteLossRepresentation}, $P\ell_{\mu}=PI\left(\mu\right)=\int I\left(\mu\left(x\right)\right)dP\left(x\right)$.
For any sequence $h_{n}\in\mathcal{H}^{K}\left(B\right)$ converging
to $0$ under the uniform norm,
\[
I'\left(\mu\left(x\right)\right):=\lim_{n\rightarrow\infty}\frac{I\left(\mu\left(x\right)+h_{n}\left(x\right)\right)-I\left(\mu\left(x\right)\right)}{h_{n}\left(x\right)}=-\Pr\left(Y\geq\mu\left(x\right)|x\right)+\Pr\left(Y<\mu\left(x\right)|x\right)
\]
by standard differentiation, as $\mu\left(x\right)$ and $h_{n}\left(x\right)$
are just real numbers for fixed $x$. By Condition \ref{Condition_absoluteLoss}
the probability is continuous so that the above can be written as
$2\Pr\left(Y\leq\mu\left(x\right)|x\right)-1$. It will be shown that
the Fr\'{e}chet derivative of $P\ell_{\mu}$ is $\int I'\left(\mu\left(x\right)\right)dP\left(x\right)$.
To see this, define 
\[
U_{n}\left(x\right):=\left|\frac{I\left(\mu\left(x\right)+h_{n}\left(x\right)\right)-I\left(\mu\left(x\right)\right)}{h_{n}\left(x\right)}-I'\left(\mu\left(x\right)\right)\right|,
\]
if $h_{n}\neq0$, otherwise, $U_{n}\left(x\right):=0$. By construction,
$U_{n}\left(x\right)$ converges to zero pointwise. By a similar method
as in the proof of Lemma 2.21 in Steinwart and Christmann (2008),
it is sufficient to show that the following converges to zero, 
\[
\lim_{\left|h_{n}\right|_{\infty}\rightarrow0}\frac{\left|P\ell_{\mu+h_{n}}-P\ell_{\mu}-\int I'\left(\mu\left(x\right)\right)dP\left(x\right)\right|}{\left|h_{n}\right|_{\infty}}\leq\int U_{n}\left(x\right)dP\left(x\right).
\]
The upper bound follows replacing $\left|h_{n}\right|_{\infty}$ with
$\left|h_{n}\left(x\right)\right|$ because $1/\left|h_{n}\right|_{\infty}\leq1/h_{n}\left(x\right)$.
The above goes to zero by dominated convergence if we find a dominating
function. To this end, the mean value theorem implies that for some
$t_{n}\in\left[0,1\right]$, 
\[
\left|\frac{I\left(\mu\left(x\right)+h_{n}\left(x\right)\right)-I\left(\mu\left(x\right)\right)}{h_{n}\left(x\right)}\right|\leq2\Pr\left(Y\leq\mu\left(x\right)+t_{n}h_{n}\left(x\right)|x\right)-1.
\]
The sequence $h_{n}\left(x\right)\in\mathcal{H}^{K}\left(B\right)$
is uniformly bounded by $B$. By monotonicity of probabilities, this
implies that the above display is bounded by $2\Pr\left(Y\leq\mu\left(x\right)+B|x\right)-1$
uniformly for any $h_{n}\left(x\right)$. This is also an upper bound
for $I'\left(\mu\left(x\right)\right)$. Hence, using the definition
of $U_{n}\left(x\right)$, it follows that 
\[
\left|U_{n}\left(x\right)\right|\leq4\Pr\left(Y\leq\mu\left(x\right)+B|x\right),
\]
where the r.h.s. is integrable. This implies the existence of a dominating
function and in consequence the first statement of the lemma. To show
that also $\partial P\ell_{\mu}$ is Fr\'{e}chet differentiable,
one can use a similar argument as above. Then, the Fr\'{e}chet derivative
of $I'\left(\mu\left(x\right)\right)$ can be shown to be $I''\left(\mu\left(x\right)\right)=2pdf\left(\mu\left(x\right)|x\right)$.
The third derivative is found similarly as long as $pdf\left(y|x\right)$
has bounded derivative. The final statements in the lemma follow by
the condition on the conditional density as stated in Condition \ref{Condition_absoluteLoss}.\end{proof}

The last statement in Lemma \ref{Lemma_frechetDerivativesAbsolutNorm}
establishes the bound in (\ref{EQ_secondFrechetDerivativeBound}).
For the central limit theorem, an estimate of complexity tailored
to the present case is needed. 

\begin{lemma}\label{Lemma_entropyL1_dLB}Suppose that Condition \ref{Condition_absoluteLoss}
holds. Consider $\ell_{\mu}\left(z\right)=\left|y-\mu\left(x\right)\right|$
(recall the notation $z=\left(y,x\right)$). For the set $\mathcal{F}:=\left\{ \partial\ell_{\mu}h:\mu\in\mathcal{H}^{K}\left(B\right),h\in\mathcal{H}^{K}\left(1\right)\right\} $,
the $L_{1}$ $\epsilon$-entropy with bracketing is 
\[
\ln N_{\left[\right]}\left(\epsilon,\mathcal{F},\left|\cdot\right|_{1}\right)\lesssim\epsilon^{-2/\left(2\eta-1\right)}+K\ln\left(\frac{B}{\epsilon}\right),
\]
\end{lemma}

\begin{proof} The first derivative of the absolute value of $x\in\mathbb{R}$
is $d\left|x\right|/dx=sign\left(x\right)=2\times1_{\left\{ x\geq0\right\} }-1$.
In consequence, it is sufficient to find brackets for sets of the
type $\left\{ y-\mu\left(x\right)\geq0\right\} $. For any measurable
sets $A$ and $A'$, $\mathbb{E}\left|1_{A}-1_{A'}\right|=\Pr\left(A\Delta A'\right)$;
here $\Delta$ is the symmetric difference. Hence, under Condition
\ref{Condition_absoluteLoss}, for $A=\left\{ Y-\mu\left(X\right)\geq0\right\} $,
$A'=\left\{ Y-\mu'\left(X\right)\geq0\right\} $, 
\begin{eqnarray*}
\Pr\left(A\Delta A'\right) & = & \Pr\left(Y-\mu\left(X\right)\geq0,Y-\mu'\left(X\right)<0\right)+\Pr\left(Y-\mu\left(X\right)<0,Y-\mu'\left(X\right)\geq0\right)\\
 & = & \Pr\left(\mu\left(X\right)\leq Y<\mu'\left(X\right)\right)+\Pr\left(\mu'\left(X\right)\leq Y<\mu\left(X\right)\right).
\end{eqnarray*}
 Using Condition \ref{Condition_absoluteLoss}, conditioning on $X=x$
and differentiating the first term on the r.h.s., w.r.t. $\mu'$,
\[
\Pr\left(\mu\left(x\right)\leq Y<\mu'\left(x\right)|x\right)\lesssim\left|\mu\left(x\right)-\mu'\left(x\right)\right|\leq\left|\mu-\mu'\right|_{\infty}.
\]
From the above two displays, deduce that $\mathbb{E}\left|1_{A}-1_{A'}\right|\lesssim\left|\mu-\mu'\right|_{\infty}$.
Hence, the $L_{1}$ bracketing number of $\left\{ \partial\ell_{\mu}:\mu\in\mathcal{H}^{K}\left(B\right)\right\} $
is bounded above by the $L_{\infty}$ bracketing number of $\mathcal{H}^{K}\left(B\right)$,
which is given in Lemma \ref{Lemma_entropyOfL}. The proof is completed
using the same remarks as at the end of the proof of Lemma \ref{Lemma_entropy_dLh}.\end{proof}

The following provides the weak convergence of $\left\{ P_{n}\partial\ell_{\mu}h:\mu\in\mathcal{H}^{K}\left(B\right),h\in\mathcal{H}^{K}\left(1\right)\right\} $. 

\begin{lemma}Let $\mu\in\mathcal{H}^{K}\left(B\right)$. Under Condition
\ref{Condition_absoluteLoss} 
\[
\sqrt{n}\left(P_{n}-P\right)\partial\ell_{\mu}h\rightarrow G\left(\partial\ell_{\mu}h\right),\,\mu\in\mathcal{H}^{K}\left(B\right),h\in\mathcal{H}^{K}\left(1\right)
\]
weakly, where $G\left(\partial\ell_{\mu}h\right)$ is a mean zero
Gaussian process indexed by 
\[
\partial\ell_{\mu}h\in\left\{ \partial\ell_{\mu}h:\mu\in\mathcal{H}^{K}\left(B\right),\,h\in\mathcal{H}^{K}\left(1\right)\right\} ,
\]
with a.s. continuous sample paths and covariance function
\[
\mathbb{E}G\left(\partial\ell_{\mu},h\right)G\left(\partial\ell_{\mu'},h'\right)=\sum_{j\in\mathbb{Z}}P_{1,j}\left(\partial\ell_{\mu}h,\partial\ell_{\mu}h'\right)
\]

\end{lemma}

\begin{proof}This just follows from an application of Theorem 8.4
in Rio (2000). That theorem applies to bounded classes of functions
$\mathcal{F}$ and stationary sequences that have summable beta mixing
coefficients. It requires that $\mathcal{F}$ satisfies 
\[
\int_{0}^{1}\sqrt{\epsilon^{-1}\ln N_{\left[\right]}\left(\epsilon,\mathcal{F},\left|\cdot\right|_{1}\right)}d\epsilon<\infty.
\]
When $\mathcal{F}$ is as in Lemma \ref{Lemma_entropyL1_dLB}, this
is the case when $\eta>3/2$, as stated in Condition \ref{Condition_absoluteLoss}.\end{proof}

Using the above results, Theorem \ref{Theorem_absoluteLossWeakConvergence}
can be proved following step by step the proofs of Theorems \ref{Theorem_consistency}
and \ref{Theorem_weakConvergence}.

\subsection{Proof of Theorem \ref{Theorem_algoErrorBound}}

Only here, for typographical reasons, write $\ell\left(\mu\right)$
instead of $\ell_{\mu}$ and similarly for $\partial\ell\left(\mu\right)$.
Let 
\[
h_{m}:=\arg\min_{h\in\mathcal{L}^{K}\left(B\right)}P_{n}\partial\ell\left(F_{m-1}\right)h.
\]
Note that by linearity, and the $l_{1}$ constraint imposed by $\mathcal{L}^{K}\left(B\right)$,
the minimum is obtained by an additive function with $K-1$ additive
components equal to zero and a non-zero one in $\mathcal{H}$ with
norm $\left|\cdot\right|_{\mathcal{H}}$ equal to $B$, i.e. $Bf^{s\left(m\right)}$,
where $f^{s\left(m\right)}\in\mathcal{H}\left(1\right)$. Define,
\[
D\left(F_{m-1}\right):=\min_{h\in\mathcal{L}^{K}\left(B\right)}P_{n}\partial\ell\left(F_{m-1}\right)\left(h-F_{m-1}\right),
\]
so that for any $\mu\in\mathcal{L}^{K}\left(B\right)$, 
\begin{equation}
P_{n}\ell\left(\mu\right)-P_{n}\ell\left(F_{m-1}\right)\geq D\left(F_{m-1}\right)\label{EQ_gapEquation}
\end{equation}
by convexity. For $m\geq1$, define $\tilde{\tau}{}_{m}=2/\left(m+2\right)$
if $\tau_{m}$ is chosen by line search, or $\tilde{\tau}_{m}=\tau_{m}$
if $\tau_{m}=m^{-1}$. By convexity, again, 
\[
P_{n}\ell\left(F_{m}\right)=\inf_{\tau\in\left[0,1\right]}P_{n}\ell\left(F_{m-1}+\tau\left(h_{m}-F_{m-1}\right)\right)\leq P_{n}\ell\left(F_{m-1}\right)+P_{n}\partial\ell\left(F_{m-1}\right)\left(h_{m}-F_{m-1}\right)\tilde{\tau}_{m}+\frac{Q}{2}\tilde{\tau}_{m}^{2}
\]
where 
\[
Q:=\sup_{h,F\in\mathcal{L}^{K}\left(B\right),\tau\in\left[0,1\right]}\frac{2}{\tau^{2}}\left[P_{n}\ell\left(F+\tau\left(h-F\right)\right)-P_{n}\ell\left(F\right)-\tau P_{n}\partial\ell\left(F\right)\left(h-F\right)\right].
\]
The above two displays together with the definition of $D\left(F_{m-1}\right)=P_{n}\partial\ell\left(F_{m-1}\right)\left(h_{m}-F_{m-1}\right)$
imply that for any $\mu\in\mathcal{L}^{K}\left(B\right)$,
\begin{eqnarray*}
P_{n}\ell\left(F_{m}\right) & \leq & P_{n}\ell\left(F_{m-1}\right)+\tilde{\tau}_{m}D\left(F_{m-1}\right)+\frac{Q}{2}\tilde{\tau}_{m}^{2}\\
 & \leq & P_{n}\ell\left(F_{m-1}\right)+\rho_{m}\left(P_{n}\ell\left(\mu\right)-P_{n}\ell\left(F_{m-1}\right)\right)+\frac{Q}{2}\rho_{m}^{2},
\end{eqnarray*}
where the second inequality follows from (\ref{EQ_gapEquation}).
Subtracting $P_{n}\ell\left(\mu\right)$ on both sides and rearranging,
we have the following recursion 
\[
P_{n}\ell\left(F_{m}\right)-P_{n}\ell\left(\mu\right)\leq\left(1-\rho_{m}\right)\left(P_{n}\ell\left(F_{m-1}\right)-P_{n}\ell\left(\mu\right)\right)+\frac{Q}{2}\rho_{m}^{2}.
\]
The result is proved by bounding the above recursion for the two different
choices of $\tilde{\tau}_{m}$. When, $\tilde{\tau}_{m}=2/\left(m+1\right)$,
the proof of Theorem 1 in Jaggi (2013) bounds the recursion by $2Q/\left(m+2\right)$.
If $\rho_{m}=m^{-1}$, then, Lemma 2 in Sancetta (2016) bounds the
recursion by $4Q\ln\left(1+m\right)/m$ for any $m>0$. It remains
to bound $Q$. By Taylor expansion of $\ell\left(F+\tau\left(h-F\right)\right)$
at $\tau=0$, 
\[
\ell\left(F+\tau\left(h-F\right)\right)=\ell\left(F\right)+\partial\ell\left(F\right)\left(h-F\right)\tau+\frac{\partial^{2}\ell\left(F+t\left(h-F\right)\right)\left(h-F\right)^{2}\tau^{2}}{2}
\]
for some $t\in\left[0,1\right]$. It follows that 
\begin{eqnarray*}
Q & \leq & \max_{t\in\left[0,1\right]}\sup_{h,F\in\mathcal{L}^{K}\left(B\right),\tau\in\left[0,1\right]}P_{n}\partial^{2}\ell\left(F+t\left(h-F\right)\right)\left(h-F\right)^{2}\\
 & \leq & 4B^{2}\sup_{\left|t\right|<B}P_{n}d^{2}L\left(\cdot,t\right)/dt^{2}.
\end{eqnarray*}

\subsection{Proof of Lemma \ref{Lemma_propertiesOfL}\label{Section_proofRemainingResults}}

Point 1 is obvious. By the relation between the $l_{1}$ and $l_{2}$
norms (derived using Minkowski and the Cauchy-Schwarz inequality),
$\left|\mu\right|_{\mathcal{H}^{K}}\leq\left|\mu\right|_{\mathcal{L}^{K}}\leq\sqrt{K}\left|\mu\right|_{\mathcal{H}^{K}}$
and this shows the inclusion in Point 2. Every subspace of a Hilbert
space is uniformly convex, hence, Point 3 is proved. By the RKHS property
$f^{\left(k\right)}\left(x^{\left(k\right)}\right)=\left\langle f^{\left(k\right)},C\left(\cdot,x^{\left(k\right)}\right)\right\rangle _{\mathcal{H}}$,
for $\mu\left(x\right)=\sum_{k=1}^{K}f^{\left(k\right)}\left(x^{\left(k\right)}\right)$,
\[
\left|\mu\left(x\right)\right|=\left|\sum_{k=1}^{K}\left\langle f^{\left(k\right)},C\left(\cdot,x^{\left(k\right)}\right)\right\rangle _{\mathcal{H}}\right|.
\]
When $\mu\in\mathcal{L}^{K}\left(B\right)$, by the Cauchy-Schwarz
inequality and the RKHS property again, the display is bounded by
\[
\sum_{k=1}^{K}\left|f^{\left(k\right)}\right|_{\mathcal{H}}\left|C\left(\cdot,x^{\left(k\right)}\right)\right|_{\mathcal{H}}\leq\sum_{k=1}^{K}\left|f^{\left(k\right)}\right|_{\mathcal{H}}\sqrt{C\left(x^{\left(k\right)},x^{\left(k\right)}\right)}\leq cB,
\]
using the definition of $\mathcal{L}^{K}\left(B\right)$ and the assumed
bound on the kernel. The above two displays imply that $\left|\mu\right|_{\infty}\leq cB$.
This shows the result for $p=\infty$. For any $p\in[1,\infty)$,
use the trivial inequality $P\left|\mu\right|^{p}\leq\left|\mu\right|_{\infty}^{p}P\left(\mathcal{X}^{K}\right)=\left|\mu\right|_{\infty}^{p}$.
When $\mu\in\mathcal{H}^{K}\left(B\right)$, by Cauchy-Schwarz inequality
it is simple to deduce from the above two displays that $\left|\mu\right|_{\infty}\leq c\sqrt{K}B$.
These remarks prove Point 4.

\section{Appendix 2: Additional Details\label{Section_ComputationalAppendix}}

\subsection{Additional Details for Examples in Section \ref{Section_examples}
\label{Section_detailsExamples}}

The function $H_{V}$ in Example \ref{Example_H0LinearityComponent}
is 
\[
H_{V}\left(\cdot,\cdot\right)=\int_{0}^{1}G_{V}\left(\cdot,u\right)G_{V}\left(\cdot,u\right)du\text{ with }G_{V}\left(r,u\right):=\max\left\{ \frac{\left(r-u\right)^{V-1}}{\left(V-1\right)!},0\right\} ,
\]
where $r,u\in\left[0,1\right]$ (Wahba, 1990, p.7-8). 

To see that Example \ref{Example_H0AdditiveH1General} fits in the
framework of the paper, let $\mathcal{X}^{K+1}=\prod_{k=1}^{K+1}\mathcal{X}^{\left(k\right)}$
and $\mathcal{H}^{K+1}=\bigoplus_{k=1}^{K+1}\mathcal{H}^{\left(k\right)}$.
Here, $\mathcal{H}^{\left(k\right)}$ is a RKHS on $\mathcal{X}^{\left(k\right)}=\left[0,1\right]$
for $k\leq K$, and $\mathcal{H}^{\left(K+1\right)}$ is a RKHS on
$\mathcal{X}^{\left(K+1\right)}=\left[0,1\right]^{K}$. (Formally,
this would also require us to define $X=\left(X^{\left(1\right)},...,X^{\left(K\right)},X^{\left(K+1\right)}\right)$
with $X^{\left(K+1\right)}=\left(X^{\left(1\right)},...,X^{\left(K\right)}\right)$.)
As the example shows, in practice, we can directly consider $\mathcal{R}_{0}$
and $\mathcal{R}_{1}$ rather than $\mathcal{H}^{K+1}$. 

\subsection{Selection of $B$ and Variable Screening\label{Section_BSelectionVariableScreening}}

The parameter $B$ uniquely identifies the Lagrange multiplier $\rho_{B,n}$
in the penalized version of the optimization problem (\ref{EQ_theEstimator})
(see Example \ref{Example_sqaureErrorLoss}) If the loss is non-negative,
$\left|\mu_{n}\right|_{\mathcal{H}^{K}}^{2}\leq\rho_{B,n}^{-1}P_{n}\ell_{0}$
(e.g., Steinwart and Christmann, 2008, Section 5.1). The exact same
argument holds for $\mathcal{L}^{K}\left(B\right)$ in place of $\mathcal{H}^{K}\left(B\right)$.
When the constraint $\mu\in\mathcal{L}^{K}\left(B\right)$ is considered,
the solution via the greedy algorithm in Section \ref{Section_estimationAlgo}
allows us to keep track of the iterations at which selected variables
are included. Variables included at the early stage of the algorithm
will be clearly included even when $B$ is increased. Hence, exploration
for the purpose of feature selection (using the constraint $\mu\in\mathcal{L}^{K}\left(B\right)$)
can be carried out using a large $B$ to reduce the computational
burden. 

Selection of $B$ is usually based on cross-validation or penalized
estimation, where the penalty estimates the ``degrees of freedom''. 

\subsection{Additional Representations for Practical Computations\label{Section_additionalRepresentation}}

At each iteration $m$, the Lagrange multiplier in (\ref{EQ_objectiveGreedyUnivariate})
is derived as follows. Define

\[
\rho_{m}^{\left(k\right)}:=\left[\frac{1}{4}\sum_{i,j=1}^{n}\frac{\partial\ell_{F_{m-1}}\left(Z_{i}\right)}{n}\frac{\partial\ell_{F_{m-1}}\left(Z_{j}\right)}{n}C\left(X_{i}^{\left(k\right)},X_{j}^{\left(k\right)}\right)\right]^{1/2}.
\]
Let $s\left(m\right)=\arg\max_{k\leq K}\rho_{m}^{\left(k\right)}$,
and $\rho_{m}=\max_{k\leq K}\rho_{m}^{\left(k\right)}$. In consequence,
$f^{s\left(m\right)}=-\frac{1}{2\rho_{m}}P_{n}\partial\ell_{F_{m-1}}\Phi^{s\left(m\right)}$. 

Recall that $C_{\mathcal{H}^{K}}\left(s,t\right)=\sum_{k=1}^{K}C\left(s^{\left(k\right)},t^{\left(k\right)}\right)$
and the series representation (\ref{EQ_covSeriesRepresentation}).
Suppose that (\ref{EQ_covSeriesRepresentation}) holds for $C$ with
a finite number of $V$ terms (either exactly, or approximately, by
Condition \ref{Condition_kernel}), and that it is known. Then, we
can reduce the computational burden from $O\left(n^{2}\right)$ to
$O\left(nV\right)$. Recall that for notational simplicity we use
the same covariance kernel for all $k=1,2,...,K$ so that $\lambda_{v}\varphi_{v}$
in (\ref{EQ_covSeriesRepresentation}) does not depend on $k$. In
this case, define 
\[
a_{v}^{\left(k\right)}=\sum_{i=1}^{n}\frac{\partial\ell_{F_{m-1}}\left(Z_{i}\right)\lambda_{v}\varphi_{v}\left(X_{i}^{\left(k\right)}\right)}{n}
\]
and note that $\rho_{m}^{\left(k\right)}=\frac{1}{2}\sqrt{\sum_{v=1}^{V}\left|a_{v}^{\left(k\right)}\right|^{2}}$,
so that 
\[
f^{s\left(m\right)}\left(x^{s\left(m\right)}\right)=-\sum_{v=1}^{V}\frac{a_{v}^{s\left(m\right)}}{\sqrt{\sum_{v=1}^{V}\left|a_{v}^{s\left(m\right)}\right|^{2}}}\lambda_{v}\varphi_{v}\left(x^{s\left(m\right)}\right)
\]
and as before $s\left(m\right)=\arg\max_{k\leq K}\rho_{m}^{\left(k\right)}$.
This representation is suited for large sample $n$ when ready access
memory (RAM) is limited. For example, if $n=O\left(10^{6}\right)$,
which is not uncommon for high frequency applications, naive matrix
methods to estimate a regression function under RKHS constraints requires
to store a $n\times n$ matrix of doubles, which is equivalent to
about half a terabyte of RAM.

\subsection{The Beta Mixing Condition\label{Section_appendixRemarksConditions} }

To avoid ambiguities, recall the definition of beta mixing. Suppose
that $\left(Z_{i}\right)_{i\in\mathbb{Z}}$ is a stationary sequence
of random variables and let $\sigma\left(Z_{i}:i\leq0\right)$, $\sigma\left(Z_{i}:i\geq k\right)$
be the sigma algebra generated by $\left\{ Z_{i}:i\leq0\right\} $
and $\left\{ Z_{i}:i\geq k\right\} $, respectively, for integer $k$.
For any $k\geq1$, the beta mixing coefficient $\beta\left(k\right)$
for $\left(Z_{i}\right)_{i\in\mathbb{Z}}$ is 
\[
\beta\left(k\right):=\mathbb{E}\sup_{A\in\sigma\left(Z_{i}:i\geq k\right)}\left|\Pr\left(A|\sigma\left(W_{i}:i\leq0\right)\right)-\Pr\left(A\right)\right|
\]
(see Rio, 2000, section 1.6, for other equivalent definitions). In
the context of Condition \ref{Condition_Dependence}, set $Z_{i}=\left(Y_{i},X_{i}\right)$.
Condition \ref{Condition_Dependence} is a convenient technical restriction
and is satisfied by any model that can be written as a Markov chain
with smooth conditional distribution (e.g., Doukhan, 1995, for a review;
Basrak et al., 2002, for GARCH). Models with innovations that do not
have a smooth density function may not be covered (e.g. Rosenblatt,
1980, Andrews, 1984, Bradley, 1986, for a well known example). 

\begin{example}\label{Example_linearRegression}Suppose that $Y_{i}=\sum_{k=1}^{K}f^{\left(k\right)}\left(X_{i}^{\left(k\right)}\right)+\varepsilon_{i}$,
where the sequence of $\varepsilon_{i}$'s and $X_{i}$'s are independent.
By independence, deduce that the mixing coefficients of $\left\{ \left(Y_{i},X_{i}\right):i\in\mathbb{Z}\right\} $
are bounded by the sum of the mixing coefficients of the $\varepsilon_{i}$'s
and $X_{i}$'s (e.g., Bradley, 2005, Theorem 5.1). Suppose that the
$\varepsilon_{i}$'s and $X_{i}$'s are positive recurrent Markov
chains with innovations with continuous conditional density function.
Under additional mild regularity conditions, Condition \ref{Condition_Dependence}
is satisfied with geometric mixing rates (e.g., Mokkadem, 1987, Doukhan,
1995, section 2.4.0.1). Examples include GARCH and others, as in the
aforementioned references.\end{example}

\begin{example}Suppose that $Y_{i}\in\left\{ -1,1\right\} $. A classification
model based on the regressors $X_{i}$ can be generated via the random
utility model 
\[
Y_{i}^{*}=\mu\left(X_{i}\right)+\varepsilon_{i}
\]
where $Y_{i}=sign\left(Y_{i}^{*}\right)$. The sigma algebra generated
by $\left\{ Y_{i}:i\in\mathcal{A}\right\} $ for any subset $\mathcal{A}$
of the integers is contained in the sigma algebra generated by $\left\{ Y_{i}^{*}:i\in\mathcal{A}\right\} $.
Hence, for the errors $\varepsilon_{i}$'s and the $X_{i}$'s as in
Example \ref{Example_linearRegression}, the variables are beta mixing
with geometric mixing rate.\end{example}

\subsection{Explicit Expressions Implied by the Compact Notation\label{Examples_notation}}

The following examples should be nearly exhaustive in making quantities
more readable: 

\[
P_{n}\partial\ell_{\mu_{n0}}\left(h-\Pi_{n}h\right)=\frac{1}{n}\sum_{i=1}^{n}\partial\ell_{\mu_{n0}}\left(\left(Y_{i},X_{i}\right)\right)\left(h\left(X_{i}\right)-\Pi_{n}h\left(X_{i}\right)\right),
\]
\[
P_{n}\partial\ell_{\mu_{n0}}^{2}\left(h-\nu\right)^{2}=\frac{1}{n}\sum_{i=1}^{n}\partial\ell_{\mu_{n0}}^{2}\left(\left(Y_{i},X_{i}\right)\right)\left(h\left(X_{i}\right)-\nu\left(X_{i}\right)\right).
\]
\begin{align*}
 & P_{n}\partial\ell_{\mu_{n0}}^{2}\left(h-\Pi_{n}h\right)\left(h'-\Pi_{n}h'\right)\\
= & \frac{1}{n}\sum_{i=1}^{n}\partial\ell_{\mu_{n0}}^{2}\left(\left(Y_{i},X_{i}\right)\right)\left(h\left(X_{i}\right)-\Pi_{n}h\left(X_{i}\right)\right)\left(h'\left(X_{i}\right)-\Pi_{n}h'\left(X_{i}\right)\right).
\end{align*}
Moreover, 
\[
P\partial^{2}\ell_{\mu_{0}}\left(\mu_{n}-\mu_{0}\right)h=\int_{\mathcal{X}^{K}}\int_{\mathcal{Y}}\partial^{2}\ell_{\mu_{0}}\left(\left(y,x\right)\right)\left(\mu_{n}\left(x\right)-\mu_{0}\left(x\right)\right)h\left(x\right)dP\left(y,x\right)
\]
where $P$ is the law of $Z=\left(Y,X\right)$.

\section{Appendix 4: Additional Numerical Details\label{Section_AppendixNumericalResults}}

The following tables report more simulation results. The column heading
``No $\Pi$'' means that no correction was used in estimating the
test statistic and the covariance function: instead of using $\left(h-\Pi_{n}h\right)$
we just use $h$, which is the naive estimator in the presence of
a nuisance parameter. The column heading ``Size'' stands for the
nominal size and the simulated frequency of rejection should be close
to this when the null is true. 

\begin{table}[H]
\caption{High Dimensional Estimation. Simulated frequency of rejections when
$n=100$, $K=10$ and the true model is Lin3. Hence, Lin1 and Lin2
should be rejected. }
\begin{tabular}{ccccccccccc}
 &  &  & \multicolumn{2}{c}{Lin1} & \multicolumn{2}{c}{Lin2} & \multicolumn{2}{c}{Lin3} & \multicolumn{2}{c}{LinAll}\tabularnewline
$\rho$ & $\sigma_{\mu/\varepsilon}^{2}$ & Size  & No $\Pi$ & $\Pi$ & No $\Pi$ & $\Pi$ & No $\Pi$ & $\Pi$ & No $\Pi$ & $\Pi$\tabularnewline
0 & 1 & 0.10 & 1.00 & 1.00 & 0.99 & 0.99 & 0.08 & 0.12 & 0.05 & 0.13\tabularnewline
0 & 1 & 0.05 & 1.00 & 1.00 & 0.96 & 0.98 & 0.03 & 0.06 & 0.02 & 0.07\tabularnewline
0 & 0.2 & 0.10 & 0.71 & 0.78 & 0.44 & 0.50 & 0.08 & 0.12 & 0.05 & 0.13\tabularnewline
0 & 0.2 & 0.05 & 0.54 & 0.66 & 0.25 & 0.36 & 0.04 & 0.06 & 0.02 & 0.07\tabularnewline
0.75 & 1 & 0.10 & 0.91 & 0.95 & 0.21 & 0.31 & 0.05 & 0.10 & 0.07 & 0.14\tabularnewline
0.75 & 1 & 0.05 & 0.80 & 0.90 & 0.12 & 0.20 & 0.02 & 0.05 & 0.03 & 0.07\tabularnewline
0.75 & 0.2 & 0.10 & 0.28 & 0.39 & 0.08 & 0.14 & 0.05 & 0.10 & 0.07 & 0.14\tabularnewline
0.75 & 0.2 & 0.05 & 0.16 & 0.25 & 0.03 & 0.06 & 0.02 & 0.05 & 0.03 & 0.07\tabularnewline
\end{tabular}
\end{table}

\begin{table}
\caption{High Dimensional Estimation. Simulated frequency of rejections when
$n=1000$, $K=10$ and the true model is Lin3. Hence, Lin1 and Lin2
should be rejected. }
\begin{tabular}{ccccccccccc}
 &  &  & \multicolumn{2}{c}{Lin1} & \multicolumn{2}{c}{Lin2} & \multicolumn{2}{c}{Lin3} & \multicolumn{2}{c}{LinAll}\tabularnewline
$\rho$ & $\sigma_{\mu/\varepsilon}^{2}$ & Size  & No $\Pi$ & $\Pi$ & No $\Pi$ & $\Pi$ & No $\Pi$ & $\Pi$ & No $\Pi$ & $\Pi$\tabularnewline
0 & 1 & 0.10 & 1.00 & 1.00 & 1.00 & 1.00 & 0.09 & 0.11 & 0.07 & 0.10\tabularnewline
0 & 1 & 0.05 & 1.00 & 1.00 & 1.00 & 1.00 & 0.05 & 0.05 & 0.04 & 0.06\tabularnewline
0 & 0.2 & 0.10 & 1.00 & 1.00 & 1.00 & 1.00 & 0.10 & 0.11 & 0.07 & 0.11\tabularnewline
0 & 0.2 & 0.05 & 1.00 & 1.00 & 1.00 & 1.00 & 0.05 & 0.05 & 0.04 & 0.06\tabularnewline
0.75 & 1 & 0.10 & 1.00 & 1.00 & 1.00 & 1.00 & 0.06 & 0.09 & 0.06 & 0.10\tabularnewline
0.75 & 1 & 0.05 & 1.00 & 1.00 & 1.00 & 1.00 & 0.03 & 0.05 & 0.03 & 0.04\tabularnewline
0.75 & 0.2 & 0.10 & 1.00 & 1.00 & 0.48 & 0.60 & 0.06 & 0.09 & 0.06 & 0.10\tabularnewline
0.75 & 0.2 & 0.05 & 1.00 & 1.00 & 0.32 & 0.45 & 0.03 & 0.05 & 0.03 & 0.04\tabularnewline
\end{tabular}
\end{table}

\begin{table}
\caption{High Dimensional Estimation. Simulated frequency of rejections when
$n=100$, $K=10$ and the true model is LinAll. Hence, Lin1, Lin2,
and Lin3 should be rejected. }

\begin{tabular}{ccccccccccc}
 &  &  & \multicolumn{2}{c}{Lin1} & \multicolumn{2}{c}{Lin2} & \multicolumn{2}{c}{Lin3} & \multicolumn{2}{c}{LinAll}\tabularnewline
$\rho$ & $\sigma_{\mu/\varepsilon}^{2}$ & Size  & No $\Pi$ & $\Pi$ & No $\Pi$ & $\Pi$ & No $\Pi$ & $\Pi$ & No $\Pi$ & $\Pi$\tabularnewline
0 & 1 & 0.10 & 1.00 & 1.00 & 1.00 & 1.00 & 1.00 & 1.00 & 0.06 & 0.11\tabularnewline
0 & 1 & 0.05 & 1.00 & 1.00 & 1.00 & 1.00 & 1.00 & 1.00 & 0.03 & 0.05\tabularnewline
0 & 0.2 & 0.10 & 0.84 & 0.88 & 0.80 & 0.85 & 0.77 & 0.82 & 0.05 & 0.13\tabularnewline
0 & 0.2 & 0.05 & 0.71 & 0.77 & 0.65 & 0.75 & 0.62 & 0.72 & 0.02 & 0.07\tabularnewline
0.75 & 1 & 0.10 & 1.00 & 1.00 & 1.00 & 1.00 & 1.00 & 1.00 & 0.07 & 0.14\tabularnewline
0.75 & 1 & 0.05 & 1.00 & 1.00 & 1.00 & 1.00 & 1.00 & 1.00 & 0.03 & 0.07\tabularnewline
0.75 & 0.2 & 0.10 & 0.95 & 0.97 & 0.89 & 0.93 & 0.80 & 0.87 & 0.07 & 0.14\tabularnewline
0.75 & 0.2 & 0.05 & 0.89 & 0.94 & 0.80 & 0.89 & 0.66 & 0.80 & 0.03 & 0.07\tabularnewline
\end{tabular}
\end{table}

\begin{table}
\caption{High Dimensional Estimation. Simulated frequency of rejections when
$n=1000$, $K=10$ and the true model is LinAll. Hence, Lin1, Lin2,
and Lin3 should be rejected. }
\begin{tabular}{ccccccccccc}
 &  &  & \multicolumn{2}{c}{Lin1} & \multicolumn{2}{c}{Lin2} & \multicolumn{2}{c}{Lin3} & \multicolumn{2}{c}{LinAll}\tabularnewline
$\rho$ & $\sigma_{\mu/\varepsilon}^{2}$ & Size  & No $\Pi$ & $\Pi$ & No $\Pi$ & $\Pi$ & No $\Pi$ & $\Pi$ & No $\Pi$ & $\Pi$\tabularnewline
0 & 1 & 0.10 & 1.00 & 1.00 & 1.00 & 1.00 & 1.00 & 1.00 & 0.49 & 0.08\tabularnewline
0 & 1 & 0.05 & 1.00 & 1.00 & 1.00 & 1.00 & 1.00 & 1.00 & 0.23 & 0.05\tabularnewline
0 & 0.2 & 0.10 & 1.00 & 1.00 & 1.00 & 1.00 & 1.00 & 1.00 & 0.08 & 0.10\tabularnewline
0 & 0.2 & 0.05 & 1.00 & 1.00 & 1.00 & 1.00 & 1.00 & 1.00 & 0.04 & 0.06\tabularnewline
0.75 & 1 & 0.10 & 1.00 & 1.00 & 1.00 & 1.00 & 1.00 & 1.00 & 0.07 & 0.10\tabularnewline
0.75 & 1 & 0.05 & 1.00 & 1.00 & 1.00 & 1.00 & 1.00 & 1.00 & 0.03 & 0.05\tabularnewline
0.75 & 0.2 & 0.10 & 1.00 & 1.00 & 1.00 & 1.00 & 1.00 & 1.00 & 0.06 & 0.10\tabularnewline
0.75 & 0.2 & 0.05 & 1.00 & 1.00 & 1.00 & 1.00 & 1.00 & 1.00 & 0.03 & 0.05\tabularnewline
\end{tabular}
\end{table}
\begin{table}
\caption{High Dimensional Estimation. Simulated frequency of rejections when
$n=100$, $K=10$ and the true model is NonLin. Hence, Lin1, Lin2,
Lin3, and LinAll should be rejected. }
\begin{tabular}{ccccccccccccc}
 &  &  & \multicolumn{2}{c}{Lin1} & \multicolumn{2}{c}{Lin2} & \multicolumn{2}{c}{Lin3} & \multicolumn{2}{c}{LinAll} & \multicolumn{2}{c}{LinPoly}\tabularnewline
$\rho$ & $\sigma_{\mu/\varepsilon}^{2}$ & Size  & No $\Pi$ & $\Pi$ & No $\Pi$ & $\Pi$ & No $\Pi$ & $\Pi$ & No $\Pi$ & $\Pi$ & No $\Pi$ & $\Pi$\tabularnewline
0 & 1 & 0.10 & 0.97 & 0.94 & 0.97 & 0.94 & 0.97 & 0.95 & 0.59 & 0.61 & 0.03 & 0.15\tabularnewline
0 & 1 & 0.05 & 0.94 & 0.91 & 0.95 & 0.92 & 0.95 & 0.92 & 0.54 & 0.50 & 0.01 & 0.09\tabularnewline
0 & 0.2 & 0.10 & 0.71 & 0.72 & 0.72 & 0.74 & 0.73 & 0.75 & 0.30 & 0.31 & 0.04 & 0.12\tabularnewline
0 & 0.2 & 0.05 & 0.57 & 0.62 & 0.58 & 0.64 & 0.61 & 0.68 & 0.23 & 0.23 & 0.01 & 0.06\tabularnewline
0.75 & 1 & 0.10 & 0.94 & 0.95 & 0.93 & 0.93 & 0.85 & 0.88 & 0.61 & 0.61 & 0.02 & 0.13\tabularnewline
0.75 & 1 & 0.05 & 0.89 & 0.92 & 0.86 & 0.90 & 0.69 & 0.79 & 0.54 & 0.52 & 0.01 & 0.06\tabularnewline
0.75 & 0.2 & 0.10 & 0.70 & 0.77 & 0.57 & 0.66 & 0.32 & 0.39 & 0.33 & 0.31 & 0.01 & 0.14\tabularnewline
0.75 & 0.2 & 0.05 & 0.55 & 0.68 & 0.39 & 0.53 & 0.17 & 0.28 & 0.25 & 0.24 & 0.00 & 0.08\tabularnewline
\end{tabular}
\end{table}
\begin{table}
\caption{High Dimensional Estimation. Simulated frequency of rejections when
$n=1000$, $K=10$ and the true model is NonLin. Hence, Lin1, Lin2,
Lin3, and LinAll should be rejected. }

\begin{tabular}{ccccccccccccc}
 &  &  & \multicolumn{2}{c}{Lin1} & \multicolumn{2}{c}{Lin2} & \multicolumn{2}{c}{Lin3} & \multicolumn{2}{c}{LinAll} & \multicolumn{2}{c}{LinPoly}\tabularnewline
$\rho$ & $\sigma_{\mu/\varepsilon}^{2}$ & Size  & No $\Pi$ & $\Pi$ & No $\Pi$ & $\Pi$ & No $\Pi$ & $\Pi$ & No $\Pi$ & $\Pi$ & No $\Pi$ & $\Pi$\tabularnewline
0 & 1 & 0.10 & 1 & 1 & 1 & 1 & 1 & 1 & 0.92 & 0.91 & 0.03 & 0.1\tabularnewline
0 & 1 & 0.05 & 1 & 1 & 1 & 1 & 1 & 1 & 0.9 & 0.88 & 0.02 & 0.05\tabularnewline
0 & 0.2 & 0.10 & 1 & 0.99 & 1 & 0.99 & 1 & 0.99 & 0.75 & 0.72 & 0.04 & 0.1\tabularnewline
0 & 0.2 & 0.05 & 1 & 0.98 & 1 & 0.98 & 1 & 0.99 & 0.7 & 0.65 & 0.01 & 0.05\tabularnewline
0.75 & 1 & 0.10 & 1 & 1 & 1 & 1 & 1 & 1 & 0.89 & 0.87 & 0.02 & 0.09\tabularnewline
0.75 & 1 & 0.05 & 1 & 1 & 1 & 1 & 1 & 1 & 0.87 & 0.84 & 0 & 0.05\tabularnewline
0.75 & 0.2 & 0.10 & 1 & 0.99 & 1 & 0.99 & 0.99 & 0.99 & 0.75 & 0.73 & 0.02 & 0.11\tabularnewline
0.75 & 0.2 & 0.05 & 1 & 0.99 & 1 & 0.99 & 0.99 & 0.98 & 0.72 & 0.67 & 0.01 & 0.06\tabularnewline
\end{tabular}
\end{table}

\begin{table}
\caption{Infinite Dimensional Estimation. Simulated frequency of rejections
for $n=100$, and various combinations of signal to noise $\sigma_{\mu/\varepsilon}^{2}$,
and variables correlation $\rho$. The true model is linear in the
first variable and nonlinear in the second variable. Hence, LinAll
should be rejected.}
\label{Table_simulationSubsetCov-1}

\begin{tabular}{cccccc}
 &  & \multicolumn{2}{c}{Lin1NonLin} & \multicolumn{2}{c}{LinAll}\tabularnewline
$\left(\sigma_{\mu/\varepsilon}^{2},\rho\right)$ & Size & No $\Pi$ & $\Pi$ & No $\Pi$ & $\Pi$\tabularnewline
\hline 
$\left(1,0\right)$ & 0.10 & \textsf{0.00} & \textsf{0.12} & \textsf{0.00} & \textsf{0.91}\tabularnewline
$\left(1,0\right)$ & 0.05 & \textsf{0.00} & \textsf{0.06} & \textsf{0.00} & \textsf{0.84}\tabularnewline
$\left(.2,0\right)$ & 0.10 & \textsf{0.00} & \textsf{0.12} & \textsf{0.00} & \textsf{0.42}\tabularnewline
$\left(.2,0\right)$ & 0.05 & \textsf{0.00} & \textsf{0.06} & \textsf{0.00} & \textsf{0.31}\tabularnewline
$\left(1,.75\right)$ & 0.10 & \textsf{0.00} & \textsf{0.11} & \textsf{0.00} & \textsf{0.97}\tabularnewline
$\left(1,.75\right)$ & 0.05 & \textsf{0.00} & \textsf{0.06} & \textsf{0.00} & \textsf{0.93}\tabularnewline
$\left(.2,.75\right)$ & 0.10 & \textsf{0.00} & \textsf{0.11} & \textsf{0.00} & \textsf{0.51}\tabularnewline
$\left(.2,.75\right)$ & 0.05 & \textsf{0.00} & \textsf{0.06} & \textsf{0.00} & \textsf{0.38}\tabularnewline
\end{tabular}
\end{table}
\begin{table}
\caption{Infinite Dimensional Estimation. Simulated frequency of rejections
for $n=1000$, and various combinations of signal to noise $\sigma_{\mu/\varepsilon}^{2}$,
and variables correlation $\rho$. The true model is linear in the
first variable and nonlinear in the second variable. Hence, LinAll
should be rejected.}
\label{Table_simulationSubsetCov-1-1}

\begin{tabular}{cccccc}
 &  & \multicolumn{2}{c}{Lin1NonLin} & \multicolumn{2}{c}{LinAll}\tabularnewline
$\left(\sigma_{\mu/\varepsilon}^{2},\rho\right)$ & Size & No $\Pi$ & $\Pi$ & No $\Pi$ & $\Pi$\tabularnewline
\hline 
$\left(1,0\right)$ & 0.10 & \textsf{0.00} & \textsf{0.09} & \textsf{0.99} & \textsf{1.00}\tabularnewline
$\left(1,0\right)$ & 0.05 & \textsf{0.00} & \textsf{0.04} & \textsf{0.82} & \textsf{1.00}\tabularnewline
$\left(.2,0\right)$ & 0.10 & \textsf{0.00} & \textsf{0.09} & \textsf{0.00} & \textsf{1.00}\tabularnewline
$\left(.2,0\right)$ & 0.05 & \textsf{0.00} & \textsf{0.04} & \textsf{0.00} & \textsf{1.00}\tabularnewline
$\left(1,.75\right)$ & 0.10 & \textsf{0.00} & \textsf{0.11} & \textsf{1.00} & \textsf{1.00}\tabularnewline
$\left(1,.75\right)$ & 0.05 & \textsf{0.00} & \textsf{0.03} & \textsf{1.00} & \textsf{1.00}\tabularnewline
$\left(.2,.75\right)$ & 0.10 & \textsf{0.00} & \textsf{0.11} & \textsf{0.17} & \textsf{1.00}\tabularnewline
$\left(.2,.75\right)$ & 0.05 & \textsf{0.00} & \textsf{0.03} & \textsf{0.02} & \textsf{1.00}\tabularnewline
\end{tabular}
\end{table}
\begin{table}
\caption{Infinite Dimensional Estimation.Simulated frequency of rejections
for $n=5000$, and various combinations of signal to noise $\sigma_{\mu/\varepsilon}^{2}$,
and variables correlation $\rho$. The true model is linear in the
first variable and nonlinear in the second variable. The column heading
``Size'' stands for the nominal size.}
\label{Table_simulationSubsetCov-1-2}

\begin{tabular}{cccccc}
 &  & \multicolumn{2}{c}{Lin1NonLin} & \multicolumn{2}{c}{LinAll}\tabularnewline
$\left(\sigma_{\mu/\varepsilon}^{2},\rho\right)$ & Size & No $\Pi$ & $\Pi$ & No $\Pi$ & $\Pi$\tabularnewline
\hline 
$\left(1,0\right)$ & 0.10 & \textsf{0} & \textsf{0.12} & \textsf{1} & \textsf{1}\tabularnewline
$\left(1,0\right)$ & 0.05 & \textsf{0} & \textsf{0.05} & \textsf{1} & \textsf{1}\tabularnewline
$\left(.2,0\right)$ & 0.10 & \textsf{0} & \textsf{0.12} & \textsf{1} & \textsf{1}\tabularnewline
$\left(.2,0\right)$ & 0.05 & \textsf{0} & \textsf{0.05} & \textsf{0.97} & \textsf{1}\tabularnewline
$\left(1,.75\right)$ & 0.10 & \textsf{0} & \textsf{0.11} & \textsf{1} & \textsf{1}\tabularnewline
$\left(1,.75\right)$ & 0.05 & \textsf{0} & \textsf{0.06} & \textsf{1} & \textsf{1}\tabularnewline
$\left(.2,.75\right)$ & 0.10 & \textsf{0} & \textsf{0.11} & \textsf{1} & \textsf{1}\tabularnewline
$\left(.2,.75\right)$ & 0.05 & \textsf{0} & \textsf{0.06} & \textsf{1} & \textsf{1}\tabularnewline
\end{tabular}
\end{table}

\pagebreak{}


\begin{thebibliography}{10}
\bibitem{key-4}Andrews, D.W.K. (1994) Asymptotics for Semiparametric
Econometric Models Via Stochastic Equicontinuity. Econometrica 62,
43-72. 

\bibitem{key-5}Banerjee, A., D. Dunson and S. Todkar (2008) Efficient
Gaussian Process Regression for Large Data Sets. Biometrika 94, 1-16. 

\bibitem{key-5}Belloni, A., V. Chernozhukov, I. Hansen (2017) Program
Evaluation and Causal Inference With High-Dimensional Data. Econometrica
85, 233-298. 

\bibitem{key-6}Buja, A., T. Hastie and R. Tibshirani (1989) Linear
Smoothers and Additive Models (with discussion). Annals of Statistics
17, 453-555.

\bibitem{key-13}Chen, X. and Z. Liao (2014) Sieve M Inference on
Irregular Parameters. Journal of Econometrics 182, 70-86. 

\bibitem{key-3}Christmann, A. and I. Steinwart (2007) Consistency
and Robustness of Kernel-Based Regression in Convex Risk Minimization.
Bernoulli 13, 799-719.

\bibitem{key-2}Christmann, A. and R. Hable (2012) Consistency of
Support Vector Machines Using Additive Kernels for Additive Models.
Computational Statistics and Data Analysis 56, 854-873.

\bibitem{key-1}Connor, G., M. Hagmann and O. Linton (2012) Efficient
Semiparametric Estimation of the Fama-French Model and Extensions.
Econometrica 80, 713-754. 

\bibitem{key-7}Fan, J., C.M. Zhang, and J. Zhang (2001) Generalized
Likelihood Ratio Statistics and Wilks Phenomenon. The Annals of Statistics
29, 153-193.

\bibitem{key-8}Fan, J. and J. Jiang (2005) Nonparametric Inferences
for Additive Models. Journal of the American Statistical Association
100, 890-907.

\bibitem{key-8}Friedman, J. (2001) Greedy Function Approximation:
A Gradient Boosting Machine. Annals of Statistics 29, 1189-1232.

\bibitem{key-1}Geyer, C.J. (1994) On the Asymptotics of Constrained
$M$-Estimation. Annals of Statistics 22, 1993-2010.

\bibitem{key-6}Hable, R. (2012) Asymptotic Normality of Support Vector
Machine Variants and Other Regularized Kernel Methods. Journal of
Multivariate Analysis 106, 92-117.

\bibitem{key-1} H\"{a}rdle, W. and E. Mammen (1993) Comparing Nonparametric
Versus Parametric Regression Fits. Annals of Statistics 21, 1926-1947.

\bibitem{key-14}Hartigan, J.A. (2014) Bounding the Maximum of Dependent
Random Variables. Electronic Journal of Statistics 8, 3126-3140.

\bibitem{key-47}L\'{a}zaro-Gredilla, M., J. Qui\~{n}onero-Candela,
C.E. Rasmussen and A.R. Figueiras-Vidal (2010) Sparse Spectrum Gaussian
Process Regression. Journal of Machine Learning Research 11, 1865-1881.

\bibitem{key-2}Linton, O.B. and J.P. Nielsen (1995) A Kernel Method
of Estimating Structured Nonparametric Regression Based on Marginal
Integration. 

\bibitem{key-9}Lyons, R.K. (1997) A Simultaneous Trade Model of the
Foreign Exchange Hot Potato. Journal of International Economics 42
275-298.

\bibitem{key-3-3}Mammen, E., O.B. Linton, and J.P. Nielsen (1999)
The Existence and Asymptotic Properties of a Backfitting Projection
Algorithm under Weak Conditions. Annals of Statistics 27, 1443-1490.

\bibitem{key-7}Meier, L., P. B\"{u}hlmann and S. van de Geer (2009).
High-Dimensional Additive Modeling. Annals of Statistics 37, 3779-3821.

\bibitem{key-4}Mendelson, S. (2002) Geometric Parameters of Kernel
Machines. In: Kivinen J., Sloan R.H. (eds) Computational Learning
Theory (COLT). Lecture Notes in Computer Science 2375. Berlin: Springer.

\bibitem{key-7}Rasmussen, C. and C.K.I. Williams (2006) Gaussian
Processes of Machine Learning. Cambridge, MA: MIT Press.

\bibitem{key-2}Ritter, K., G.W. Wasilkowski and H. Wozniakowski (1995)
Multivariate Integration and Approximation for Random Fields Satisfying
Sacks-Ylvisaker Conditions. Annals of Applied Probability 5, 518-540.

\bibitem{key-25-1}Sancetta A. (2016) Greedy Algorithms for Prediction.
Bernoulli 22, 1227-1277.

\bibitem{key-1}Sch\"{o}lkopf, B., R. Herbrich, and A.J. Smola (2001)
A Generalized Representer Theorem. In D. Helmbold and B. Williamson
(eds.) Neural Networks and Computational Learning Theory 81, 416-426.
Berlin: Springer.

\bibitem{key-9}Shen, X. and J. Shi (2005) Sieve Likelihood Ratio
Inference on General Parameter Space. Science in China Series A: Mathematics
48, 67-78.

\bibitem{key-3-2}van der Vaart, A. (1998) Asymptotic Statistics.
Cambridge: Cambridge University Press.

\bibitem{key-2-1}van der Vaart, A. and J.A. Wellner (2000) Weak Convergence
and Empirical Processes. New York: Springer.

\bibitem{key-1}Wahba, G. (1990) Spline Models for Observational Data.
Philadelphia: SIAM.
\end{thebibliography}

\begin{thebibliography}{10}
\bibitem{key-1-1}Adler, R.J. and J.E. Taylor (2007) Random Fields
and Geometry. New York: Springer.

\bibitem{key-1-2}Andrews, D.W.K. (1984) Nonstrong Mixing Autoregressive
Processes. Journal of Applied Probability 21, 930-934. 

\bibitem{key-3-1}Bathia, R. (1997) Matrix Analysis. New York: Springer.

\bibitem{key-1-3}Basrak, B., R.A. Davis and T. Mikosch (2002) Regular
Variation of GARCH Processes, Stochastic Processes and their Applications
99, 95-115.

\bibitem{key-14}Bosq, D. (2000) Liner Processes in Function Spaces.
Berlin: Springer.

\bibitem{key-15}Bradley, R.C. (1986) Basic Properties of Strong Mixing
Conditions. In E. Eberlein and M.S. Taqqu (eds.), Dependence in Probability
and Statistics, 165-192. Boston: Birkhauser.

\bibitem{key-12}Doukhan, P. (1994) Mixing: Properties and Examples.
New York: Springer.

\bibitem{key-23}Doukhan, P., P. Massart and E. Rio (1995) Invariance
Principles for Absolutely Regular Empirical Processes

\bibitem{key-1-5}Jaggi, M. (2013) Revisiting Frank-Wolfe: Projection-Free
Sparse Convex Optimization. Journal of Machine Learning Research (Proceedings
ICML 2013). URL: <http://jmlr.org/proceedings/papers/v28/jaggi13-supp.pdf>.

\bibitem{key-3-4}Li, W.V. and W. Linde (1999) Approximation, Metric
Entropy and Small Ball Estimates for Gaussian Measures. Annals of
Probability 27, 1556-1578.

\bibitem{key-1-7}Meinshausen, N. and P. B\"{u}hlmann (2010) Journal
of the Royal Statistical Society B 72, 417-473.

\bibitem{key-6-2}Mokkadem, A. (1988) Mixing Properties of ARMA Processes.
Stochastic Processes and their Applications 29, 309-315.

\bibitem{key-1-6}Rosenblatt, M. (1980) Linear Processes and Bispectra.
Journal of Applied Probability 17, 265-270.

\bibitem{key-2-2}Steinwart, I., and A. Christmann (2008) Support
Vector Machines. Berlin: Springer. 

\bibitem{key-44}van der Vaart, A. and J.H. van Zanten (2008) Reproducing
Kernel Hilbert Spaces of Gaussian Priors. Pushing the Limits of Contemporary
Statistics: Contributions in Honor of Jayanta K. Ghosh, 200-222. Beachwood,
Ohio: Institute of Mathematical Statistics. 

\bibitem{key-2-1-1}van der Vaart, A. and J.A. Wellner (2000) Weak
Convergence and Empirical Processes. New York: Springer.

\bibitem{key-1-9}Wahba, G. (1990) Spline Models for Observational
Data. Philadelphia: SIAM.
\end{thebibliography}
\end{document}